\documentclass[12pt]{article}

\usepackage[margin=1in]{geometry}
\usepackage{amsmath, amssymb, amsthm}
\usepackage{mathtools}
\usepackage{hyperref}
\usepackage[round,authoryear]{natbib}
\bibpunct{(}{)}{;}{a}{}{,}
\usepackage{enumitem}
\usepackage{booktabs}
\usepackage{array}
\usepackage{tikz}
\usetikzlibrary{arrows.meta,positioning,calc,decorations.pathreplacing}

\hypersetup{
	 unicode=true,
  colorlinks=true,
  linkcolor=blue,
  citecolor=blue,
  urlcolor=blue
}

\usepackage{macros_prx_econ}

\newenvironment{crusoe}[1][]{%
 \par\medskip
 \noindent\textbf{Running example: Crusoe\ifx\\#1\\\else\ --- #1\fi.}%
 \par\nobreak\smallskip
}{%
 \par\medskip
}

\title{{\Large A Formalization of Austrian Economics}\\[6pt]
\small\emph{Praxeological Foundations:\\
The Base System and Its Derived Theorems}}
\author{%
  Rafa{\l} Komendarczyk\thanks{Tulane University, Department of Mathematics. Corresponding author: \texttt{rako@tulane.edu}.}%
  ,\ Walter E. Block\thanks{Loyola University New Orleans, J.A. Butt College of Business. \texttt{wblock@loyno.edu}.}%
  ,\ John Levendis\thanks{Tulane University, Connolly Alexander Institute for Data Science (CAIDS). \texttt{jlevendis@tulane.edu}.}\\
  and Frank J. Tipler\thanks{Tulane University, Department of Mathematics. \texttt{tipler@tulane.edu}.}%
}
\date{\today}

\begin{document}
\maketitle

\begin{abstract}
This paper presents an axiomatization of Ludwig von Mises' praxeology in many-sorted first-order logic, 
isolating the foundational layer. We introduce a formal language with five sorts ({\sf Actors}, {\sf Actions},
 {\sf Ends}, {\sf Things}, {\sf Times}) 
and six primitive relations ({\em Acts}, {\em Avail}, {\em EndOf},
{\em Use}, a preference order, and a time order), together
with a base axiom system organized into three layers: the structure
of action itself, the actor's preference order together with its
demonstration in choice, and material scarcity. The base system captures purposeful action in its bare praxeological form.

Working entirely within the base system we derive the core classical Misesian propositions as Hilbert-style
theorems: the asymmetry of demonstrated preference, the existence of opportunity cost, the structural
scarcity of time, the subjectivity of opportunity cost, the law of diminishing marginal utility, 
and the increasing marginal disutility of labor. Where a theorem requires structure beyond 
the praxeological core -- as with diminishing marginal utility -- the additional premises are made explicit; 
identifying these hidden premises is one of the methodological payoffs of the approach.

A self-contained \Lean{} companion encodes the language as \Lean{}
type classes and constructs a concrete infinite-time Robinson Crusoe
model whose acceptance by the type-checker is a constructive
consistency proof of the full base theory.
\end{abstract}

\newpage

\tableofcontents

\section{Introduction}
\label{sec:intro}

\subsection{Mises' programme and the role of formalization}
\label{subsec:mises_programme}

In \emph{Human Action} Ludwig von Mises set out to ground
economic theory in a small number of propositions about purposeful
action --- propositions whose evidence, he held, is implicit in
the very fact that human beings act at all.  The resulting
discipline, \emph{praxeology}, was advanced with little use of mathematical notation.  Critics
of the school sometimes treat the absence of equations as a
feature of the substance of praxeology rather than as a
historical accident of its presentation.  As we observe in Section~\ref{subsec:historical} below, this
perception is mistaken: nothing in the foundation of praxeology
militates against formal analysis.\footnote{This is also the
methodological position of \citet{Hudik2014}, developed in
Section~\ref{subsec:historical}; the present paper is a
constructive realisation of that diagnosis at the foundational
layer.}  What it does militate against
is a particular interpretive use of formal analysis ---
\emph{testing} synthetic-a-priori propositions by statistical
means.  Praxeological propositions are not refutable by
econometric tests, but they are perfectly susceptible to
mathematical formulation.  Rothbard's \emph{Man, Economy, and
State}~\citep{Rothbard2009} is the most extensive systematic
working-out of this programme; the present formalization draws
on both Mises's verbal axiomatics and Rothbard's illustrative
reconstructions of it, with citations at the relevant theorems.

The present paper carries out one such formulation.
We axiomatize Mises' praxeology in many-sorted first-order
logic \citep{Enderton2001}.  The base theory $\Tprx$\footnote{The
subscript ``$\mathrm{prx}$'' abbreviates \emph{praxeology}:
$\Tprx$ is the base praxeological theory, and $\Lprx$
(Section~\ref{sec:language}) is the formal language it is
written in.} captures the structure of
purposeful action. 
Within
this axiomatization,  classical Misesian propositions can be
derived as formal theorems, and the question whether the
policy-maker's allocation problem admits an algorithmic solution
becomes a question of computability theory.

The motivation is not to replace verbal praxeology but to
\emph{verify} its propositions on the grounds of first-order
logic.  The Hilbert-style format used throughout
makes every inference explicit: every step is either an instance
of an axiom, an application of a definition, or a logical
deduction from previously established statements.  Where a
result requires premises beyond the praxeological core --- as
with diminishing marginal utility --- this becomes visible in
the proof itself.
The exposition therefore doubles as a diagnostic of the
\emph{premise structure} of standard Austrian arguments.

Mises is emphatic that ``Human action is necessarily always
rational; the term `rational action' is therefore pleonastic''
\citep[Ch.~I, \S4]{Mises1949}.\footnote{``Rational'' here is in
the formal sense of \axref{O0}--\axref{O4} (a strict total order
on choice-relevant ends, grounded in choice) and the
(MU)-rationality of
allocation; it does not import the colloquial connotation of
cool deliberation.  \citet[ch.~1,
pp.~6--7]{Rothbard2009} preempts that misreading directly:
``The individual may make a decision to act hastily, or after
great deliberation\ldots\ none of these courses affects the
fact that action is being taken.''}  In Mises's verbal-deductive
tradition the substantive structure is packed into the single
word \emph{rational}, and the reader unpacks it from context.
That structure has four pieces --- a strict ranking of ends,
homogeneity of units, scarce means, and a schedule of allotments
respecting the ranking.
The argument is verbally precise.

It is not, however, mathematically precise, and that is the
gap the present formalization is designed to close.  Written
symbolically, every implicit piece of ``rational'' must become a
stated premise, or the inference is invalid: ``Crusoe rationally
allocates his five buckets of water'' carries the verbal
argument, but the formal claim that an additional unit serves a
less urgent end follows only once those four pieces appear as
separately stated axioms.

The body of this paper is therefore an exercise in
decomposition: ``rational'' in $\Lprx$ is not a single primitive
but a structure distributed across axiomatic layers.  The
structural background is given by
\axref{P1}--\axref{P6} and \axref{T1}--\axref{T6} (\emph{action and time
structure}).  The substantive rationality content is then
unpacked into distinct facets: the action choice axiom \axref{C1}
(\emph{rationality of choice}); the preference axioms
\axref{O0}--\axref{O4} --- grounding in choice plus a strict
total order on the ends at each time
(\emph{rationality of preference}); the existence-of-scarcity
axiom \axref{S1} (\emph{rationality of scarcity recognition});
and, where required for downstream theorems, further
enrichments such as \axref{MU0}--\axref{MU4} for diminishing
marginal utility (\emph{rationality of allocation}) or
\axref{LB1} for the disutility of labor (\emph{rationality of
labor}).  Each is a labeled piece of what Mises packed into
the one word \emph{rational}.

The diagnostic value of the formalization is therefore that the labeled pieces
are now more clearly visible.  This is the
same move Hilbert made for Euclid: Greek geometry was obviously
correct; the axiomatisation made the structure of ``obviously''
mathematically tractable.

Mises grounds praxeology in a single irreducible proposition:
\emph{human action is purposeful behavior} ---
the Action Axiom of \emph{Human Action} Ch.~I~\S1
\citep[p.~11]{Mises1949}.  The axiom is presented as
apodictically certain: any attempt to deny it is itself an
act, and therefore an instance of what is being denied.  In
Mises' analysis, action requires three prerequisites in the
actor's situation:
\begin{enumerate}[label=(\roman*), itemsep=2pt]
\item \emph{felt uneasiness} --- the actor perceives his
present state as suboptimal;
\item \emph{an imagined better state} --- the actor conceives
of a state of affairs preferable to the present one; and
\item \emph{the expectation of efficacy} --- the actor expects
that purposeful behavior has the power to bring the better
state about, or at least to alleviate the uneasiness.
\end{enumerate}
If any of the three is absent, the actor remains at rest:
``a man perfectly content with the state of his affairs would
have no incentive to change things'' \citep[Ch.~I, \S2, pp.~13--14]{Mises1949}.
We do not promote prerequisites (i)--(iii) to primitive
predicates: doing so would not affect any theorem here, so they
are recorded only verbally.

\subsection{Scope and roadmap}
\label{subsec:scope}

\noindent
The rest of this introduction sets the stage:
Section~\ref{subsec:historical} gives the historical and
methodological context --- why mathematics has been sparse in
Austrian economics, what changes when praxeology is restated in
first-order logic, and what the paper offers a mathematician
unfamiliar with praxeology --- and
Section~\ref{subsec:objections} anticipates two objections to
the base system before the formal development begins.

\medskip
\noindent
The body then proceeds as follows.

\begin{itemize}[itemsep=3pt]
\item Section~\ref{sec:language} introduces the formal language
$\Lprx$ --- five sorts and six primitive relations --- and the
base axiom system $\Tprx$ in three layers: action structure
(T1--T6, P1--P6, C1), preference and its demonstration in choice
(O0--O4), and material scarcity \axref{S1}.
\item Section~\ref{subsec:crusoe} works the apparatus through a
Robinson Crusoe model over $\mathbb{N}$-time, with a folded
state diagram that lets time run for as many rounds as needed.
\item Section~\ref{sec:derived_base} gives Hilbert-style proofs
of the core Misesian theorems derivable in $\Tprx$ alone (or in
$\Tprx$ plus explicit auxiliary premises), in three subsections:
the praxeological core (Subsection~\ref{subsec:derived_core}:
the asymmetry of demonstrated preference, Theorem~\ref{thm:asymm};
the existence of opportunity cost, Theorem~\ref{thm:opp_cost};
the structural scarcity of time, Theorem~\ref{thm:time_scarcity};
and the subjectivity of opportunity cost,
Theorem~\ref{thm:subjectivity}); diminishing marginal utility
with the five hidden premises \axref{MU0}--\axref{MU4} made
explicit (Subsection~\ref{subsec:derived_dmu}); and the
increasing marginal disutility of labor
(Subsection~\ref{subsec:derived_labor}).
\item Section~\ref{sec:outlook} sketches the directions for
future research that extend the present paper.
\item Two appendices support the development: a \Lean{} companion
(Appendix~\ref{sec:lean}) that encodes the axiom system as \Lean{}
type classes, proves the foundational theorems inside \Lean{}, and
constructs two concrete instances --- a finite three-period
snapshot of the Crusoe model of Section~\ref{subsec:crusoe} and
the full $\mathbb{N}$-time model itself --- whose acceptance by
the type-checker is a constructive consistency proof of the full
base theory $\Tprx$; and Appendix~\ref{app:dmu_proof}, the full
proof of Theorem~\ref{thm:DMU} (DMU) sketched in
Subsection~\ref{subsec:derived_dmu}.
\end{itemize}

\noindent
Throughout, axioms and theorems are followed by remarks mapping
the formalism to the original Misesian concepts.

\medskip
\noindent
What this paper does \emph{not} contain --- production,
ownership, exchange, money, and the transition map; the
policy-maker, the Mises Inexpressibility theorem, and the
undecidability of the global policy problem --- is sketched
as future research in Section~\ref{sec:outlook}.

\subsection{Formalization and Austrian methodology}
\label{subsec:historical}

The figures most closely identified with the Austrian school
---
\citet{Mises1949}, \citet{Menger1871},
\citet{BohmBawerk1889}, \citet{Hayek1945},
and \citet{Rothbard2009} ---
wrote almost exclusively in prose; their successors have
produced an extensive secondary literature defending that
choice on what are advanced as principled rather than
stylistic grounds.  The objections range from the classical
methodological
\citep{Mises1977,Rothbard1960,Rothbard1976,Hazlitt1959,LeoniFrola1977}
to the contemporary
\citep{Herbener1996,BarnettBlock2010,MurphyWutscherBlock2010}.
The common thread is that mathematical formulation imports
assumptions foreign to the praxeological core --- continuity,
indifference, cardinality, optimisation under unrestricted
quantification over states of the world --- and that an
Austrian theorem rewritten in symbols ceases, in some
substantive sense, to be Austrian.  That Murray N.\ Rothbard,
himself a capable mathematician, nonetheless cast \emph{Man,
Economy, and State} entirely in prose \citep{Rothbard2009} is
taken to show that the absence of equations from the canonical
literature reflects something more than historical accident.

This view has been challenged by three contemporary Austrian
authors, spanning the methodological spectrum.

\medskip
\noindent\textbf{Hud\'{i}k.}\;
\citet{Hudik2012,Hudik2014} treats the verbal style
of Austrian exposition as a contingent feature of the
school's presentation rather than as a constitutive
commitment of its method.  On his account mathematics is, in
the relevant respect, a language: more precise than English,
and the common language of mainstream economics.  Translating
an Austrian argument from prose into symbols is therefore
comparable to translating a book from English into French
--- the argument is preserved, while the medium changes.
Austrians who reject mathematics, on this view, are defending
a writing convention while paying real isolation costs in the
broader profession.  Citing \citet{Moorhouse1993},
Hud\'{i}k concludes that there is ``no major methodological
gulf between praxeology and neo-classical mathematical
economics,'' and recommends that Austrians write in symbols
when symbols are the clearer choice.

\medskip
\noindent\textbf{Linsbichler.}\;
\citet{Linsbichler2023} offers a
measured defense of formalization under a translation
paradigm.  He grants that translation is not always neutral:
a poor translation can introduce assumptions foreign to the
original, and a paper full of equations is not automatically
clearer than one written in words.  Some Austrian worries
about over-mathematisation are accordingly well founded.
Linsbichler denies, however, that there is a
\emph{principled} Austrian objection to formal methods: such
methods are tools, usable well or badly, and incompatibility
with Mises' programme is not entailed by their adoption.  His
practical recommendation is that Austrian economists who
already employ formal methods should stop apologising for the
practice, and that the school should make peace with the fact
that some Austrian work will be written in symbols.

\medskip
\noindent\textbf{Nguyen.}\;
\citet{Nguyen2026} denies the linguistic premise that
both Hud\'{i}k and Linsbichler share.  Mathematics is not, in
Nguyen's view, a language like English or French; it is
something different in kind.  The argument runs through the
mid-twentieth-century debate between Rudolf Carnap and Kurt
G\"odel.  Carnap held that mathematics is the syntax of a
language --- a system of rules for manipulating symbols, with
mathematical truth identified with the output of those rules.
G\"odel's incompleteness theorems show that any formal system
rich enough for elementary arithmetic contains true statements
that the system cannot prove; mathematical truth therefore
exceeds syntactic derivability, and any rendering of
mathematics as ``the syntax of a language'' is incomplete on
its own terms.  Applied to economics, this means that
rewriting an Austrian argument in symbols is doing something
more than translating: it commits the user to assumptions not
present in the verbal original, and reimports the
truth-versus-provability gap that G\"odel identified.  The
translation defense of Austrian formalization accordingly
rests, on this account, on a mistaken philosophy of
mathematics.  Nguyen's own conclusion is not anti-formalist:
he advocates that Austrians employ mathematical logic and the
theory of part-whole relations, but on grounds independent of
the translation thesis, and recommends giving up the strongest
form of the synthetic-a-priori doctrine of praxeology.

\medskip
\noindent\textbf{The present paper's position.}\;
The three positions agree on a single practical
recommendation --- that Austrian work in symbols is worth
producing --- while disagreeing on the philosophical grounds.
Our view, taken in this paper, is to adopt Linsbichler's
diagnosis as the operating assumption.  Mises' praxeological programme is axiomatised in
many-sorted first-order logic, and the mathematical content
is advanced not as a translation in Hud\'{i}k's sense (which
Nguyen has shown to be philosophically under-supported) but
as a \emph{Hilbert-style verification} of the underlying
verbal arguments: every premise is written, every step is
named, and the derivation of any Misesian theorem reduces to
a finite application of the inference rules of first-order
logic.  Where the formalization imports substantive
assumptions beyond Mises' verbal programme, those assumptions
appear in the axiom list and can be inspected, debated, and
rejected by the reader; the diagnostic value of the
formalization, as we observed in
Section~\ref{subsec:mises_programme}, lies precisely in this
visibility.

The classical Austrian objection that formalization imports
econometric methodology --- that mathematical analysis
inevitably entails \emph{testing} synthetic-a-priori
propositions by statistical means --- is left intact.  When
mainstream economists perform statistical work on subjects
that admit of synthetic-a-priori reasoning, such as rent
control or the minimum wage, they typically claim to be
\emph{testing} the basic economic laws involved; if the
statistical findings do not bear the laws out, the laws are
held to be falsified.  Austrians maintain the opposite: an
econometric analysis that fails to detect unemployment among
workers whose productivity falls below a mandated minimum
wage shows that the analysis is flawed, not that the law of
demand for labor has been refuted.  What is known on the
basis of the logic of the situation remains valid; the
legitimate function of empirical work in this domain is to
\emph{illustrate}, not to test, the foundations of
economics.  The present formalization likewise verifies,
rather than tests, the propositions of praxeology, and
nothing in the apparatus assigns a refuting role to
statistical evidence.

A further benefit of the formalization accrues to readers
trained in mathematical logic.  For the first time, the
praxeological foundations of Austrian economics can be
introduced to that audience in the language of first-order
logic, with a stand-alone consistency proof available in the
form of the \Lean{} kernel's acceptance of the Crusoe model
(Appendix~\ref{sec:lean}).

\subsection{Anticipating two objections to the base system}
\label{subsec:objections}

The base system is deliberately thin, and a reader may take its
thinness for a defect.  Two worries are worth meeting head-on.

\medskip
\noindent\emph{Is this merely neoclassical revealed-preference
theory in discrete dress?}  It is not.  The preference primitive is
purely ordinal: \axref{O2}--\axref{O4} carry no cardinal measure and
no indifference relation.  And the order carries a separate time
index, with no axiom tying $\pref{a}{t}$ to $\pref{a}{s}$ for distinct
times: the base system imposes no \emph{global} utility
$u:\mathsf{Ends}\to\mathbb{R}$ --- a single scale holding at all times.  The constancy of preference
across observations that the Samuelsonian apparatus requires is
exactly the assumption the base system refuses: what choice puts on
record is Rothbard's \emph{demonstrated} preference, not the
constancy-laden \emph{revealed} preference the recovery program needs
--- the point developed, with its Rothbardian pedigree, at
Definition~\ref{def:revpref} and Remark~\ref{rem:constancy}.

\medskip
\noindent\emph{Does reasoning about roads not taken, and about
plans aimed at the future, need a richer logic than plain
first-order logic?}  At the base layer it does not.  Opportunity
cost is the end of an available action that was not performed
(Theorem~\ref{thm:opp_cost}), and availability is already a
primitive relation, so the option foregone is named directly, with
the relations the system already has --- no extra logical
machinery for possibility or tense (modal or temporal logic) is
required.  Uncertainty, expectation, and the gap
between an intended and a realized end are real phenomena; they are
deliberately left out of the base layer
(Remarks~\ref{rem:endof_intended} and~\ref{rem:causality}) and
belong to the production and profit-and-loss enrichments sketched
in Section~\ref{sec:outlook}.

\medskip
\noindent A third, positive point has already been made in
\S\ref{subsec:mises_programme}: splitting ``rational'' into
separately named axioms, whose roles are tabulated
(Remark~\ref{rem:dmu_method}), means a failed derivation points to
the missing assumption rather than hiding it inside a word --- the
diagnostic payoff of the approach.

\section{The Formal Language \texorpdfstring{$\Lprx$}{Lprx}}
\label{sec:language}

\subsection{Sorts and primitives}
\label{subsec:primitives}

We work in classical many-sorted first-order logic with equality \citep{Enderton2001,Hodges1997}.\footnote{Readers
more at home with set theory than with the logician's vocabulary
will find the apparatus entirely familiar: a \emph{sort} is a set,
a \emph{primitive relation} a subset of a Cartesian product of
those sets, a \emph{structure} a choice of both, and
``$S \models \varphi$'' the statement that $\varphi$ is true of
that structure --- nothing added, nothing lost.
Appendix~\ref{app:fol} gives a self-contained primer and glossary
for the full vocabulary.} The language
$\Lprx$ has five sorts and six primitive relations.  We use Latin
letters ($a, b, p, \ldots$) for actors and Greek letters ($\alpha,
\beta, \gamma, \ldots$) for actions throughout, so that the actor
$a$ and the action $\alpha$ are visually distinct even though both
appear in the same expression.

\medskip
\begin{center}
\begin{tabular}{lll}
\toprule
\textbf{Sort} & \textbf{Variables} & \textbf{Meaning} \\
\midrule
$\mathsf{Actors}$  & $a, b, p, \ldots$ \emph{(Latin)}       & Agents who act purposefully \\
$\mathsf{Actions}$ & $\alpha, \beta, \ldots$ \emph{(Greek)} & Purposeful acts \\
$\mathsf{Ends}$    & $E, F, G, \ldots$       & Goals toward which actions are directed\\
        &                         & (atomic or composite; see Rem.~\ref{rem:composite_ends}) \\
$\mathsf{Things}$  & $x, y, z, \ldots$       & Objects in the given universe\\
        &                         & (become means only via $\Use$;
                                    cf.\ Rem.~\ref{rem:means_via_use}) \\
$\mathsf{Times}$   & $t, s, r, \ldots$       & Moments at which actions occur \\
\bottomrule
\end{tabular}
\end{center}

The signature has six primitive relations:
\begin{align*}
&\Acts(a,\alpha,t) &&\text{``actor $a$ performs action $\alpha$ at time $t$''}\\
&\Avail(a,\alpha,t) &&\text{``action $\alpha$ is available to $a$ at time $t$''}\\
&\EndOf(\alpha,E) &&\text{``action $\alpha$ is directed toward end $E$''}\\
&\Use(\alpha,x) &&\text{``action $\alpha$ employs thing $x$''}\\
&E \pref{a}{t} F &&\text{``actor $a$ at time $t$ ranks end $E$ above end $F$''}\\
&t < s &&\text{``time $t$ strictly precedes time $s$''}
\end{align*}

These primitives translate directly from
Mises' account of action \citep{Mises1949}. Actors are purposeful agents. Actions
are the central category: every action is directed toward an end (goal) and
employs means (things). The time ordering is irreversible --- a fundamental
feature of the praxeological system distinguishing it from atemporal logic.
The preference primitive $\pref{a}{t}$ is the actor's \emph{scale of
values} --- the valuation component of Mises's state vector.
It is primitive in the
\emph{language} but hidden to \emph{observation}: the only access to
it is through choice, as Mises insists --- ``the scale of values or
wants manifests itself \emph{only} in the reality
of action'' \citep[p.~95]{Mises1949}.  Layer~2 of
Section~\ref{subsec:axioms} makes both halves formal: a
\emph{demonstrated-preference record} (Rothbard) derived from $\Acts$
(Definition~\ref{def:revpref}) and a grounding axiom \axref{O0}
requiring the primitive order to agree with that record.

\medskip
The next four remarks record what the primitives \emph{intend} but do
not formally encode --- meta-level conventions about the intended
interpretation of each sort and relation.  None of them adds an axiom
to $\Tprx$; each pins down semantic content that future enrichments
might internalise.

\begin{remark}[The Actions sort\footnote{i.e. the set of possible actions, \citep{Enderton2001}} excludes reflexive behavior]
\label{rem:actions_purposeful}
The Actions sort is intended to range only over \emph{conscious,
purposeful} action.  Reflexes, autonomic nervous-system responses,
and other unconscious bodily reactions are deliberately outside the
sort, even though an external observer might describe them as
``things the body does''.  Mises is explicit:
``conscious or purposeful behavior is in sharp contrast to
unconscious behavior, i.e., the reflexes and the involuntary
responses of the body's cells and nerves to stimuli''
\citep[Ch.~I, \S1]{Mises1949}.  This is not a formal axiom of
$\Tprx$ --- the language has no predicate that distinguishes
conscious from unconscious behavior --- but a constraint on the
intended populations of the sort when specifying a model.
\end{remark}

\begin{remark}[$\EndOf$ records intended, not realized, ends]
\label{rem:endof_intended}
The relation $\EndOf(\alpha, E)$ records the end the actor
\emph{aims at} when performing $\alpha$, not an objective
guarantee that $\alpha$ in fact produces $E$.  Mises is again
explicit: ``a means is every thing which acting men consider as
such'' \citep[Ch.~IV, \S1]{Mises1949}.  The praxeological framework
treats man as ``weak and subject to error''; the gap between an
action's intended end and its realized outcome is the formal home
of profit and loss, but at the base layer of $\Tprx$ this gap is
silent: $\Acts(a, \alpha, t)$ records that $a$ performed $\alpha$
intending the end of $\alpha$, not that the intention succeeded.
The fallibility apparatus that activates this gap is introduced in
future work on extensions as a small enrichment.
\citet[ch.~1, p.~7]{Rothbard2009} ties the gap directly
to uncertainty: ``the omnipresence of uncertainty introduces the
ever-present possibility of error in human action.  The actor may
find, after he has completed his action, that the means have been
inappropriate to the attainment of his end.''  A profit/loss apparatus
in future research (Section~\ref{sec:outlook}) formalizes the consequences.
\end{remark}

\begin{remark}[Actions are individuated by meaning]
\label{rem:action_individuation}
$\EndOf(\alpha, E)$ and $\Use(\alpha, x)$ carry no actor or time
argument: an action has the same end and the same inputs wherever
it occurs.  This is not the claim that one \emph{bodily movement}
serves a single end for all actors at all times --- walking can be
exercise, commuting, or escape --- but a consequence of how the
Actions sort is individuated.  Actions in $\Lprx$ are individuated
\emph{by their meaning}: walking-as-exercise and
walking-as-commuting aim at different ends and are therefore
\emph{distinct elements} of the sort.  This is Mises's own action
concept --- an action is defined by the meaning the actor gives it,
not by the physiology \citep[Ch.~I, \S1]{Mises1949} --- and it is
what makes the global $\EndOf(\alpha, E)$ harmless: a different
intention is already a different action.  The actor- and
time-indexed alternative $\EndOf(a, \alpha, t, E)$,
$\Use(a, \alpha, t, x)$ is recovered by taking the Actions sort to
consist of fully individuated tokens (actor--movement--time
triples); we keep the lighter signature at the base layer, where
meaning-individuation makes the extra indices redundant, and add
finer indexing only where an enrichment needs it.
\end{remark}

\begin{remark}[``Things'' become ``means'' only via $\Use$]
\label{rem:means_via_use}
$\Lprx$ has a sort for Things but no sort for \emph{means}.
The reason is that \emph{being-a-means is not an intrinsic
property of any thing in isolation}: it is a relational fact
that holds between an action and a thing.  Mises is explicit:
``Means are not in the given universe; in this universe there
exist only things.  A thing becomes a means when human reason
plans to employ it for the attainment of some end''
\citep[Ch.~IV, \S1]{Mises1949}.  Promoting Means to a primitive
sort would force this relational, action-dependent property into
a unary predicate on its own, contrary to Mises's analysis.

\smallskip
The relational framing of $\Use(\alpha, x)$ captures the
correct dependency directly: a thing $x$ enters the
praxeological apparatus as a \emph{means} exactly when there is
an action $\alpha$ with $\Use(\alpha, x)$ --- i.e., when some
action employs it.  Different actions make different things
into means; the same thing serves as a means under one action
and as inert material under another.

\smallskip
The same relational framing absorbs Mises's notion of
\emph{services} without adding apparatus: a service is, in our
terms, the use one makes of a thing \citep[Ch.~VII, \S1]{Mises1949}.  In
$\Lprx$, services do not require a separate sort; the relation
$\Use(\alpha, x)$ \emph{is} the service relation, recording the
particular use that action $\alpha$ makes of thing $x$.
\end{remark}

\begin{remark}[Causality is presupposed, not formalized]
\label{rem:causality}
Every instance of $\Acts(a, \alpha, t)$ silently carries a
commitment by actor $a$ that action $\alpha$ stands in some
causal relation
to its end --- otherwise $a$ would not act on it.  Mises insists
on this: ``acting requires and presupposes the category of
causality'' and ``where man does not see any causal relation, he
cannot act'' \citep[Ch.~I, \S5]{Mises1949}.  The base apparatus
$\Tprx$ does not contain a causality predicate or a belief-state
relation; the causal hypothesis is hidden inside the actor's
adoption of a particular $(\Use, \EndOf)$ pattern.  Making this
hypothesis explicit would require an epistemic enrichment
beyond the scope of the present paper.
\end{remark}

\subsubsection*{The choice menu}
\label{subsec:choice_menu}

The availability relation $\Avail$ determines, at every
$(a, t)$, the set of actions the actor is in a position to
perform.  This set --- the actor's \emph{choice menu} at $t$ ---
is the central object on which every later definition and
theorem in the paper hinges, and we fix notation for it now:
\begin{equation}\label{eq:choice_menu}
  \Cmenu_{a,t} \;:=\; \{\alpha \in \mathsf{Actions}
  : \Avail(a, \alpha, t)\}.
\end{equation}
We refer throughout to $\Cmenu_{a,t}$ as the choice menu of
actor $a$ at time $t$.

The choice menu
$\Cmenu_{a,t}$ of \emph{actions} should not be confused with
the end menu $\EndAt(a,t)$ of \emph{ends} introduced at
\eqref{eq:choice-menu-E_at} in
Section~\ref{subsec:axioms} ($\EndAt(a,t)$ is the image
$\EndOf(\Cmenu_{a,t})$ of the action menu under
$\EndOf$, restricted to choice-relevant ends).

\subsubsection*{Structures and models}

\begin{definition}[Structure and model]\label{def:model}
	A \emph{structure} for $\Lprx$ is an assignment of:
	\begin{enumerate}[label=(\roman*), itemsep=2pt]
		\item a nonempty set (a \emph{domain}) for each of the five sorts
		--- Actors, Actions, Ends, Things, Times;
		\item an interpretation of each primitive relation ($\Acts$,
		$\Avail$, $\EndOf$, $\Use$, $\succ$, $<$) as a relation on the
		appropriate domains.
	\end{enumerate}
	A structure is a \emph{model of\/ $\Tprx$} if every axiom of
	$\Tprx$ is true under this interpretation\footnote{A set of axioms does not describe a single
	intended object but rather delineates a \emph{class} of structures
	that satisfy those axioms.  This is the Hilbert--Tarski perspective
	\citep{Hilbert1899,Tarski1959}: the axioms of Euclidean geometry do
	not define ``the'' plane; they characterize every structure in which
	the axioms hold.  The same perspective governs $\Lprx$.}.  We write
	$S \models \Tprx$ to mean that the structure $S$ is a model of
	$\Tprx$.\footnote{This is the standard set-theoretic semantics
	for many-sorted first-order logic; Appendix~\ref{app:fol} is a
	self-contained primer, and \citet{Hodges1997,Enderton2001} give
	full treatments.}
\end{definition}

A structure for $\Lprx$ is a
complete description of a possible economic world: who the actors
are, what actions exist, which are available to whom, what ends
they serve, and so on.  The axioms $\Tprx$ constrain which such
worlds are \emph{admissible} --- they rule out, for example,
worlds where actors perform unavailable actions \axref{P2} or pursue
actions with no end \axref{P3}.  A model of $\Tprx$ is an admissible
economic world.  Theorems of $\Tprx$ are propositions true in
\emph{every} admissible world, not just in some particular
intended economy.

\subsection{Core axioms}
\label{subsec:axioms}

Each axiom below formalizes a feature of the Misesian account of
purposeful action.  We present them in three layers: the structure of
action itself, the actor's preference order together with its
demonstration in choice, and material scarcity.

\subsubsection*{Layer 1: Action structure}

\begin{remark}[Time is a praxeological category]
\label{rem:time_praxeological}
The time-sort of $\Lprx$ is not imported from physics; it is
constituted by action.  Mises:
``the idea of time is a praxeological category''
\citep[Ch.~V, \S2]{Mises1949}, and ``it is acting that provides
man with the notion of time'' \citep[Ch.~V, \S2]{Mises1949}.  This justifies treating
time as an axiomatised primitive sort rather than as an external
parameter borrowed from a physical model: the structure imposed
on $\mathsf{Times}$ by axioms \axref{T1}--\axref{T6} below is the structure
that the categorial form of action itself impresses on it
(directedness, irreversibility, the episodic ``sooner / later''
ordering of an actor's actions).
\end{remark}

\noindent\textbf{Time-order axioms.}\footnote{The axiom list is
deliberately not minimal.  \axref{T5} follows from \axref{T1} and
\axref{T2} (if $t<s$ and $s<t$ then $t<t$ by transitivity,
contradicting irreflexivity); \axref{T4} follows from \axref{T6}
together with the non-emptiness of the time domain (any time has a
successor by \axref{T6}); and \axref{P5} below is interderivable
with the action choice axiom \axref{C1} (instantiating \axref{P5} at
$s=t$ gives \axref{C1}; conversely \axref{C1} forces $t \neq s$
whenever two distinct actions are performed by one actor).  Each
is nonetheless stated separately because each carries its own
praxeological content and its own textual anchor in Mises; we do
not claim the list is independent.}  Time in the praxeological
system is directed and irreversible --- the past cannot be
revisited, and action proceeds in sequential episodes.
\begin{enumerate}[label=(T\arabic*), leftmargin=2.5em]
\item\label{ax:T1} \textbf{Irreflexivity.} $\forall t\;\neg(t<t)$.
\item\label{ax:T2} \textbf{Transitivity.} $\forall t\,s\,r\;(t<s\land s<r\Rightarrow t<r)$.
\item\label{ax:T3} \textbf{Trichotomy.} $\forall t\,s\;(t<s\lor t=s\lor s<t)$.
\item\label{ax:T4} \textbf{Nontriviality.} $\exists t\,s\;(t<s)$.
\item\label{ax:T5} \textbf{Irreversibility.} $\forall t\,s\;(t<s\Rightarrow\neg(s<t))$.
\hfill\emph{Time is directed; the past cannot be revisited.}
\item[(T0)]\label{ax:T0} \textbf{First time.} $\exists t_0\,\forall t\;(t_0\leq t)$.
\hfill\emph{Economic history has a beginning.}
\item[(T6)]\label{ax:T6} \textbf{Discreteness.}\footnote{As with
any first-order theory of discrete linear order,
\axref{T0}--\axref{T6} do not characterize $(\mathbb{N},<)$ up to
isomorphism: nonstandard models with limit-like tails (e.g.\
$\omega + \zeta$) also satisfy them.  Nothing in this paper or its
sequels depends on categoricity --- the theorems hold in every
model, the \Lean{} consistency model interprets $\mathsf{Times}$ as
$\mathbb{N}$ (Appendix~\ref{sec:lean}), and the computability
results of future research are statements about the model class
whose example economies are constructed over standard time.}
$\forall t\;\exists s\;(t<s\land\neg\exists r\,(t<r\land r<s))$.
\hfill\emph{Each moment has an immediate successor; action occurs in episodes.}
\end{enumerate}

\noindent\textbf{Incidence axioms.}  These capture the basic structure
of purposeful action as Mises describes it: every actor has options,
can only perform what is available, and every action is directed toward
exactly one end.
\begin{enumerate}[label=(P\arabic*), leftmargin=2.5em]
\item\label{ax:P1} $\forall a\,t\;\exists\alpha\;\Avail(a,\alpha,t)$.
\hfill\emph{The choice menu $\Cmenu_{a,t}$ is nonempty over every $(a, t)$.}
\item\label{ax:P2} $\forall a\,\alpha\,t\;(\Acts(a,\alpha,t)\Rightarrow\Avail(a,\alpha,t))$.
\hfill\emph{Only available actions can be performed.}
\item\label{ax:P3} $\forall\alpha\;\exists E\;\EndOf(\alpha,E)$.
\hfill\emph{Every action is directed toward some end.}
\item\label{ax:P4} $\forall\alpha\,E\,F\;(\EndOf(\alpha,E)\land\EndOf(\alpha,F)\Rightarrow E=F)$.
\hfill\emph{Each action has exactly one end.}\footnote{Together with~\axref{P3},
this makes $\EndOf$ the graph of a total function
$\textbf{Actions} \to \textbf{Ends}$ in every model of $\Tprx$.  We declare
$\EndOf$ as a relation rather than a function symbol so that $\Lprx$ remains
a uniform relational signature.  By contrast, $\Use$ is genuinely many-valued --- a single
action may employ several things --- and the relational signature is the
only one that fits the full set of primitives.}
\item\label{ax:P5} \textbf{Nonsynchronism.} $\forall a\,\alpha\,\beta\,t\,s\;(\Acts(a,\alpha,t)\land\Acts(a,\beta,s)\land\alpha\neq\beta\Rightarrow t\neq s)$.
\hfill\emph{Distinct actions by the same actor occur at different
times.}\footnote{Mises states this directly: ``Two actions of an
individual are never synchronous; their temporal relation is that
of sooner and later'' \citep[Ch.~V, \S4]{Mises1949}.}
\item\label{ax:P6} \textbf{Free-good exclusion.} $\forall a\,t\,\alpha\,x\;(\Acts(a,\alpha,t)\land\Use(\alpha,x)\Rightarrow\exists\beta\,s\;(\Avail(a,\beta,s)\land\Use(\beta,x)\land\neg\Acts(a,\beta,s)))$.
\hfill\emph{Every act of employment leaves some employment of the
same thing available but unrealised.}
\end{enumerate}

\noindent Axiom \axref{P6} is the formal expression of Mises' insistence
that economics concerns only \emph{scarce} means: ``means are
necessarily always limited, i.e., scarce''
\citep[Ch.~IV, \S1]{Mises1949}.  \citet[ch.~1,
pp.~4--5]{Rothbard2009} sharpens the framing into a general
means/general-conditions distinction: ``the environment external
to the individual may be divided into two parts: those elements
which he believes he cannot control and must leave unchanged,
and those which he can alter\ldots\ The former may be termed the
general conditions of the action; the latter, the means
used\ldots\ Air, then, though indispensable, is not a means, but
a general condition of human action and human welfare.''  Air
is excluded by~\axref{P6} not because it does not matter but
because the actor takes it as a general condition.
A good that is available in
unlimited quantity --- a ``free good'' --- generates no economic
problem and is excluded from the analysis.  What \axref{P6}
asserts is that nothing \emph{employed in action} is in this
position: every act of employment leaves some employment of the
same thing available yet unrealised, so employment is never free
of a foregone alternative.  The precise scope of the axiom ---
and its division of labor with the existence-of-scarcity axiom
\axref{S1} of Layer~3 --- is the subject of the next remark.

\begin{remark}[Scope of \axref{P6}]
\label{rem:p6_scope}
The unrealised alternative that \axref{P6} demands is deliberately liberal:
it may be a \emph{different} action competing for $x$ at the same
moment, or the \emph{same} action available at another moment and
not taken.  Both are genuine opportunity-cost readings; the
second is the \emph{temporal} margin on which durable means are
economized --- employing the boat today forgoes employing it at
some other time.  The stronger requirement that every employed
thing have a distinct competing action at the same moment would
be false in economically reasonable models: an absolutely
specific capital good --- Crusoe's boat, employed by
$\mathsf{DeepSeaFish}$ and by nothing else --- has no
cross-action competitor at any moment, yet is plainly scarce on
the temporal margin.  (Absolute specificity is precisely the
configuration on which the imputation theorem of future work
turns.)  Genuine same-moment cross-action rivalry is instead the
content of \axref{S1} below, which asserts existentially that
somewhere in the model two available actions compete for one
thing; the ownership enrichment of future work strengthens that
pattern further.  \axref{P6} and \axref{S1} are thus
complementary, and neither implies the other: \axref{P6} is
universal in form but per-thing liberal; \axref{S1} is sharper
in content but existential.
\end{remark}

\bigskip

\noindent\textbf{Action choice axiom \emph{(rationality of choice)}.}

\begin{enumerate}[label=(C1), leftmargin=2.5em]
\item\label{ax:C1} $\forall a\,t\,\alpha\,\beta\;(\Acts(a,\alpha,t)\land\Acts(a,\beta,t)\Rightarrow\alpha=\beta)$.
\hfill\emph{At any given time, each actor performs at most one
action.}\footnote{Together with~\axref{P1} and~\axref{P2}, this says $\Acts$
picks at most one action from each choice menu $\Cmenu_{a,t}$
(Eq.~\ref{eq:choice_menu}): the menu is non-empty (by~\axref{P1})
and the at-most-one performed action lies in it (by~\axref{P2}).
Axiom~\axref{O1} below (Always-acting) strengthens ``at most
one'' to ``exactly one''.}
\end{enumerate}

\noindent This axiom is the formal expression of \emph{choice}: at
each moment, the actor selects \emph{one} course of action from among
the available alternatives.  It is what makes demonstration possible ---
the act of choosing one action over others is the observable datum
that puts preference on record
(Definition~\ref{def:revpref} below).

\begin{remark}[Composite actions and unitarity]
\label{rem:composite_actions}
An element $\alpha$ of the Actions sort may represent what an
observer would decompose into simultaneous components ---
walking \emph{and} chewing gum, for instance.  One could
formalize this by letting the Actions sort carry labels drawn
from the power set of an underlying set of atomic components;
nothing in $\Tprx$ prevents such a labeling.  Axiom~\axref{C1},
however, treats the composite as a single praxeological unit:
at any time, each actor performs \emph{one} action, however
internally complex.  This reflects Mises's conception of
action as indivisible: ``Action is an attempt to substitute a
more satisfactory state of affairs for a less satisfactory
one'' \citep[Ch.~IV, \S4]{Mises1949}.  The agent does not
separately choose to walk and separately choose to chew; the
agent chooses one course of conduct whose description may
happen to mention both activities.  Decomposition into
components is an observer's modeling convenience, not a
praxeological primitive.  In $\Lprx$, this means the
granularity of ``atomic'' versus ``composite'' resides in how
one populates the Actions sort when specifying a structure,
not in the logic itself.\footnote{The ``exactly one action''
content of \axref{C1} should not be misread as requiring bodily
novelty: \emph{continuing} a previously chosen activity when
another was on the menu is itself an action token at $t$.
\citet[ch.~1, p.~6]{Rothbard2009} makes this explicit:
``Continuing to watch the game is just as much action as going
for a drive\ldots\ action does not necessarily mean that the
individual must stop doing what he has been doing and do
something else.''  What~\axref{C1} forbids is the actor's
performance of two distinct purposive units in one moment, not
the steady-state continuation of a previously selected one.}
\end{remark}

\begin{remark}[Composite ends]
\label{rem:composite_ends}
Symmetrically, nothing in $\Tprx$ requires ends to be atomic.
An element $E$ of the Ends sort may represent a complex state
of affairs --- ``arrive at the destination comfortably and on
time,'' for instance --- rather than a single indivisible
objective.  One natural formalization is to fix a set
$\mathcal{L}$ of \emph{atomic goals} (labels such as
``nourishment,'' ``shelter,'' ``leisure'') and let the Ends sort
be a subset of $\mathcal{P}(\mathcal{L})$: each end is a bundle
of goals that the action simultaneously serves.
Axioms \axref{P3}--\axref{P4} are fully compatible with such a labeling:
\axref{P3} requires that every action target \emph{some} bundle, and
\axref{P4} requires that this bundle be unique --- neither axiom
constrains whether the bundle is a singleton or a larger subset.
The preference order (a primitive of the signature, axiomatized
in the next subsection) then ranks bundles, not isolated goals,
which is precisely Mises's position: the actor ranks \emph{states of
affairs} as indivisible wholes, and any decomposition into
component goals is the theorist's analytical convenience
\citep[Ch.~IV, \S1]{Mises1949}.
\end{remark}

\subsubsection*{Layer 2: Preference and its demonstration in choice}

Two relations share this layer.  The primitive $\pref{a}{t}$,
declared in the signature of Section~\ref{subsec:primitives}, is
the actor's \emph{scale of values} (Mises).  The
\emph{demonstrated-preference record} $\revpref{a}{t}$ (Rothbard),
defined now, is the trace that choice leaves in the Layer-1 primitives.  The grounding axiom \axref{O0}
below ties the two together, and \axref{O2}--\axref{O4} give the
primitive order its structure.  Following Mises' insistence that
``the only source from which our knowledge concerning these scales
is derived is the observation of a man's actions''
\citep[p.~95]{Mises1949}, the record is the \emph{only} access
route to the order: the theory posits the scale but assigns it no
observational content beyond what choice reveals.

\begin{definition}[Demonstrated preference]\label{def:revpref}
Actor $a$ \emph{demonstrates a preference} for end $E$ over end $F$ at
time $t$, written $E\revpref{a}{t}F$,\footnote{We use Rothbard's term
\citep{Rothbard1956}: preference is \emph{demonstrated} in the act
itself, shown afresh at each $t$, with no axiom connecting distinct
instants (Remark~\ref{rem:constancy} below).  The relation
$\revpref{a}{t}$ is the bare, per-act core of what the choice-theory
literature calls \emph{revealed} preference
\citep{Samuelson1938,Richter1966}; the difference is exactly the
cross-observation constancy that the recovery program adds and we
decline.  Rothbard's choice of ``demonstrated'' over ``revealed''
marks just this point --- choice demonstrates a preference rather than
uncovering a standing one --- and that constancy is precisely his
principal objection to revealed-preference theory.  It is also what
the recovery rests on: \citet{Afriat1967}, building on
\citet{Houthakker1950}, shows that a finite family of choices is
rationalisable by a single utility function \emph{precisely} when it
satisfies the revealed-preference consistency axioms across
observations.  Time-indexing the record forgoes that constancy, and
with it any recovered utility, leaving only the ordinal, time-local
relation.} iff $E \neq F$\footnote{Under
\axref{P4} the clause $\alpha \neq \beta$ already follows from
$E \neq F$; both are stated so the definition reads correctly
independently of \axref{P4}.  The $E \neq F$ clause excludes the
trivial self-demonstration that means-ends multiplicity would
otherwise generate: at Crusoe's $t=1$
(Section~\ref{subsec:crusoe}) the chosen $\mathsf{BuildBoat}$ and
the available $\mathsf{BuildNet}$ share the end
$\mathsf{Capital}$, and without the clause the definition would
put $\mathsf{Capital} \revpref{\mathsf{C}}{1} \mathsf{Capital}$
on record.} and there exist actions $\alpha,\beta$ such that
\[
\Acts(a,\alpha,t)\;\land\;\Avail(a,\beta,t)\;\land\;\alpha\neq\beta
\;\land\;\EndOf(\alpha,E)\;\land\;\EndOf(\beta,F).
\]
\end{definition}

$E\revpref{a}{t}F$ holds whenever actor
$a$ chooses an action directed at $E$ while an action directed at
a different end $F$ was simultaneously available.  The record is
constituted by choice; the order it records is not.
The order's primitive status is a fact about the
\emph{signature}, not about the actor's psyche: $\pref{a}{t}$ is
the valuational primitive of the formal representation ---
Mises's ``instrument for the interpretation of a man's acting''
\citep[p.~95]{Mises1949} given explicit formal standing --- and
the theory assigns it no observational content beyond the record:
no introspection report, no interpersonal comparison, no cardinal
measure enters the language.  This is all that demonstrated
preference demands \citep{Rothbard1956}.\footnote{Mises himself
denies the scale independent existence: ``These scales have no
independent existence apart from the actual behavior of
individuals'' \citep[p.~95]{Mises1949}; and he warns against
hypostasizing the scale of value into ``the cause and motive of
the various individual actions'' \citep[pp.~102--103]{Mises1949}.
The primitive status of $\pref{a}{t}$ asserts nothing of the
kind.  Primitive-in-the-language is a statement about the
theory's representational apparatus --- an undefined symbol,
where a defined one would misreport the inferential situation
(Remark~\ref{rem:acts_vs_pref}) --- not a claim that a ranking
exists in the actor's head prior to and independently of action.
The grounding axiom \axref{O0} is the adequacy condition on the
instrument, and the per-instant index $t$ withholds exactly the
``duration and immutability'' that Mises's anti-hypostasis
passage targets (cf.\ Remark~\ref{rem:constancy}).}

\begin{remark}[Subjectivism is double-indexed]
\label{rem:subjectivism}
The relation $\pref{a}{t}$ is indexed by \emph{both} the actor and
the time, encoding two independent axes of subjectivism that
Mises insists on.  Across actors: ``ultimate ends are ultimately
given, they are purely subjective, they differ with various people
and with the same people at various moments in their lives''
\citep[Ch.~IV, \S2]{Mises1949}.  Across times for a single actor:
the same passage permits preference reversal across different
moments in one life, formally captured by the $t$-index.  And
``value is not intrinsic, it is not in things.  It is within us''
\citep[p.~96]{Mises1949}: the whole apparatus of $\pref{a}{t}$ refuses to
locate value in any cardinal property of ends or things: the
scale is ordinal, and its only observational content is the
record left by the actor's own choice history \axref{O0}.
\end{remark}

\begin{remark}[$\Acts$ and $\pref{a}{t}$ are not equivalent
information]
\label{rem:acts_vs_pref}
A natural temptation is to think that, given the grounding axiom
\axref{O0}, the conduct data $\Acts$ and the order $\pref{a}{t}$
are interchangeable --- that either could be taken as fundamental
with the other recovered.  The relation between them is asymmetric
in both directions, and the asymmetry is structurally important.

$\Acts$ together with $\Avail$ and $\EndOf$ determines the
\emph{record} $\revpref{a}{t}$ outright
(Definition~\ref{def:revpref}), and through \axref{O0} the record
constrains the order from below --- without ever determining it.
Conversely, the order pins down the chosen \emph{end} outright
(Proposition~\ref{prop:chosen_max}) but not, in general, the
chosen \emph{action}:

\begin{enumerate}[label=(\roman*), itemsep=2pt]
\item \textbf{Action-vs-end ambiguity.}  $\EndOf$ is functional
(by~\axref{P3} and~\axref{P4}) but need not be injective on the
menu: distinct available actions $\alpha \neq \beta$ may share an
end.  Neither the record nor the order can then distinguish which
was chosen --- choosing either yields the same demonstrated pairs.
The chosen \emph{action} is recoverable only when $\EndOf$ is
injective on $\Cmenu_{a,t}$.
\item \textbf{Sparseness at a single time.}  The record at a
single $(a,t)$ is a \emph{star}:\footnote{The term is
graph-theoretic.  Read the record as a directed graph on the
choice-relevant ends, with an edge $E \to F$ for each recorded pair
$E \revpref{a}{t} F$.  At a single $(a,t)$ every edge issues from
the one realized end --- the hub --- to a foregone end, and no edge
runs between two foregone ends; the graph is therefore the star
$K_{1,n}$ centred on the realized end --- equivalently, the
record is a strict partial order of height one, whose Hasse
diagram is this star \citep{DaveyPriestley2002}.  Being a star it is in
particular acyclic, which is what lets the record be extended to a
strict total order --- the satisfiability argument behind
Theorem~\ref{thm:asymm}.} by \axref{C1} and \axref{P4}
every demonstrated pair has the realized end on the left, so the
record says nothing about the order between two ends neither of
which was chosen.  The strict total order that
\axref{O2}--\axref{O4} impose on $\EndAt(a,t)$ is genuinely
additional structure --- the primitive order outruns its record.
Multi-time $\Acts$-data fills in further pairs, but only the
subset reached by repeated choice.
\end{enumerate}

The first
obstacle above gives rise to a natural equivalence relation on
the choice menu.  Define $\alpha \sim_{a,t} \beta$ if both
$\Avail(a, \alpha, t)$ and $\Avail(a, \beta, t)$ hold and
$\EndOf(\alpha) = \EndOf(\beta)$.  This is an equivalence
relation on $\Cmenu_{a,t}$, and its quotient
$\Cmenu_{a,t}/{\sim_{a,t}}$ is in bijection with the end-menu
$\EndOf(\Cmenu_{a,t})$.  $\pref{a}{t}$ lives at the level of the
quotient: it cannot distinguish among different means to the
same end.  This is precisely the means-ends multiplicity
exhibited at $\mathsf{Capital}$ in
Section~\ref{subsec:crusoe} ($\mathsf{BuildNet}$ and
$\mathsf{BuildBoat}$ lie in the same $\sim_{a,t}$-class), and
it is the structural reason why opportunity cost admits the
two readings of
Remark~\ref{rem:two_readings_opp_cost} (means lost vs.\ wants
unmet).

The order therefore carries information the conduct does not (the
ranking of unchosen ends), and the conduct carries information the
order does not (which of several same-end means realized the top);
neither subsumes the other.  This asymmetry is why $\pref{a}{t}$
is a separate primitive rather than an $\Acts$-derived
abbreviation: it underwrites the hidden-state framing of
Remark~\ref{rem:state_hidden} --- were the order $\Acts$-derived,
nothing would be left to hide --- and keeps the valuation
component first-class in the state, where Mises locates it.

The chosen end is \emph{always} the unique
$\pref{a}{t}$-maximum of $\EndAt(a,t)$
(Proposition~\ref{prop:chosen_max}); in a ``static
fully-observed'' regime --- a single time $t$, finite menu,
$\EndOf$ injective on the menu --- the chosen \emph{action} too
is recoverable from the order, and on a two-element end menu the
record already pins down the entire order.  The general regime is
the Misesian one: action, not preference, is the
\emph{epistemically} fundamental datum --- ``the only source from
which our knowledge concerning these scales is derived is the
observation of a man's actions'' \citep[p.~95]{Mises1949}.
\end{remark}

\begin{remark}[$\Avail$ is primitive at base level]
\label{rem:avail_primitive}
A symmetric question arises for $\Avail$: can the availability
relation $\Avail(a, \alpha, t)$ be \emph{computed} from $\Use$
and the Things sort, given that $\Use$ encodes the
technological-input pattern of every action?  At base level the
answer is \emph{no}, for three structural reasons.

\begin{enumerate}[label=(\roman*), itemsep=2pt]
\item \textbf{Time mismatch.}  $\Use(\alpha, x)$ has no time
argument; it is a static technological relation, true uniformly
across $t$.  $\Avail(a, \alpha, t)$ varies with $t$.  There is
no source in $\Use$ for the time-dependence that gives a Crusoe
$\Cmenu_0 \subsetneq \Cmenu_1 \subsetneq \Cmenu_2$ as capital
accumulates.
\item \textbf{Actor mismatch.}  $\Use(\alpha, x)$ has no actor
argument; it is actor-blind.  $\Avail(a, \alpha, t)$ has an
actor index, and two distinct actors at the same time may have
disjoint menus.
\item \textbf{History dependence.}  $\Avail$ at $t$ encodes the
actor's situation, which depends on prior history (what was
done before $t$, what capital was accumulated, what was acquired
or consumed).  None of that history-dependence is in $\Use$ or
in Things.
\end{enumerate}

The base axioms
constrain $\Avail$ but do not determine it: \axref{P1} requires
nonempty choice menus (Eq.~\ref{eq:choice_menu}), \axref{P2}
gives $\Acts \subseteq \Avail$, and \axref{P6} together with
\axref{S1} demand at least the existential structure of
shared-$\Use$ scarcity.  Subject to these, $\Avail$ is
\emph{primitive data specified by the modeler} --- the Crusoe
example of Section~\ref{subsec:crusoe} specifies the choice
menus $\Cmenu_0, \Cmenu_1, \Cmenu_2$ directly by inspection.

Combined with
Remark~\ref{rem:acts_vs_pref}, this places $\Avail$ and $\Acts$
in a distinguished position among the primitives of $\Lprx$:
$\Avail$ encodes \emph{what is possible}, $\Acts$ encodes
\emph{what is chosen}, and these are the praxeologically primary
observables.  The remaining base-level material either supports
them ($\EndOf$, $\Use$, $<$) or is downstream of them (the
record $\revpref{a}{t}$, $\EndAt(a,t)$, $\sigma_a(t)$).  The one
exception is the valuational primitive $\pref{a}{t}$: independent
data, disciplined by them through the grounding axiom \axref{O0}
and accessible only through them (Remark~\ref{rem:state_hidden}).
\end{remark}

\begin{definition}[Choice-relevant end]\label{def:endat}
End $E$ is on actor $a$'s \emph{end menu} at time $t$, written
$\EndAt(a,t,E)$, iff some action directed at $E$ is available to $a$ at $t$:
\[
\EndAt(a,t,E) \;:\Leftrightarrow\; \exists\alpha\;\bigl(\Avail(a,\alpha,t)\,\land\,\EndOf(\alpha,E)\bigr).
\]
\end{definition}
We write
\begin{equation}\label{eq:choice-menu-E_at}
\EndAt(a,t):=\{E\mid\EndAt(a,t,E)\}
\end{equation}
for the set of ends accessible to actor $a$ at time $t$ ---
the \emph{end menu} (parallel to the action menu
$\Cmenu_{a,t}$ of Eq.~\ref{eq:choice_menu}).
The two menus are related by
$\EndAt(a,t) = \EndOf(\Cmenu_{a,t})$, restricted to the image
in $\mathsf{Ends}$.

\medskip
\noindent\textbf{Layer 2 axioms \emph{(rationality of preference)}.}
The following axioms govern the primitive order $\pref{a}{t}$ and
tie it to the record: \axref{O0} grounds the order in conduct;
\axref{O2}--\axref{O4} give it the structure of a strict total
order on $\EndAt(a,t)$ --- what Mises calls the \emph{scale of
values}; \axref{O1} guarantees there is conduct to ground it in.
The diminishing-marginal-utility theorem of
Section~\ref{subsec:derived_dmu} relies on this order directly,
not on a duplicated (MU)-axiom.
Among these, \axref{O2} is the paper's clearest specimen of an
exposed hidden premise: completeness over never-chosen ends is an
idealisation that no finite choice history could certify ---
precisely the kind of silent import that the Hilbert-style format
is designed to make visible (cf.\ the Hud\'{i}k footnote in
Section~\ref{subsec:derived_dmu}).
\begin{enumerate}[label=(O\arabic*), leftmargin=2.5em, start=0]
\item\label{ax:O0} \textbf{Grounding.}
$\forall a\,t\,E\,F\;(E\revpref{a}{t}F\Rightarrow E\pref{a}{t}F)$.
\hfill\emph{What the record shows, the order
contains.}\footnote{This formalizes Mises's ``Every action is
always in perfect agreement with the scale of values or wants
because these scales are nothing but an instrument for the
interpretation of a man's acting''
\citep[Ch.~IV, \S2]{Mises1949}, read as an \emph{adequacy
condition} on the interpretive instrument: a scale adequate to
interpret the conduct must, at minimum, contain every pair the
conduct reveals.}
\item\label{ax:O1} \textbf{Always-acting.} $\forall a\,t\;\exists\alpha\;\Acts(a,\alpha,t)$.
\hfill\emph{At every time each actor performs some action.}
\item\label{ax:O2} \textbf{Menu-comparability.} $\forall a\,t\,E\,F\;(E\neq F\land\EndAt(a,t,E)\land\EndAt(a,t,F)\Rightarrow(E\pref{a}{t}F\lor F\pref{a}{t}E))$.
\hfill\emph{Any two distinct menu-available ends are comparable.}\footnote{With
\axref{O4} the disjunction is \emph{exclusive}: for distinct
menu-available ends exactly one of $E\pref{a}{t}F$, $F\pref{a}{t}E$
holds --- there is no third, ``indifferent'' case.  Here the primitive
order $\pref{a}{t}$ parts company with the neoclassical
preference-\emph{relation} tradition
\citep{Debreu1959,MasColell1995,Kreps1988}, which starts from a
\emph{weak} order $\succeq$ (complete and transitive) and
\emph{derives} indifference $\sim$ from it; praxeology admits only
\emph{strict} preference, on the Mises--Rothbard
\emph{demonstrated-preference} doctrine that action shows which end
is preferred and which forgone, indifference --- being actionless ---
never being demonstrated \citep{Rothbard1956}.  This is the
no-indifference face of demonstrated preference, imposed here on the
scale $\pref{a}{t}$; its no-constancy face --- the same doctrine
across time rather than within a menu --- is what the
demonstrated-preference record $\revpref{a}{t}$ carries, and is where
it parts from the constancy-laden \emph{revealed} preference of the
recovery program \citep{Samuelson1938} (footnote to
Definition~\ref{def:revpref}).  Apparent indifference between
\emph{units} of a good (as distinct from ends) is taken up at
\axref{MU2}.}
\item\label{ax:O3} \textbf{Transitivity.} $\forall a\,t\,E\,F\,G\;(E\pref{a}{t}F\land F\pref{a}{t}G\Rightarrow E\pref{a}{t}G)$.
\hfill\emph{Preference is transitively consistent at each moment.}
\item\label{ax:O4} \textbf{Asymmetry.}
$\forall a\,t\,E\,F\;(E\pref{a}{t}F\Rightarrow\neg(F\pref{a}{t}E))$.
\hfill\emph{In particular ($E=F$): the order is irreflexive.}
\end{enumerate}

\begin{remark}[Transitivity is within-moment]
Axiom \axref{O3} applies to $\pref{a}{t}$ at a \emph{fixed} time $t$. It says nothing
about the relation between $\pref{a}{t}$ and $\pref{a}{t'}$ for $t\neq t'$.
The configuration $E\pref{a}{t}F$ and $F\pref{a}{t'}E$ for $t<t'$ is consistent
with the entire axiom system --- it is a \emph{preference reversal over time},
not a contradiction. Mises: ``Constancy and rationality are entirely different
notions'' \citep[p.~103]{Mises1949}. Axiom \axref{P5} ensures that the two
demonstrating choices occur at different times, so no within-moment
cycle can ever be put on record.
\end{remark}

Together, \axref{O2}--\axref{O4} make $\pref{a}{t}$ a strict total
order on the finite menu $\EndAt(a,t)$ --- what Mises calls the
\emph{scale of values}.  The scale is doubly local: to the menu
$\EndAt(a,t)$, and to the instant $(a,t)$, with no axiom forcing one
ranking to persist across times (Remark~\ref{rem:constancy}).  On a
finite menu such a per-instant order is automatically representable by
its rank function, so what is withheld is not a per-instant utility
encoding but a \emph{global utility} $u:\mathsf{Ends}\to\mathbb{R}$ ---
a single, constant scale ranking all ends at all times.  The axioms are jointly
satisfiable over every Layer-1 structure: at each $(a,t)$ the
record is a \emph{star} whose left element is the single realized
end --- a consequence of \axref{C1} and \axref{P4}, established
within the proof of Theorem~\ref{thm:asymm} --- hence acyclic,
and on the finite end menu (Remark~\ref{rem:finiteness}) any
acyclic relation extends to a strict total order; the Crusoe model
of Section~\ref{subsec:crusoe} exhibits such an extension
explicitly.

\subsubsection*{Layer 3: Material scarcity}

The previous two layers settle the structure of action and the
preference order on ends (together with its demonstrated record).  The
third layer adds a single axiom asserting that a genuine resource
conflict exists somewhere in the model --- the minimal anchor that
keeps the framework non-trivial.

\begin{enumerate}[label=(S\arabic*), leftmargin=2.5em]
\item\label{ax:S1} \textbf{Existence of scarcity} \emph{(rationality of scarcity recognition).}
$\exists a\,t\,\alpha\,\beta\,x\;(\alpha\neq\beta\land\Avail(a,\alpha,t)\land\Avail(a,\beta,t)\land\Use(\alpha,x)\land\Use(\beta,x))$.
\hfill\emph{At least one genuine resource conflict exists.}
\end{enumerate}

\noindent \axref{S1} ensures that the framework is nontrivial: at least one
actor faces a genuine resource constraint.  Without scarcity, there is
no economic problem --- a point Mises repeatedly emphasizes:
``means are necessarily always limited, i.e., scarce''
\citep[Ch.~IV, \S1]{Mises1949}; and where ``man is not restrained
by the insufficient quantity of things available, there is no need
for any action''~\citep[Ch.~IV, \S1]{Mises1949}.

\subsubsection*{Synthesis: states of affairs and actions}
\label{subsec:state}

The base axioms already determine, at any moment $(a, t)$, all
the data needed to characterize the actor's \emph{configuration}:
what is available (the choice menu $\Cmenu_{a,t}$ of
Eq.~\ref{eq:choice_menu}) and how the actor ranks the induced
ends (the $\pref{a}{t}$-restriction).  
 We collect the
irredundant
state data into an object we call the actor's \emph{state of
affairs} at $(a, t)$, formalizing Mises's ubiquitous use of the
phrase.  No new primitive is introduced and no new sort is
added: the state is fully recoverable from the primitive
relations of $\Lprx$ restricted to $(a, t)$.  This minimalism is
deliberate --- praxeology, in Mises's framing, is
\emph{process}-primary, not state-primary; the state is an
abstraction we extract from action, not an object we postulate
ahead of action.

\begin{definition}[State of affairs]
\label{def:state}
The \emph{state of affairs} of actor $a$ at time $t$, written
$\sigma_a(t)$, is the choice menu $\Cmenu_{a,t}$
(Eq.~\ref{eq:choice_menu}) equipped with the preference order
$\pref{a}{t}$ on the end menu $\EndAt(a,t)$:
\[
  \sigma_a(t) \;:=\; \bigl(\,\Cmenu_{a,t}\,,\;\pref{a}{t}\,\bigr).
\]
The chosen action $\Acts(a, \cdot, t)$ is \emph{not} part of the
state: in the reading below, $\Acts$ records an \emph{arrow}
between states rather than a feature of either endpoint.  The
choice menu is the carrier; $\pref{a}{t}$ is structure on it.
The two components differ in epistemic status: the menu is
externally inferable, while the preference order is hidden ---
an observer reaches it only through the record of
Definition~\ref{def:revpref} (Remark~\ref{rem:state_hidden}).
Other time-invariant data ($\EndOf$, $\Use$, the $<$-order on $\mathsf{Times}$)
lives at the model level rather than in the state.
\end{definition}

Why preferences in the state?  The choice menu
$\Cmenu_{a,t}$ alone records what the actor \emph{can} do; the
preference order $\pref{a}{t}$ records the \emph{valuation} of
the available alternatives.  Mises's own state vector for a market
day, $S(D) = (V_D, O_D, P_D, T_D)$ \citep[Ch.~XXVI, \S6]{Mises1949}, opens with the valuation component $V_D$ for
exactly the same reason: the satisfactory description of an
acting world at a moment must include not only what is materially
on hand but how it is ranked.  Our $\sigma_a(t)$ is the
single-actor shadow of that vector --- $V_D$ mirrored by
$\pref{a}{t}$ --- and, like it, a state from which action carries
one configuration to the next, history entering only through the
current state.  In the base layer here only
$\pref{a}{t}$ carries this content; the production-side
components $O_D, P_D, T_D$ are inert until the production
enrichment of future research is introduced.  Postponing $\pref{a}{t}$ to a
non-state
``parameter'' would force the same content back into the state
via a different name once that apparatus is introduced in future
research, with no gain in clarity at the base level.

We deliberately do \emph{not} promote $\sigma_a(t)$ to a
primitive sort, for two complementary reasons.

Technically, the state is already fully determined by
the existing primitive-restrictions
$(\Avail|_{(a,t)},\,\pref{a}{t})$.  Adding a States sort would
duplicate data already in the language, require new axioms to
enforce consistency between the sort and the underlying
primitives, and force re-verification of the existing theorems
through the enlarged signature --- all with no gain in
expressive power at the base level.

Philosophically, Mises is explicit on this point.  In
the calculation chapter he writes: ``equilibrium and equilibrium
prices --- these notions are foreign to real life and action;
they are auxiliary tools of praxeological reasoning for which
there is no mental means to conceive the ceaseless restlessness
of action other than to contrast it with the notion of perfect
quiet'' \citep[Ch.~XXVI, \S6]{Mises1949}.  States are
\emph{auxiliary tools}; action is the praxeological primitive.
Promoting State to a primitive sort would invert this ontological
priority --- it would import the Walrasian habit of treating
state as fundamental and action as a relation between primitive
states, exactly the inversion Mises spends Ch.~XXVI arguing
against.  Keeping $\sigma_a(t)$ outside the formal signature --- a
defined abbreviation, not a sort --- preserves the priority of
action over state.

Action is the operator that transforms $\sigma_a(t)$ into
$\sigma_a(t')$, where $t'$ is a successor of $t$ in the time
order (axiom~\axref{T6} guarantees one):
\begin{equation}\label{eq:action_arrow}
  \alpha = \Acts(a, \cdot, t) \;:\; \sigma_a(t)
  \;\longrightarrow\; \sigma_a(t').
\end{equation}
The end $E = \EndOf(\alpha)$ is the arrow's \emph{intended}
target --- an imagined post-action state.  When the action
succeeds, the realized $\sigma_a(t')$ matches $E$; when it fails
(a phenomenon treated in future research), the two diverge, and the gap
is the content of profit and loss.  This is the Misesian
picture: ``Action is an attempt to substitute a more
satisfactory state of affairs for a less satisfactory one''
\citep[Ch.~IV, \S4]{Mises1949}; ``ACTION always is essentially the
exchange of one state of affairs for another state of affairs''
\citep[Ch.~X, \S1]{Mises1949}.
\citet[ch.~1, pp.~19,~71]{Rothbard2009} makes the
cost-benefit content of this exchange explicit: the actor's
\emph{costs} are ``his forgone opportunities to enjoy consumers'
goods,'' and the expected greater utility on his value scale is
his \emph{psychic income}.  The arrow $\sigma_a(t) \to
\sigma_a(t')$ is then a one-step state exchange with foregone
alternatives as cost and the realized end as psychic income;
Theorem~\ref{thm:opp_cost} formalizes the cost side, and the
\emph{ex ante} gain criterion of
future work on extensions
formalizes the gain side.

A consequence of this reading is that
preferences are forward-looking.  By~\axref{O1} and~\axref{C1} exactly one action
$\alpha_{a,t}$ is performed at $(a,t)$, and the arrow reading
above takes its end
$E = \EndOf(\alpha_{a,t})$ as the actor's imagined target --- the
post-action state aimed at, not yet realized.  The
\emph{foregone} ends of $\EndAt(a,t)$ are likewise imagined: each
is the target of an available action $\beta \neq \alpha_{a,t}$
that was not chosen.  Neither $E$ nor $F$ in $E\,\pref{a}{t}\,F$
therefore denotes a \emph{present} state of affairs; both are
positions in the actor's forward projection from $t$.  The
preference order is the actor's
valuation of \emph{imagined-future} ends, not an ordering on
present configurations; the record of
Definition~\ref{def:revpref} inherits the same forward-looking
character from the choices that constitute it.

This is the formal counterpart
of the second of Mises's three prerequisites of action.  ``Felt
uneasiness'' (the first) is what motivates action; ``the image
of a more satisfactory state''
\citep[Ch.~I, \S1]{Mises1949} (the second) is what action aims at.
$\pref{a}{t}$ formalizes the second --- it ranks imagined
targets, not realities.

This forward-looking reading underlies three later choices.  It
is why $\sigma_a(t)$ (Definition~\ref{def:state}) carries the
preference component $\pref{a}{t}$ rather than a present-tense
ranking --- the hiddenness of Remark~\ref{rem:state_hidden} is
essential because the comparison lives in the actor's forward
projection, not in any configuration an observer could read off.
It is why the \emph{ex ante} profit-and-loss criterion of future
work compares an expected post-action end to a status-quo end,
both forward-looking.  And it is why the allocation predicate
$\Allot(a,t,x,E)$ behind the marginal-utility theorem
(Section~\ref{subsec:derived_dmu}) is a \emph{forward-directed
plan} --- the actor's intended assignment of unit $x$ to the
imagined end $E$ --- not a record of past deployment.

\begin{remark}[The irrelevance of sunk costs]
\label{rem:sunk_costs}
A structural consequence of the forward-looking character of
$\pref{a}{t}$ and the minimal content of the state
is the classical irrelevance of sunk costs.  The state
$\sigma_a(t)$ carries the current menu and the forward-looking
order $\pref{a}{t}$, and nothing else
(Definition~\ref{def:state}); past expenditures --- resources
consumed, labor expended, ends foregone at earlier times ---
appear nowhere in it.  The past bears on the present only
through what it has left in the menu: Crusoe's $\Cmenu_2$
contains $\mathsf{DeepSeaFish}$ because the boat is in hand,
not because building it was costly.  And since $\pref{a}{t}$
ranks imagined-future ends only
(as noted above), no axiom gives a past outlay
any standing in the valuation.  ``Bygones are bygones'' is
therefore not a counsel of prudence that an actor might fail
to heed; in $\Tprx$ it is a typing fact --- the state has no
slot for a sunk cost to occupy.  This is the structural face
of Mises's forward-looking account of cost: ``costs are equal
to the value attached to the satisfaction which one must
forego in order to attain the end aimed at''
\citep[p.~97]{Mises1949}.
\end{remark}

\begin{remark}[The praxeological state is partially hidden]
\label{rem:state_hidden}
The state $\sigma_a(t)$ has both observable and hidden
components.  The menu $\Cmenu_{a,t}$ and the resource-competition
structure on it (the shared-$\Use$ pattern on $\Cmenu_{a,t}$)
are externally inferable from physical possibilities.  The
preference component $\pref{a}{t}$ is structurally hidden: it is a
primitive of $\Lprx$, but the only axiom connecting it to anything
observable is the grounding axiom \axref{O0}, so an external
observer holds only the record $\revpref{a}{t}$
(Definition~\ref{def:revpref}) --- the partial trace that the
actor's choice history leaves of the scale.  In the language of
state machines, the praxeological setup is a \emph{hidden} state
machine: $\sigma_a(t)$ is the hidden state, $\Acts(a, \cdot, t)$
is the observation, and the record is the only access route.
Primitive/defined is a \emph{syntactic} classification (undefined
symbol versus abbreviation); observable/hidden is an
\emph{epistemic} one (evaluable from conduct data versus not).
The two cross-classify: $\pref{a}{t}$ is primitive yet hidden, the
record is defined yet observable.
This hiddenness is exactly Mises's point that ``the only source
from which our knowledge concerning these scales is derived is
the observation of a man's actions'' \citep[p.~95]{Mises1949}.
\end{remark}

\begin{remark}[Finiteness assumptions]
	\label{rem:finiteness}
	We assume throughout that each actor's choice menu of
	available actions $\Cmenu_{a,t}$ is finite at every time:
	only finitely many actions are available to any actor at any
	moment.  Its image under $\EndOf$, the end menu
	$\EndAt(a,t) = \EndOf(\Cmenu_{a,t})$, is then finite as well.
	The sort $\mathsf{Things}$ is likewise assumed finite (there
	are finitely many commodity types).  These are economically
	uncontroversial restrictions that hold in every model
	constructed in this paper.

	No finiteness assumption is imposed on the sorts
	$\mathsf{Actors}$, $\mathsf{Ends}$, or $\mathsf{Times}$.
	The axiom system $\Tprx$ admits models with arbitrarily many
	actors, arbitrarily many ends, and an unbounded time horizon.
	This generality is essential for the undecidability theorem
	of forthcoming work, which depends on the model class of $\Tprx$
	containing economies of unbounded size.  In the present paper
	the generality is preserved for compatibility with the later
	results, but plays no further role.

	Finiteness is a condition on the (set-theoretic) models of
	Definition~\ref{def:model}, not a first-order axiom: by
	compactness no $\Lprx$-theory pins it down (constants forcing
	an arbitrarily large menu can always be adjoined).  It is
	therefore imposed at the level of the model class, not as an
	axiom of $\Tprx$.  Accordingly the
	results that rely on them --- the marginal-end apparatus of
	Definition~\ref{def:marginal_end},
	Lemma~\ref{lem:served_choice_relevant}, and
	Theorem~\ref{thm:DMU}, where finiteness of $\EndAt(a,t)$ is
	what supplies the existence of a least-preferred served end ---
	carry the finiteness of $\EndAt(a,t)$ as an explicit
	hypothesis rather than leaving it tacit.  Every model
	constructed in this paper, and both machine-checked
	consistency models of Appendix~\ref{sec:lean}, have finite
	choice menus and finitely many commodity types; the
	infinite-time model realizes an unbounded $\mathsf{Times}$,
	which these restrictions explicitly permit.
\end{remark}

\section{A worked illustration: Robinson Crusoe}
\label{subsec:crusoe}

Before turning to the derived theorems, we illustrate $\Lprx$ with a
small running model --- a Crusoe economy whose first three periods
already exhibit the characteristic structure of capital formation:
monotone growth of the choice menu as productive capital accumulates.  This
model is the running example for the present paper.\footnote{The single-actor Crusoe model is deliberately the
\emph{simplest} case in which the apparatus does any work.  As a
general matter, problems of central planning do not arise in such
an economy: when only Crusoe acts, the question ``which actor
should receive the scarce resource?'' is trivially answered by
``Crusoe.''  The interesting Misesian content of central planning
emerges already as soon as a second actor (a ``Friday'') joins
the island --- and it is in the Crusoe--Friday economy
$\mathfrak{C}_2$ taken up in future research.}  The
machine-checked version of this model in \Lean{} is described in
Appendix~\ref{sec:lean}.

Crusoe inhabits a small island with access to coastal and deep-sea
resources.  Define:
\begin{align*}
\mathsf{Actors}^{\mathfrak{C}} &:= \{\mathsf{Crusoe}\}\\
\mathsf{Actions}^{\mathfrak{C}} &:= \{\mathsf{Forage},\;\mathsf{BuildNet},\;
                                  \mathsf{ShoreFish},\;\mathsf{BuildBoat},\;
                                  \mathsf{DeepSeaFish}\}\\
\mathsf{Ends}^{\mathfrak{C}} &:= \{\mathsf{Subsist},\;\mathsf{Capital},\;
                           \mathsf{ShoreCatch},\;\mathsf{DeepCatch}\}\\
\mathsf{Times}^{\mathfrak{C}} &:= \mathbb{N}
\end{align*}
\noindent The time sort is $\mathbb{N}$: time does not end, and the
discreteness axiom \axref{T6} --- every moment has an immediate
successor --- admits no finite model.  What we display and reason
about is the \emph{snapshot} $\{0,1,2\} \subseteq \mathbb{N}$: the
capital-accumulation window, in which Crusoe passes through the
three epochs $Q_0, Q_1, Q_2$ of Table~\ref{tab:crusoe_epochs} below
along his canonical trajectory.  The history continues over
$\mathbb{N}$: the actor keeps acting, traversing the state
diagram of Figure~\ref{fig:crusoe_state_diagram} (whose
self-loops allow menu-preserving repeats) for as many rounds as
the line provides.  The complete specification --- whose
\emph{non-greedy} continuation is what secures \axref{T6} and the
free-good axiom \axref{P6} --- is the machine-checked
$\mathbb{N}$-time model of Appendix~\ref{sec:lean} (also
Section~\ref{sec:full_model}).

\medskip
With action-end assignments: 
\begin{align*}
&\EndOf(\mathsf{Forage},\mathsf{Subsist}),\quad
\EndOf(\mathsf{BuildNet},\mathsf{Capital}),\quad
\EndOf(\mathsf{ShoreFish},\mathsf{ShoreCatch}),\\
&\EndOf(\mathsf{BuildBoat},\mathsf{Capital}),\quad
\EndOf(\mathsf{DeepSeaFish},\mathsf{DeepCatch}).
\end{align*}

\noindent The naming is suggestive but does \emph{not} put
actions in one-to-one correspondence with ends:
$\mathsf{Capital}$ is the shared end of both capital-formation
actions $\mathsf{BuildNet}$ and $\mathsf{BuildBoat}$, while
$\mathsf{ShoreCatch}$ and $\mathsf{DeepCatch}$ are the distinct
ends of the two consumption actions those capital goods make
possible.  We call this pattern \emph{means-ends multiplicity}
at $\mathsf{Capital}$: two or more distinct available actions
$\alpha \neq \beta$ in $\Cmenu_{a,t}$ share an end
$\EndOf(\alpha) = \EndOf(\beta) = E$, so the restriction of
$\EndOf$ to $\Cmenu_{a,t}$ fails to be injective at $E$.  This
is the simplest exhibition of the Misesian distinction between
actions and the wants they serve: the actor faces alternative
tools for one purpose.\footnote{Mises is explicit that means and ends do not
stand in one-to-one correspondence.  In Ch.~IV~\S2 he locates the
end at the level of \emph{felt uneasiness} (``the end, goal, or
aim of any action is always the relief from a felt uneasiness; a
means is what serves to the attainment of any end''), and in
Ch.~XV~\S5 he speaks of ``various means'' that ``allow for various
uses,'' setting the actor the task of allocating them.  An action
admits a proximate end (here $\mathsf{NetCapital}$ or 
$\mathsf{BoatCapital}$, an immediate result) and a more abstract
end-category (here $\mathsf{Capital}$, a class of want).}  Here
we treat the end names purely as labels; the higher-order
goods structure that justifies the
naming is an enrichment to be introduced in future research.

The choice menu grows monotonically as capital accumulates.  At
$t=0$, lacking any tools, Crusoe can only forage or begin building a
net.  Once a net is in hand (after $t=0$), shore fishing and boat
building become available; once a boat is in hand (after $t=1$),
deep-sea fishing becomes available.  Formally:
\begin{align*}
\Cmenu_0 &= \{\mathsf{Forage},\;\mathsf{BuildNet}\}, \\
\Cmenu_1 &= \{\mathsf{Forage},\;\mathsf{BuildNet},\;\mathsf{ShoreFish},\;
              \mathsf{BuildBoat}\}, \\
\Cmenu_2 &= \{\mathsf{Forage},\;\mathsf{BuildNet},\;\mathsf{ShoreFish},\;
              \mathsf{BuildBoat},\;\mathsf{DeepSeaFish}\}.
\end{align*}
The total
number of menu-compatible global trajectories is
$2 \times 4 \times 5 = 40$.  These are combinatorial selections
from the fixed menus, not all of them coherent histories: the
menus shown are those obtaining along Crusoe's canonical
capital-accumulation trajectory, so a deviating selection
(e.g.\ $\mathsf{DeepSeaFish}$ at $t=2$ with no boat ever built)
is menu-compatible without being capital-feasible.

\medskip\noindent
The dependency of $\Cmenu_t$ on the prior history --- a net,
then a boat --- is encoded here \emph{externally}: we specify
the choice menus $\Cmenu_t$ directly, by inspection of what
the actor must have done before $t$ for each menu item to be
physically meaningful.  

Crusoe's \emph{actual} history is the natural capital-accumulation
trajectory --- build a net, then a boat, then enjoy the
highest-yield catch the new equipment makes available:
\[
\Acts(\mathsf{Crusoe},\mathsf{BuildNet},0),\quad
\Acts(\mathsf{Crusoe},\mathsf{BuildBoat},1),\quad
\Acts(\mathsf{Crusoe},\mathsf{DeepSeaFish},2).
\]

\begin{figure}
\centering
\begin{tikzpicture}[
  state/.style={circle, draw, thick, minimum size=1.3cm, font=\small,
                inner sep=1pt},
  emph/.style={state, line width=1pt, fill=blue!8},
  arrow/.style={-{Latex[length=2.5mm]}, thick},
  canon/.style={arrow, line width=1.2pt},
  selfloop/.style={arrow, looseness=12},
  canonselfloop/.style={canon, looseness=12},
  menubox/.style={font=\scriptsize, align=center, draw=gray!40,
                  rounded corners=2pt, fill=gray!4, inner sep=4pt},
  arrowlabel/.style={font=\scriptsize, midway},
  looplabel/.style={font=\scriptsize, inner sep=1pt}
]

\node[emph] (Q0) at (0, 0)  {$Q_0$};
\node[emph] (Q1) at (5, 0)  {$Q_1$};
\node[emph] (Q2) at (10, 0) {$Q_2$};

\node[font=\scriptsize\itshape, below=2pt of Q0] (sub)   {subsistence};
\node[font=\scriptsize\itshape, below=2pt of Q1] (coast) {coastal};
\node[font=\scriptsize\itshape, below=2pt of Q2] (deep)  {deep-sea};

\node[menubox, below=4pt of sub]
  {\textbf{end menu}\\$\{\mathsf{S}, \mathsf{K}\}$};
\node[menubox, below=4pt of coast]
  {\textbf{end menu}\\$\{\mathsf{S}, \mathsf{K}, \mathsf{SC}\}$};
\node[menubox, below=4pt of deep]
  {\textbf{end menu}\\$\{\mathsf{S}, \mathsf{K}, \mathsf{SC}, \mathsf{DC}\}$};

\draw[canon] (Q0) -- node[arrowlabel, above]{$\mathsf{BN}$ / K} (Q1);
\draw[canon] (Q1) -- node[arrowlabel, above]{$\mathsf{BB}$ / K} (Q2);

\draw[selfloop] (Q0) edge[loop above]
  node[looplabel, above]{$\mathsf{Fo}$ / S} (Q0);

\draw[selfloop] (Q1) edge[out=140, in=110]
  node[looplabel, above left]{$\mathsf{Fo}$ / S} (Q1);
\draw[selfloop] (Q1) edge[out=100, in=80]
  node[looplabel, above]{$\mathsf{BN}$ / K} (Q1);
\draw[selfloop] (Q1) edge[out=70, in=40]
  node[looplabel, above right]{$\mathsf{SF}$ / SC} (Q1);

\draw[selfloop] (Q2) edge[out=160, in=130]
  node[looplabel, above left]{$\mathsf{Fo}$ / S} (Q2);
\draw[selfloop] (Q2) edge[out=120, in=100]
  node[looplabel, above]{$\mathsf{BN}$ / K} (Q2);
\draw[selfloop] (Q2) edge[out=90, in=60]
  node[looplabel, above]{$\mathsf{SF}$ / SC} (Q2);
\draw[selfloop] (Q2) edge[out=50, in=20]
  node[looplabel, above right]{$\mathsf{BB}$ / K} (Q2);
\draw[canonselfloop] (Q2) edge[out=-20, in=-50]
  node[looplabel, right]{\ $\mathsf{DSF}$ / DC} (Q2);

\node[menubox] (legend) at (5, -4)
  {\textbf{abbreviations}\\[2pt]
   \begin{tabular}{@{}r@{\,$=$\,}l@{\hspace{2em}}r@{\,$=$\,}l@{}}
   \multicolumn{2}{c}{\textit{actions}} & \multicolumn{2}{c}{\textit{ends}}\\
   $\mathsf{Fo}$  & $\mathsf{Forage}$      & $\mathsf{S}$  & $\mathsf{Subsist}$    \\
   $\mathsf{BN}$  & $\mathsf{BuildNet}$    & $\mathsf{K}$  & $\mathsf{Capital}$    \\
   $\mathsf{SF}$  & $\mathsf{ShoreFish}$   & $\mathsf{SC}$ & $\mathsf{ShoreCatch}$ \\
   $\mathsf{BB}$  & $\mathsf{BuildBoat}$   & $\mathsf{DC}$ & $\mathsf{DeepCatch}$  \\
   $\mathsf{DSF}$ & $\mathsf{DeepSeaFish}$ & \multicolumn{2}{c}{}                 \\
   \end{tabular}};

\end{tikzpicture}
\caption{State-and-action diagram for the nets-and-boats Crusoe.
Each state $Q_i$ is identified by capital inventory (see
Table~\ref{tab:crusoe_epochs}); the time index of the action is
folded into the dynamics --- the actor traverses the diagram for
as many rounds as needed.  Each state carries its \emph{end menu}
$\EndOf(\Cmenu_{Q_i})$.  The figure thus displays only the
\emph{menu} component of $\sigma_a(t)$
(Definition~\ref{def:state}); the preference decoration
$\pref{a}{t}$ is suppressed, in keeping with its hiddenness
(Remark~\ref{rem:state_hidden}) --- which of the available ends
ranks above which is not on the page.  At every visit the actor
performs one $\alpha \in \Cmenu_{Q_i}$, realizing
$E = \EndOf(\alpha)$; the remaining ends in $\EndOf(\Cmenu_{Q_i})$
are \emph{foregone} --- losses in the opportunity-cost sense.
Forward edges are capital formation; self-loops are
menu-preserving rounds.  Edge labels of the form $\alpha/E$ pair
the action with the end it realizes.  $\mathsf{BN}$ and
$\mathsf{BB}$ share the end $\mathsf{K}$ --- the means-ends
multiplicity introduced above.  The bold trajectory $Q_0
\xrightarrow{\mathsf{BN}} Q_1 \xrightarrow{\mathsf{BB}} Q_2
\xrightarrow{\mathsf{DSF}} Q_2$ is Crusoe's canonical
capital-accumulation history, with $\mathsf{DSF}$ bolded inside
$Q_2$'s self-loop.  Future extensions work refines this picture along four
independent axes (transition map $\tau$, fallibility predicates,
regression edges from degradation, recipe-DAG zoom)
without redrawing it.}
\label{fig:crusoe_state_diagram}
\end{figure}
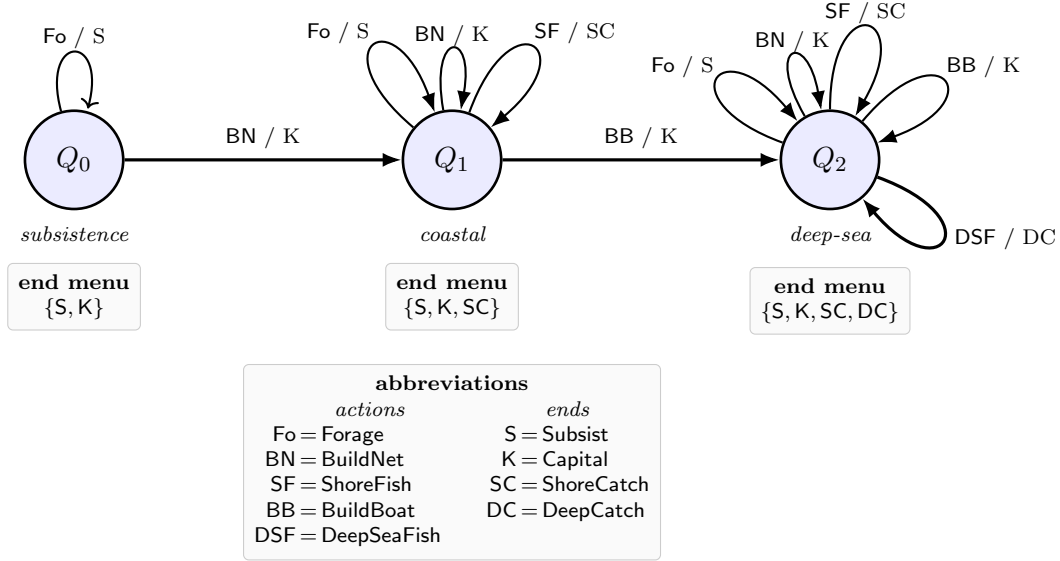

The three choice menus $\Cmenu_0, \Cmenu_1, \Cmenu_2$ are pairwise
distinct as sets, and each strictly contains its predecessor.
Anticipating future research, we give these three menu-classes
the names $Q_0, Q_1, Q_2$.  These $Q_i$ are the
\emph{capital epochs} of $\mathfrak{C}$ --- the equivalence
classes of states $\sigma_a(t)$ (Definition~\ref{def:state})
under menu equality.  Two states with the same action menu
inhabit the same capital epoch; the actor's preferences
$\pref{a}{t}$ vary across visits to the same epoch but the
menu does not.  In the figure above, $Q_i$-nodes denote these
equivalence classes; in axiomatic prose elsewhere we continue
to write $\sigma_a(t)$ for the full state at $(a,t)$.

\begin{table}[ht]
\centering
\renewcommand{\arraystretch}{1.3}
\caption{The three capital epochs of $\mathfrak{C}$, named in
anticipation of the macro-graph treatment of
future work on extensions.  Each
epoch is fully characterized by the menu of actions available in
it.}
\label{tab:crusoe_epochs}
\begin{tabular}{@{}lll@{}}
\toprule
\textbf{Epoch} & \textbf{Menu} &
\textbf{When the choice menu is in this epoch} \\
\midrule
$Q_0$ (Subsistence) &
  $\{\mathsf{Forage},\mathsf{BuildNet}\}$ &
  no tools yet \\
$Q_1$ (Coastal) &
  $Q_0 \cup \{\mathsf{ShoreFish},\mathsf{BuildBoat}\}$ &
  net built, boat not yet \\
$Q_2$ (Deep Sea) &
  $Q_1 \cup \{\mathsf{DeepSeaFish}\}$ &
  net and boat both built \\
\bottomrule
\end{tabular}
\end{table}

\noindent
The trajectory underlying $\mathfrak{C}$ traverses these three
epochs in order: $\Cmenu_0$ realizes the Subsistence menu, $\Cmenu_1$
the Coastal menu, $\Cmenu_2$ the Deep-Sea menu.  Building the net at
$t=0$ unlocks $Q_1$; building the boat at $t=1$ unlocks $Q_2$;
fishing the deep sea at $t=2$ exercises the productive option that
the two earlier capital-formation steps have made physically
possible.

\begin{remark}[The figure: forgetting preferences, folding time]
\label{rem:state_to_epoch}
The figure comes from the state history by two projections of
different kinds.  First, at each moment $\pi_1$ drops the preference
decoration,
\[
\sigma_a(t) = \bigl(\Cmenu_{a,t},\, \pref{a}{t}\bigr)
\;\xrightarrow{\;\pi_1\;}\; \Cmenu_{a,t},
\]
keeping only the menu the actor faces; the ranking $\pref{a}{t}$
is structurally hidden (Remark~\ref{rem:state_hidden}) and cannot
be drawn.  Second, $\pi_2$ folds time.  The menu history
$t \mapsto \Cmenu_{a,t}$ takes only finitely many values --- here
the three $\Cmenu_{Q_0}, \Cmenu_{Q_1}, \Cmenu_{Q_2}$, growing
monotonically and then holding steady --- so the whole
$\mathbb{N}$-time line projects onto the finite state diagram,
\[
t \;\xrightarrow{\;\pi_2\;}\; Q_i
\qquad (\Cmenu_{a,t} = \Cmenu_{Q_i}),
\]
each moment sitting at the node that carries its menu, the
time-order supplying the forward edges $Q_0 \to Q_1 \to Q_2$ and
menu-preserving rounds the self-loops.  It is this finiteness ---
only three menus ever occur --- that lets the unbounded time line
be drawn as one finite picture.

The two losses are of different kinds.  The preference loss at
$\pi_1$ is irreducible: $\pref{a}{t}$ cannot be read back from the
figure.  The time loss at $\pi_2$ is not: a node may be visited on
many rounds, but the trajectory pins down which, since the
position in the action sequence fixes $t$ (and $a = \mathsf{Crusoe}$
is constant in $\mathfrak{C}$).

In a multi-actor model the actor label $a$ \emph{must} be kept:
the projection sends $\sigma_a(t)$ to $(a, Q_i^a)$, the node on
$a$'s own diagram, and the diagrams may differ between actors
(recipes, ownership, and accumulated capital all vary).
\end{remark}

At $t=0$, Crusoe chooses $\mathsf{BuildNet}$ while $\mathsf{Forage}$
is available.  By Definition~\ref{def:revpref} this puts on record
\[
\mathsf{Capital} \;\revpref{\mathsf{C}}{0}\; \mathsf{Subsist},
\]
and so, by \axref{O0},
$\mathsf{Capital} \pref{\mathsf{C}}{0} \mathsf{Subsist}$.  The
choice menu $\Cmenu_0$ contains only these two actions, so the
record already pins down the entire order on
$\EndAt(\mathsf{C},0)$: $\pref{\mathsf{C}}{0}$ is the strict total
order $\mathsf{Capital} \succ \mathsf{Subsist}$ on the end menu.

At $t=1$, Crusoe chooses $\mathsf{BuildBoat}$ while
$\mathsf{Forage}, \mathsf{BuildNet}, \mathsf{ShoreFish}$ are also
available.  Both $\mathsf{BuildBoat}$ and $\mathsf{BuildNet}$
share the end $\mathsf{Capital}$, so the choice between them
puts nothing on record at the end-level --- the $E \neq F$ clause
of Definition~\ref{def:revpref} is silent on the shared end (this
is the means-ends multiplicity introduced above).  Comparing
$\mathsf{BuildBoat}$ against the two genuinely distinct
alternatives puts on record
\[
\mathsf{Capital} \;\revpref{\mathsf{C}}{1}\; \mathsf{Subsist},\quad
\mathsf{Capital} \;\revpref{\mathsf{C}}{1}\; \mathsf{ShoreCatch},
\]
hence, by \axref{O0},
$\mathsf{Capital} \pref{\mathsf{C}}{1} \mathsf{Subsist}$ and
$\mathsf{Capital} \pref{\mathsf{C}}{1} \mathsf{ShoreCatch}$.
The relative ranking between $\mathsf{Subsist}$ and
$\mathsf{ShoreCatch}$ is not on record.  Axiom~\axref{O2} requires
the primitive order to rank them nevertheless, so the model must
specify an extension of the record; we adopt one for concreteness:
\[
\mathsf{Capital} \pref{\mathsf{C}}{1}
\mathsf{ShoreCatch} \pref{\mathsf{C}}{1}
\mathsf{Subsist}.
\]

At $t=2$, Crusoe chooses $\mathsf{DeepSeaFish}$ while all four other
actions are available.  Three distinct alternative ends appear in
the foregone alternatives ($\mathsf{Subsist}$ via $\mathsf{Forage}$,
$\mathsf{Capital}$ via $\mathsf{BuildNet}$ or $\mathsf{BuildBoat}$,
$\mathsf{ShoreCatch}$ via $\mathsf{ShoreFish}$), so the choice
puts on record
\[
\mathsf{DeepCatch} \revpref{\mathsf{C}}{2} \mathsf{Subsist},\quad
\mathsf{DeepCatch} \revpref{\mathsf{C}}{2} \mathsf{Capital},\quad
\mathsf{DeepCatch} \revpref{\mathsf{C}}{2} \mathsf{ShoreCatch},
\]
and \axref{O0} carries each into the order.  The three-way
ranking among the unchosen ends is left open by the record; the
model adopts
\[
\mathsf{DeepCatch} \pref{\mathsf{C}}{2}
\mathsf{ShoreCatch} \pref{\mathsf{C}}{2}
\mathsf{Subsist} \pref{\mathsf{C}}{2}
\mathsf{Capital},
\]
expressing that at $t=2$ Crusoe values \emph{realized} catches more
than further capital deepening (he already has both pieces of
equipment).  Other extensions of the record are equally compatible
with the choice data --- the record fixes only the chosen end on
top and leaves the rest free, exactly as in
Remark~\ref{rem:acts_vs_pref}(ii).

Theorem~\ref{thm:opp_cost} guarantees that whenever an actor
chooses among multiple available actions, at least one alternative
end is forgone.  At $t=0$, by choosing $\mathsf{BuildNet}$, Crusoe
forgoes $\mathsf{Subsist}$ --- the only alternative available, and
hence the opportunity cost.  At $t=1$, by choosing
$\mathsf{BuildBoat}$, Crusoe forgoes the actions
$\mathsf{Forage}$, $\mathsf{BuildNet}$, and $\mathsf{ShoreFish}$;
the corresponding ends are $\mathsf{Subsist}$, $\mathsf{Capital}$,
and $\mathsf{ShoreCatch}$, but $\mathsf{Capital}$ is also the
realized end and is therefore not unmet.  The genuinely forgone
ends are $\mathsf{Subsist}$ and $\mathsf{ShoreCatch}$; the
opportunity cost is the highest-ranked among them, namely
$\mathsf{ShoreCatch}$ in our adopted extension.  At $t=2$, by
choosing $\mathsf{DeepSeaFish}$, Crusoe forgoes the four other
actions and hence the three distinct unmet ends $\mathsf{Subsist}$,
$\mathsf{Capital}$, $\mathsf{ShoreCatch}$; the opportunity cost in
our extension is $\mathsf{ShoreCatch}$.

Note that, as in any actor's case, identifying opportunity cost
requires the full preference ordering, including the ranking among
unchosen alternatives.  Where that ranking is not fully demonstrated by
observed choices --- as at $t=1$ and $t=2$ --- the opportunity cost
is not fully determined by behavioral data alone.  This is a
foretaste of the formal knowledge problem to be developed in future research.

\begin{remark}[Two readings of opportunity cost]
\label{rem:two_readings_opp_cost}
The means-ends multiplicity at $\mathsf{Capital}$ splits the
notion of ``what is foregone'' at each round into two
non-equivalent readings.  Tabulating the canonical trajectory:

\begin{center}
\renewcommand{\arraystretch}{1.25}
\begin{tabular}{cclllll}
\toprule
Round & At & Chosen & Realized end & Foregone actions & Foregone ends \\
& & & & (\emph{means lost}) & (\emph{wants unmet}) \\
\midrule
1 & $Q_0$ & $\mathsf{BN}$  & $\mathsf{K}$  & $\{\mathsf{Fo}\}$
              & $\{\mathsf{S}\}$ \\
2 & $Q_1$ & $\mathsf{BB}$  & $\mathsf{K}$  & $\{\mathsf{Fo}, \mathsf{BN}, \mathsf{SF}\}$
              & $\{\mathsf{S}, \mathsf{SC}\}$ \\
3 & $Q_2$ & $\mathsf{DSF}$ & $\mathsf{DC}$ & $\{\mathsf{Fo}, \mathsf{BN}, \mathsf{SF}, \mathsf{BB}\}$
              & $\{\mathsf{S}, \mathsf{K}, \mathsf{SC}\}$ \\
\bottomrule
\end{tabular}
\end{center}

\noindent
At Round~2 the actor performs $\mathsf{BB}$ and \emph{three}
actions are foregone, but only \emph{two} ends are unmet ---
because $\mathsf{BN}$, although unchosen, would have realized the
same end $\mathsf{K}$ that $\mathsf{BB}$ realizes.  The gap
$|\text{foregone actions}| - |\text{foregone ends}| = 3 - 2 = 1$
is exactly
the multiplicity at $\mathsf{Capital}$.  At Round~3 the gap is
again $4 - 3 = 1$, now for a different reason: $\mathsf{BN}$ and
$\mathsf{BB}$ are \emph{both} foregone and share the end
$\mathsf{K}$, so two foregone means collapse into one unmet
want.  Only Round~1 has a $0$ gap --- its single foregone action
carries its own end.  The two readings diverge \emph{precisely
when} distinct actions in the menu share an end: either a
foregone action shares the \emph{realized} end (Round~2), or two
foregone actions share a \emph{foregone} end (Round~3).  Both
readings are licensed by
Theorem~\ref{thm:opp_cost}; in the bijective case (Mises's
``each act has its end'' as a labeling convenience) they
collapse onto one another.
\end{remark}

On the snapshot one checks the per-time axioms directly:
\axref{P2} holds since each chosen action is listed as available at
the corresponding time; \axref{C1} since exactly one action is
chosen at each time; \axref{P5} since the three chosen actions are
distinct and occur at distinct times; and the adopted orders at
$t=0,1,2$ each extend the demonstrated record \axref{O0}, are total on
the end menu \axref{O2}, transitive \axref{O3}, and asymmetric
\axref{O4}, with the top-ranked end varying across periods
($\mathsf{Capital}$ at $t=0,1$, $\mathsf{DeepCatch}$ at $t=2$).  The
\emph{global} axioms hold over the whole $\mathbb{N}$-line rather
than on the snapshot alone: \axref{T6} outright, since $\mathbb{N}$
gives every moment a successor; and \axref{P6}, whose
unrealised employment the non-greedy continuation supplies
(a greedy ``deep-sea fishing forever'' history would leave the
boat-using action with no available-but-unperformed instance).  The
complete mechanical verification --- every axiom, on the
$\mathbb{N}$-time model --- is the \Lean{} development of
Appendix~\ref{sec:lean}.

\begin{remark}[Constancy versus consistency]\label{rem:constancy}
The model $\mathfrak{C}$ illustrates that the axiom system allows for
preference variation across time, and pins down what that means.  It
comes down to the reversal of a single pair:
$\mathsf{Capital} \pref{\mathsf{C}}{t} \mathsf{Subsist}$ at $t=0$ and
$t=1$, but $\mathsf{Subsist} \pref{\mathsf{C}}{2} \mathsf{Capital}$ at
$t=2$, where capital --- both pieces of equipment now in hand --- falls
to the bottom.  No single global ranking of the ends restricts to all
three moments, so no one scale (no global utility) underlies the
choices: the actor genuinely re-ranks.  This is not irrational --- it
reflects genuinely changed circumstances at each moment.

That the \emph{top} end also moves ($\mathsf{Capital}$ at $t=0,1$, then
$\mathsf{DeepCatch}$) is, by itself, no such sign: $\mathsf{DeepCatch}$
is merely absent from the menu until $t=2$, and the constant ranking
$\mathsf{DeepCatch} \succ \mathsf{Capital} \succ \mathsf{ShoreCatch}
\succ \mathsf{Subsist}$ would reproduce the same three choices, its top
changing only because a higher-ranked end enters the growing menu.  It
is the pair reversal, not the moving top, that no constant scale could
produce.  Axiom~\axref{O3} constrains only \emph{within-moment}
rankings; it says nothing across times.  Mises: ``Constancy and
rationality are entirely different notions'' \citep[p.~103]{Mises1949}.
\end{remark}

The state diagram of Figure~\ref{fig:crusoe_state_diagram} is
\emph{objective}: it records what is possible, not what is
chosen.  Different actors who share Crusoe's availability
structure but differ in their preference orders will choose
differently --- and inference from choice recovers the
record (the demonstrated core of each one's preference order) from
the path they trace through the diagram.
The four sample paths below illustrate the range.  All begin at
$Q_0$ and run for three rounds.

\medskip
\noindent
Each path is read as its own structure.  A path $P$ fixes the
availability relation --- $\Cmenu^{P}_t = \Cmenu_{Q(t)}$, the
state-menu the actor faces at the state it occupies each round,
read off the diagram exactly as in the figure --- and the actor's
choices along it, hence the demonstrated-preference record
$\revpref{\mathsf{C}}{t}$.  It does not, by itself, fix the model:
the primitive order $\pref{\mathsf{C}}{t}$ is free off the realized
end, so we obtain a model $\mathfrak{C}_P$ only by adopting, at
each round, a strict total order extending that round's record
(just as with the adopted orders above, \axref{O0}).  What all
four structures share is the \emph{state-indexed} menu structure
$\Cmenu_{Q_0}, \Cmenu_{Q_1}, \Cmenu_{Q_2}$ of the diagram,
together with $\Use$ and $\EndOf$; the canonical $\mathfrak{C}$ of
this section is the Path-A structure.  The records displayed below
are therefore read off each path's own choices, and every
comparison is taken against the menu that path actually faces ---
never against a higher state's menu the path has not reached.

\begin{center}
\renewcommand{\arraystretch}{1.25}
\begin{tabular}{cll}
\toprule
Path & Trajectory & Informal characterization \\
\midrule
A & $Q_0 \xrightarrow{\mathsf{BN}} Q_1 \xrightarrow{\mathsf{BB}} Q_2 \xrightarrow{\mathsf{DSF}} Q_2$
  & capital deepening \\
 &  & then realized consumption (canonical) \\
B & $Q_0 \xrightarrow{\mathsf{Fo}} Q_0 \xrightarrow{\mathsf{Fo}} Q_0 \xrightarrow{\mathsf{Fo}} Q_0$
  & pure subsistence; high time preference \\
C & $Q_0 \xrightarrow{\mathsf{BN}} Q_1 \xrightarrow{\mathsf{SF}} Q_1 \xrightarrow{\mathsf{SF}} Q_1$
  & coastal specialist; builds net, then fishes \\
D & $Q_0 \xrightarrow{\mathsf{BN}} Q_1 \xrightarrow{\mathsf{SF}} Q_1 \xrightarrow{\mathsf{BB}} Q_2$
  & alternator; \\
  &  & preferences reverse between rounds \\
\bottomrule
\end{tabular}
\end{center}

\noindent
For each path, comparisons between two actions that share an end
($\mathsf{BN}$ vs $\mathsf{BB}$ at $Q_1$, both ending at
$\mathsf{K}$) put nothing on record and are omitted; only
\emph{distinct-end} comparisons carry information.  The displays
list the instances that each path's choices put on record, carried
into the order $\pref{a}{t}$ by \axref{O0}.

\[
  \begin{aligned}
  t=0:\;& \mathsf{K} \revpref{\mathsf{C}}{0} \mathsf{S}\\
  t=1:\;& \mathsf{K} \revpref{\mathsf{C}}{1} \mathsf{S},\quad
          \mathsf{K} \revpref{\mathsf{C}}{1} \mathsf{SC}\\
  t=2:\;& \mathsf{DC} \revpref{\mathsf{C}}{2} \mathsf{S},\quad
          \mathsf{DC} \revpref{\mathsf{C}}{2} \mathsf{K},\quad
          \mathsf{DC} \revpref{\mathsf{C}}{2} \mathsf{SC}
  \end{aligned}
\]

An actor with high
time preference (or low marginal capital productivity) who never
finds capital formation worthwhile.
\[
  t = 0, 1, 2:\quad \mathsf{S} \revpref{a}{t} \mathsf{K}.
\]
The state never advances past $Q_0$; no capital accumulates and
no higher-order ends become reachable.

Builds the net at
$t=0$, then shore-fishes from $t=1$ onward.
\[
  \begin{aligned}
  t=0:\;& \mathsf{K} \revpref{a}{0} \mathsf{S}\\
  t=1:\;& \mathsf{SC} \revpref{a}{1} \mathsf{S},\quad
          \mathsf{SC} \revpref{a}{1} \mathsf{K}\\
  t=2:\;& \mathsf{SC} \revpref{a}{2} \mathsf{S},\quad
          \mathsf{SC} \revpref{a}{2} \mathsf{K}
  \end{aligned}
\]
The state never reaches $Q_2$; deep-sea catch is in principle
attainable one further capital-formation step away
($\mathsf{BB}$ would take the actor there) but the preferences do
not warrant the further investment.

Same as Path~C through
$t=1$, but at $t=2$ the actor returns to capital formation.
\[
  \begin{aligned}
  t=0:\;& \mathsf{K} \revpref{a}{0} \mathsf{S}\\
  t=1:\;& \mathsf{SC} \revpref{a}{1} \mathsf{S},\quad
          \mathsf{SC} \revpref{a}{1} \mathsf{K}\\
  t=2:\;& \mathsf{K} \revpref{a}{2} \mathsf{S},\quad
          \mathsf{K} \revpref{a}{2} \mathsf{SC}
  \end{aligned}
\]
The comparison $\mathsf{SC}$ vs $\mathsf{K}$ \emph{reverses}
between $t=1$ and $t=2$: $\mathsf{SC} \revpref{a}{1} \mathsf{K}$ but
$\mathsf{K} \revpref{a}{2} \mathsf{SC}$.  This is permitted by the
axiom system --- the record at each time lifts (by~\axref{O0}) into a
\emph{separate} strict total order $\pref{a}{t}$, and~\axref{O3}
places no across-time constraint.  The reversal is
\emph{consistent} (each $\pref{a}{t}$ is internally a strict total
order) but not \emph{constant} (the across-time direction
flips), exactly the distinction of
Remark~\ref{rem:constancy}.

\medskip
\noindent
What stays constant across paths: the sorts (the end set
$\mathsf{Ends}^{\mathfrak{C}}$ included), $\EndOf$, $\Use$, and the
state-indexed menus $\Cmenu_{Q_i}$ of the diagram --- the
objective availability structure.  What changes: $\Acts$,
$\pref{a}{t}$, and --- through the state the path occupies --- the
time-indexed menus $\Cmenu_t$ themselves.
The state diagram constrains \emph{what is possible}; the path
through it is the actor's subjective resolution of that structure
--- the observable trace from which the record, and through
\axref{O0} part of the $\pref{a}{t}$ history, is read off.

Holding the Crusoe ontology fixed (domains, $\EndOf$ assignments),
many distinct models satisfy the core axioms: different
availability structures, different preference orderings at each
time, different choice histories.  Praxeology constrains a
\emph{class} of structures, not a unique trajectory.  Economic
conclusions arise only when additional axioms restrict this class
further.  This is exactly the Hilbert/Tarski point: {\bf axioms
define a model class; theorems are invariants of that class.}

\subsection{The model in full: all six primitives over \texorpdfstring{$\mathbb{N}$}{N}}\label{sec:full_model}

The snapshot above fixes the primitives only on $\{0,1,2\}$ --- exhibiting
$\mathfrak{C}$ as a genuine model of $\Tprx$ requires fixing all
primitive relations at \emph{every} $t \in \mathbb{N}$.  We do so here;
the result is the model machine-checked in
Appendix~\ref{sec:lean}, of which the snapshot is the restriction
to the first three times $\{0,1,2\}$.
\begin{itemize}
\item \textbf{Time.} $\mathsf{Times}^{\mathfrak{C}} = \mathbb{N}$ with the
  standard strict order $<$.
\item \textbf{History ($\Acts$).}
  $\alpha_{\mathsf{C},0} = \mathsf{BuildNet}$,
  $\alpha_{\mathsf{C},1} = \mathsf{BuildBoat}$,
  $\alpha_{\mathsf{C},t} = \mathsf{DeepSeaFish}$ for even $t \geq 2$
  and $\mathsf{ShoreFish}$ for odd $t \geq 3$.  After the boat is
  built Crusoe \emph{alternates} deep-sea and shore fishing rather
  than always taking the deep-sea catch; this non-greedy
  alternation leaves an unrealised employment at each moment and so
  secures the free-good axiom \axref{P6} (a greedy history would
  violate it).
\item \textbf{Menus ($\Avail$).} $\mathsf{Forage}$ and
  $\mathsf{BuildNet}$ are available at every $t$; $\mathsf{ShoreFish}$
  and $\mathsf{BuildBoat}$ from $t \geq 1$; $\mathsf{DeepSeaFish}$
  from $t \geq 2$.  The menu is thus $\Cmenu_0$ at $t=0$, $\Cmenu_1$
  at $t=1$, and $\Cmenu_2$ for all $t \geq 2$ (the actor remains in
  epoch $Q_2$).
\item \textbf{Ends ($\EndOf$).} As assigned at the head of this
  section.
\item \textbf{Resources ($\Use$).} Things
  $\{\mathsf{Wood}, \mathsf{Net}, \mathsf{Boat}\}$:
  $\mathsf{BuildNet}$ and $\mathsf{BuildBoat}$ both use
  $\mathsf{Wood}$, $\mathsf{ShoreFish}$ uses $\mathsf{Net}$,
  $\mathsf{DeepSeaFish}$ uses $\mathsf{Net}$ and $\mathsf{Boat}$, and
  $\mathsf{Forage}$ uses nothing.  The shared use of $\mathsf{Wood}$
  by the two distinct, simultaneously available capital actions
  makes scarcity \axref{S1} hold at every $t \geq 1$.
\item \textbf{Scale of values ($\pref{\mathsf{C}}{t}$).} A strict
  total order on the four ends, co-varying with the history so that
  the realized end ranks on top --- which is what makes the record
  agree with the order, \axref{O0}:
  \[
  \begin{array}{ll}
  t = 0, 1: & \mathsf{Capital} \succ \mathsf{ShoreCatch} \succ
              \mathsf{Subsist} \succ \mathsf{DeepCatch};\\[2pt]
  \text{even } t \geq 2: & \mathsf{DeepCatch} \succ \mathsf{ShoreCatch}
              \succ \mathsf{Subsist} \succ \mathsf{Capital};\\[2pt]
  \text{odd } t \geq 3: & \mathsf{ShoreCatch} \succ \mathsf{DeepCatch}
              \succ \mathsf{Subsist} \succ \mathsf{Capital}.
  \end{array}
  \]
  The $t=4$ and $t=5$ rows differ --- the deep-sea/shore reversal of
  Remark~\ref{rem:constancy}, here machine-checked.
\end{itemize}

\noindent These data fix all six primitives at every
$t \in \mathbb{N}$.  One checks the base axioms directly:
\axref{T0}--\axref{T6} hold of $(\mathbb{N}, <)$;
\axref{P1}--\axref{P5} and \axref{C1} of the single-valued history;
\axref{P6} by the non-greedy alternation;
\axref{O0}--\axref{O4} because each $\pref{\mathsf{C}}{t}$ is a strict
total order topped by the realized end; and \axref{S1} by the shared
use of $\mathsf{Wood}$.  This is the constructive consistency proof
the type-checker carries out in Appendix~\ref{sec:lean}; the finite
snapshot $\{0,1,2\}$ is its restriction to the first three times --- a
structure, but not itself a model, since \axref{T6} admits no finite
models.


\section{Derived theorems of the base system}
\label{sec:derived_base}

The axiom system $\Tprx$ is not merely a language for stating
later results.  Like Hilbert's axioms for geometry, it also
yields the central propositions of the theory it formalizes.
In this section we derive the core of those propositions
--- the ones that follow from $\Tprx$ alone, or from $\Tprx$
together with explicit auxiliary premises that operate at the
same foundational layer.  Theorems requiring the production,
ownership, exchange, or monetary enrichment layers are deferred to
future work.

Each proof follows the Hilbert-style format: formal deductions
from explicitly cited axioms, with verbal commentary after each
step.  The shorter proofs appear in the main text; the proof of
Theorem~\ref{thm:DMU} (diminishing marginal utility), which
requires more detailed case analysis, is given in
Appendix~\ref{app:dmu_proof}.

Throughout this section, fix an actor $a$ and time $t$.  By axiom
\axref{O1}, there exists an action $\alpha_{a,t}$ such that
$\Acts(a,\alpha_{a,t},t)$.  By \axref{C1}, this action is unique.
We write $\EndAt(a,t) := \{E \mid \EndAt(a,t,E)\}$ for the
end menu at $(a,t)$, as defined in
Section~\ref{subsec:axioms}.

\subsection{Results from the core axiom system}
\label{subsec:derived_core}

The following results require only the base theory $\Tprx$
(Layers~1--2).

\begin{theorem}[Asymmetry of demonstrated preference]\label{thm:asymm}
$\forall a\,t\,E\,F\;(E\revpref{a}{t}F\Rightarrow\neg(F\revpref{a}{t}E))$.
\end{theorem}

If an actor's choice at $t$ puts $E$ over $F$ on record, no
choice at the same $t$ can put $F$ over $E$ on record.  This
follows directly from the action choice axiom \axref{C1}: the actor
performs one action at a time, so a single moment of conduct
cannot testify both ways.  The theorem concerns the record; for
the primitive order the corresponding property is
axiom~\axref{O4}.  The proof below establishes more than the
statement.  Every recorded pair at $(a,t)$ has the single
realized end on the left (by \axref{C1} and \axref{P4}), so the
record is a \emph{star}, and hence acyclic.  Jointly with the
grounding axiom~\axref{O0}, this means conduct can never force a
violation of~\axref{O4} --- the grounding is \emph{coherent}.

\begin{proof}\footnote{Throughout the paper proofs end with the
symbol $\square$ (equivalent to QED), which marks the formal
close of the argument --- the mathematical convention
\emph{quod erat demonstrandum}, ``what was to be shown.''}
Suppose $E \revpref{a}{t} F$ and $F \revpref{a}{t} E$.  By
Definition~\ref{def:revpref} the first record names a performed
action $\alpha$ with $\Acts(a,\alpha,t)$ and $\EndOf(\alpha,E)$,
the second a performed action $\gamma$ with $\Acts(a,\gamma,t)$
and $\EndOf(\gamma,F)$ --- each record carrying its realized end
on the left.  By \axref{C1} the action performed at $(a,t)$ is
unique, so $\alpha = \gamma$; by \axref{P4} its end is unique,
forcing $E = F$ --- contradicting the $E \neq F$ clause of
Definition~\ref{def:revpref}.
\end{proof}

\begin{proposition}[The chosen end tops the scale]
\label{prop:chosen_max}
Fix $(a,t)$.  Let $\alpha_{a,t}$ be the action performed at
$(a,t)$ (by \axref{O1} and \axref{C1}) and $\hat{E}$ its end
(by \axref{P3} and \axref{P4}); $\hat{E}$ is itself choice-relevant
by \axref{P2}.  Then $\hat{E}$ is the unique
$\pref{a}{t}$-maximum of the end menu:
\[
\forall F\;\bigl(\EndAt(a,t,F)\land F\neq\hat{E}
\;\Rightarrow\;\hat{E}\pref{a}{t}F\bigr).
\]
\end{proposition}

\begin{proof}
Let $F$ be choice-relevant with $F \neq \hat{E}$.  By
Definition~\ref{def:endat} there is $\beta$ with
$\Avail(a,\beta,t)$ and $\EndOf(\beta,F)$.  If
$\beta = \alpha_{a,t}$ then $F = \hat{E}$ by \axref{P4} ---
excluded.  So $\beta \neq \alpha_{a,t}$, the conditions of
Definition~\ref{def:revpref} are met, and
$\hat{E} \revpref{a}{t} F$; by \axref{O0},
$\hat{E} \pref{a}{t} F$.  Uniqueness follows from the
domination just established together with \axref{O4}: a second
maximum $G \neq \hat{E}$ would give $G \pref{a}{t} \hat{E}$ and
$\hat{E} \pref{a}{t} G$.

The performed action's end outranks every other end on the
menu: the actor does what he most prefers.  This is Mises's
``every action is always in perfect agreement with the scale of
values'' \citep[Ch.~IV, \S2]{Mises1949}, recovered as a
consequence of the grounding axiom rather than packed into the
word ``rational.''  In this sense the proposition is an \emph{adequacy check} on
\axref{O0}, not an independent discovery: it confirms that
grounding makes the primitive scale cohere with conduct.  That
this coherence is always achievable --- no Layer-1 structure is
excluded --- is the satisfiability established in
Theorem~\ref{thm:asymm}.
\end{proof}

\begin{theorem}[Existence of opportunity cost]\label{thm:opp_cost}
If there exist two available actions at $(a,t)$ with distinct ends, then the
chosen action --- which exists and is unique by \axref{O1}
and~\axref{C1} --- excludes at least one alternative end.\footnote{The opportunity-cost
concept was first articulated as ``Kosten als entgangener Nutzen''
(cost as foregone utility) by Friedrich von
\citet{Wieser1889}; the proposition formalizes Wieser's
insight in the present axiomatic setting.  The two-readings
refinement of Remark~\ref{rem:two_readings_opp_cost} extends
Wieser's framework once means-ends multiplicity is admitted.}
\end{theorem}

Scarcity of time is structural:
whenever multiple ends are choice-relevant at a moment, pursuing one forecloses
the others. Opportunity cost is not an additional assumption --- it is a theorem
of the action structure.  This is Mises' point that action always
involves renunciation: ``Action is an attempt to substitute a more
satisfactory state of affairs for a less satisfactory one''
\citep[Ch.~IV, \S4]{Mises1949}.

\begin{proof}
Assume actions $\alpha,\beta$ with
$\Avail(a,\alpha,t)$, $\Avail(a,\beta,t)$,
$\EndOf(\alpha,E)$, $\EndOf(\beta,F)$, and $E \neq F$.
(By \axref{P3} every available action has an end; by hypothesis
two such ends differ.)

We assume a genuine choice situation: at least two ends are
choice-relevant at the same time, i.e.\ $|\EndAt(a,t)| \geq 2$.

\medskip

By \axref{O1} and \axref{C1}, there exists a unique $\alpha_{a,t}$ with
$\Acts(a,\alpha_{a,t},t)$; and by~\axref{P4} this performed action
has a unique end.  Hence at most one of $E, F$ coincides with
$\EndOf(\alpha_{a,t})$; since $E \neq F$, at least one of $E,F$ is
then available but unrealised --- the foregone end the proposition
asserts.
\end{proof}

\begin{definition}[Opportunity cost]\label{def:opp_cost}
Let $\hat{E} = \EndOf(\alpha_{a,t})$ be the realized end at $(a,t)$
(Proposition~\ref{prop:chosen_max}).  The \emph{opportunity cost} is the most-preferred foregone end
--- the $\pref{a}{t}$-greatest end the actor did not realize:
\[
\mathrm{Cost}(a,t) \;:=\; \max\nolimits_{\pref{a}{t}}
\bigl(\EndAt(a,t) \setminus \{\hat{E}\}\bigr),
\]
defined when $\lvert\EndAt(a,t)\rvert \geq 2$ (some end is foregone).
\end{definition}

\begin{corollary}[Existence, uniqueness, identification]
\label{cor:opp_cost}
If at least two ends are choice-relevant at $(a,t)$
($\lvert\EndAt(a,t)\rvert \geq 2$), then $\mathrm{Cost}(a,t)$ exists
and is unique, and it is the \emph{second-ranked} end of the menu
--- the end immediately below $\hat{E}$ in $\pref{a}{t}$.
\end{corollary}

\begin{proof}
$\EndAt(a,t) \setminus \{\hat{E}\}$ is non-empty (as
$\lvert\EndAt(a,t)\rvert \geq 2$) and finite
(Remark~\ref{rem:finiteness}), and $\pref{a}{t}$ is a strict total
order on $\EndAt(a,t)$ by \axref{O2}--\axref{O4}; a non-empty finite
subset of a strict total order has a unique maximum, so
$\mathrm{Cost}(a,t)$ exists and is unique.  By
Proposition~\ref{prop:chosen_max}, $\hat{E}$ is the maximum of the
whole menu; deleting it promotes the second-ranked end to the
maximum of the remainder, which is therefore $\mathrm{Cost}(a,t)$.
\end{proof}

\noindent Theorem~\ref{thm:opp_cost} establishes that
\emph{something} is foregone; Definition~\ref{def:opp_cost}
identifies \emph{which} foregone end carries the valuation.  This is
Wieser's \emph{Kosten als entgangener Nutzen} --- cost as the value
of the best foregone alternative \citep{Wieser1889} --- now stated
as a specific menu position (the second-ranked end) rather than the
bare existence of an unmet end.  ($\mathrm{Cost}$ answers the \emph{wants-unmet} reading; which of
several same-end means is given up is the separate
\emph{means-lost} reading of
Remark~\ref{rem:two_readings_opp_cost}.)

\medskip

Two further results follow from the same axioms.

\begin{theorem}[Scarcity of time is structural]
\label{thm:time_scarcity}
If there exist two distinct available actions at $(a,t)$, then
time at $t$ is scarce for $a$: the temporal instant cannot
accommodate both actions.\footnote{The treatment of time as an
economic primitive --- and the analysis of capital structure built
on it --- originates with Eugen von B\"ohm-Bawerk's \emph{Positive
Theory of Capital}~\citep{BohmBawerk1889}.  The theorem here
extracts the temporal-scarcity claim from a single layer of the
axiom system; the capital-structure content of B\"ohm-Bawerk's
analysis enters only with the production enrichment of
future research.}
\end{theorem}

\begin{proof}
Let $\alpha \neq \beta$ be two distinct available actions at
$(a,t)$.  By~\axref{C1} at most one is performed, so at least one
--- say $\beta$, with end $F$ by~\axref{P3} --- remains unrealised
at $t$.  Its end is feasible yet foreclosed, not by any scarcity
of goods but purely by the uniqueness of action at a
time.\footnote{The claim that only one action can occupy an
instant invites the obvious
objection that a person can manifestly walk, talk, look around,
and chew gum at the same time, so the temporal instant
\emph{does} seem to accommodate concurrent activity.  The
resolution lies in the praxeological notion of a unitary action,
made explicit in Remark~\ref{rem:composite_actions}: the
Actions sort ranges over \emph{purposeful units of conduct},
each directed at a single end (by~\axref{P3}+\axref{P4}).  A coordinated
sequence of bodily movements subordinated to one chosen end
counts here as one action, however internally complex; what \axref{C1}
excludes is two \emph{distinct} purposeful units, with
\emph{different} ends, by the same actor at the same time.
``Walk and chew gum'' typically falls under one purposeful
unit (e.g.\ ``cross the room while staying alert''); two
genuinely competing ends with conflicting means
(``go left'' \emph{and} ``go right'' simultaneously) is what
\axref{C1} rules out.} Time itself is thus the scarce means.
\end{proof}

\begin{remark}[Misesian interpretation of temporal scarcity]
\label{rem:time_scarcity}
Even in a world without material scarcity, time would remain
scarce.  Whenever multiple ends are choice-relevant at a given
instant, the actor must select one action and thereby forgo
others.  Temporal scarcity is thus a structural consequence of
the action axiom together with the non-simultaneity of
individual action.  This is Mises's point in his treatment of
``the economization of time'': ``his time is scarce.  He must
economize it as he does other scarce factors''
\citep[Ch.~V, \S3]{Mises1949}; and ``even in the land of Cockaigne
man would be forced to economize time'' \citep[Ch.~V, \S3]{Mises1949}.
Rothbard's \emph{Man, Economy, and State}~\citep[ch.~1,
p.~5]{Rothbard2009} reaches the same conclusion by a slightly
different route, anchored more directly in the action choice axiom:
``A man's time is always scarce.  He is not immortal\dots\
Action takes place by choosing which ends shall be satisfied
by the employment of means.  Time is scarce for man only
because whichever ends he chooses to satisfy, there are others
that must remain unsatisfied.''  The structural character of the
scarcity is exactly what Theorem~\ref{thm:time_scarcity}
records: it follows from \axref{C1} and~\axref{P3} alone, with no
appeal to material scarcity or even to the time-order $<$.  The
actor performs at most one action at $(a,t)$, so a second
choice-relevant end must go unmet.

The Misesian and Rothbardian formulations agree on what time
scarcity \emph{is}: a structural \emph{choice-among-ends}
constraint at a single instant, not a claim that ``every
action uses $t' - t$ units of time'' for some non-zero
duration.  The duration-of-an-action reading is a separate
matter: \citet[ch.~1, pp.~13--16]{Rothbard2009} distinguishes
within any action the \emph{period of production}, the
\emph{duration of serviceability}, and the \emph{period of
provision}, each finer than the bare $<$-order on
$\mathsf{Times}$ that the base layer carries; all three are
deferred to the production and time-preference enrichments of
future research.
\end{remark}

\begin{theorem}[Subjectivity of opportunity cost --- a metatheorem]
\label{thm:subjectivity}
Opportunity cost depends on the actor's preference order
$\pref{a}{t}$ --- accessed only through choice --- not
on any physical properties of actions or things.
\end{theorem}

Unlike the preceding results, this is a statement \emph{about}
the theory $\Tprx$ rather than a formula of $\Lprx$ derived
inside it: its content is that the axiom list contains no
postulate linking $\pref{a}{t}$ to a physical magnitude, and
the proof accordingly proceeds by inspection of the axioms ---
the standard form of a \emph{metatheorem}.

\begin{proof}
Consider $(a,t)$ with at least two choice-relevant ends.  By
Theorem~\ref{thm:opp_cost} pick an available, unrealised action
$\beta$ with end $F$ distinct from the performed end
$E = \EndOf(\alpha_{a,t})$ (which exists by~\axref{P3}
and~\axref{P4}).  The conditions of Definition~\ref{def:revpref}
then hold, so $E \revpref{a}{t} F$, and $E \pref{a}{t} F$
by~\axref{O0}.

No axiom of $\Tprx$ links $\pref{a}{t}$ to a physical metric: its
only connection is to conduct, through~\axref{O0}, while~\axref{P6}
and the ordinal time axioms~\axref{T1}--\axref{T6} introduce no
numerical measure.  Opportunity cost is therefore fixed by the
actor's ordinal comparisons alone, accessed only through choice
--- Mises' point that ``costs are equal to the value attached to
the satisfaction which one must forego in order to attain the end
aimed at'' \citep[p.~97]{Mises1949}.
\end{proof}

\begin{corollary}[Subjectivity, model-theoretically]
\label{cor:subjectivity_model}
There are two models $\mathfrak{M}_1, \mathfrak{M}_2$ of $\Tprx$
with \emph{identical} Layer-1 parts --- the same domains and the
same $\Avail, \Acts, \EndOf, \Use, <$ --- in which the opportunity
cost differs:
$\mathrm{Cost}^{\mathfrak{M}_1}(a,t) \neq
\mathrm{Cost}^{\mathfrak{M}_2}(a,t)$.
\end{corollary}

\begin{proof}
Take Crusoe at $t=1$ (Section~\ref{subsec:crusoe}): the
choice-relevant ends are
$\{\mathsf{S}, \mathsf{K}, \mathsf{SC}\}$ and the realized end is
$\mathsf{K}$.  The record pins only the chosen end on top ---
$\mathsf{K} \revpref{\mathsf{C}}{1} \mathsf{S}$ and
$\mathsf{K} \revpref{\mathsf{C}}{1} \mathsf{SC}$, the star property
of Remark~\ref{rem:acts_vs_pref} --- and says nothing about
$\mathsf{S}$ versus $\mathsf{SC}$.  Both completions
\[
\mathsf{K} \succ \mathsf{SC} \succ \mathsf{S}
\qquad\text{and}\qquad
\mathsf{K} \succ \mathsf{S} \succ \mathsf{SC}
\]
extend the record and satisfy \axref{O0}--\axref{O4}, yielding
models $\mathfrak{M}_1$ and $\mathfrak{M}_2$ that agree on every
Layer-1 primitive --- both are the canonical $\mathbb{N}$-time
Crusoe of Section~\ref{subsec:crusoe}, identical in every domain
and in $\Avail, \Acts, \EndOf, \Use, <$ (and in $\pref{a}{t}$ at
every time other than $t=1$); only the free $\mathsf{S}$-vs-$\mathsf{SC}$
order at $t=1$ is varied.  But the most-preferred foregone end is
$\mathsf{SC}$ in $\mathfrak{M}_1$ and $\mathsf{S}$ in
$\mathfrak{M}_2$: $\mathrm{Cost}(\mathsf{C},1)$ differs across two
models with identical observable structure.
\end{proof}

\noindent The example needs at least three choice-relevant ends:
with two, the star record already pins the whole order (the
``static fully-observed'' regime of
Remark~\ref{rem:acts_vs_pref}), and no freedom remains.  Where the
metatheorem says the language carries no physical handle on
opportunity cost, the corollary makes the dependence operational:
hold everything observable fixed, vary only the hidden scale
$\pref{a}{t}$, and the cost itself moves --- subjectivity in the
model-theoretic sense the term usually carries.

\begin{crusoe}[scarcity of time and opportunity cost]
These results were already visible in
$\mathfrak{C}$ (Section~\ref{subsec:crusoe}).  The formal
theorems confirm what the example exhibited:

\medskip
\begin{center}
\renewcommand{\arraystretch}{1.3}
\begin{tabular}{l>{\raggedright\arraybackslash}p{8cm}}
\toprule
\textbf{Result} & \textbf{Crusoe instance} \\
\midrule
Temporal scarcity &
  At each $t$, Crusoe performs exactly one action.
  At $t = 0$ he builds the net and thereby forgoes
  foraging --- the temporal instant cannot accommodate
  two distinct actions of one actor. \\
Opportunity cost &
  At $t = 0$, the opportunity cost of $\mathsf{BuildNet}$ is
  $\mathsf{Subsist}$, the only alternative end available
  in $\Cmenu_0$. \\
Subjectivity &
  The cost is determined entirely by Crusoe's ranking
  $\mathsf{Capital} \pref{\mathsf{C}}{0} \mathsf{Subsist}$,
  not by any physical property of nets, fish, or
  foraged plants. \\
\bottomrule
\end{tabular}
\end{center}
\end{crusoe}

\subsection{Diminishing marginal utility}
\label{subsec:derived_dmu}

The law of diminishing marginal utility is frequently presented
as a direct consequence of purposeful action allocating scarce
means to subjectively ranked ends.  In a Hilbert-style framework,
however, this conclusion does \emph{not} follow from the
praxeological core alone.  The purpose of this subsection is to
make explicit the additional ontological and behavioral structure
required for a valid deduction --- which is, in keeping with the
methodological frame of Section~\ref{subsec:mises_programme},
nothing other than \emph{the definition of what it means to
allocate rationally}.

Six facets enter the deduction.  Two structural axioms govern
the schedule itself, (MU0) \emph{functionality} and (MU1)
\emph{feasibility}; the strict total order on $\EndAt(a,t)$ is
already supplied by \axref{O2}--\axref{O4} of Layer~2
(\emph{preference}); and three further axioms complete it ---
(MU2) \emph{fungibility}, (MU3) \emph{scarcity awareness}, and
(MU4) \emph{coherence of the served set}.  Each is defined
below, and none is optional: removing any one breaks the
framework in a different way (Remark~\ref{rem:dmu_method}).
Since the preference facet is carried by Layer~2, the
(MU)-enumeration is (MU0)--(MU4), and diminishing marginal
utility is then a structural theorem about \emph{any} rational
allocation.

\subsubsection*{Additional primitives}

In addition to the praxeological primitives $\Acts$, $\Avail$,
$\EndOf$, $\Use$, $\pref{a}{t}$, we introduce a partition of the
Things sort
into homogeneous goods and an \emph{allocation} relation
distinct from action.

Introduce a sort $\mathsf{Goods}$ (variables $g,h,\ldots$) and
two predicates
\begin{align*}
\UnitOf(x,g)\quad &\text{(``thing $x$ is a unit of good $g$''),}\\
\Allot(a,t,x,E)\quad &\text{(``actor $a$ at $t$ allots unit $x$ to end $E$'').}
\end{align*}

A ``good'' is a class of interchangeable units.  A \emph{unit}
here is Rothbard's \emph{technological unit} --- the smallest
quantity of the good that enters action
\citep[ch.~1, p.~28]{Rothbard2009}: a pound of rice, not a single
grain; the thirty logs a cabin takes, not one log.  Taken at that
grain the units of a good are genuinely interchangeable, and a
want calls for whole units of it.  The partition into goods is
itself subjective, determined by the actor's appraisal of
serviceability.  (Units are \emph{things}, not a separate
``means'' kind --- the means role is $\Use(\alpha,x)$ in the
core.)

Why a separate $\Allot$ predicate?
Allocation is the actor's \emph{valuation} act, distinct from
performance: Carl Menger's table of uses~\citep[ch.~III]{Menger1871},
a stock ranged across ends in order of urgency.  We can now say
exactly what that valuation is.  The schedule $\Allot$ sends each
committed unit to an end; the scale of values $\pref{a}{t}$ ranks
those ends.  The valuation at $(a,t)$ is their composite --- each
allotted unit valued by the rank of the end it secures in the end
menu $\EndAt(a,t)$, ordered by $\pref{a}{t}$:
\[
x \;\xrightarrow{\ \Allot\ }\; E \;\xrightarrow{\ \pref{a}{t}\ }\;
\text{rank of $E$ in } \EndAt(a,t).
\]
By \axref{MU0} the first arrow is single-valued, and with
technological units it is a bijection of the allotted units onto
the served set $\Served(a,t,g)$ (the footnote to \axref{MU0}), so
each served end is secured by exactly one unit.  The good's own
\emph{marginal utility} is the stock-relative bottom rung of this schedule
--- the marginal end (Definition~\ref{def:value}), carried by the marginal unit.

This valuation is a schedule, not a temporal sequence of action
tokens: the first unit of \good{water} to $\mathsf{Drink}$, the second to
$\mathsf{Cook}$, the third to $\mathsf{Wash}$ --- a stock allotted
across ends at one moment.  The action choice axiom \axref{C1} says
that at a given moment $t$ exactly one action is performed; it
places no constraint on how many units the actor has, at that same
moment, allotted across his ends.  The performed action realizes
one cell of the schedule; the rest stand pending.  So DMU is a
theorem about \emph{valuations} --- multi-end schedules under
\axref{MU0}--\axref{MU4} --- not about the singleton performance
governed by \axref{C1}.

\subsubsection*{Derived notions}
\begin{definition}[Serviceable end]\label{def:serviceable}
End $E$ is \emph{serviceable by good $g$ at $(a,t)$} iff some
available action with end $E$ uses some unit of $g$:\footnote{``Uses''
is $\Use(\alpha,x)$ in the sense of \emph{employment}, not
\emph{consumption}: it records that $\alpha$ \emph{employs} the unit
$x$, with no claim that $x$ is used up.  $\Use$ carries no time argument
(the static relation of Remarks~\ref{rem:means_via_use}
and~\ref{rem:avail_primitive}), so the action serving the actor's top
end, though it employs a unit of $g$, does \emph{not}, at this layer,
deplete his stock or retire any end from the menu.  Whether employment
\emph{consumes} a unit (a swallow of water) or \emph{leaves it intact}
(a durable tool --- Crusoe's boat) is the consumable-versus-durable
distinction of Mises and Rothbard, formalized by the $\Consumable$ and
$\Result$ predicates of future research; the resulting stock-dynamic
($n \to n-1$, a satisfied want dropping off) belongs to the transition
map, not to the momentary valuation here.  This is why serving the
\emph{top} end never collides with the marginal end's sitting at the
\emph{bottom} (Remark~\ref{rem:valuation_vs_performance}).}
\[
\Serviceable(a,t,g,E) :\Leftrightarrow
\exists\alpha\,\exists x\;
\bigl(\Avail(a,\alpha,t) \land \EndOf(\alpha,E) \land
\Use(\alpha,x) \land \UnitOf(x,g)\bigr).
\]
We write 
\[
\Serviceable(a,t,g) := \{E\in\mathsf{Ends} \mid \Serviceable(a,t,g,E)\}
\]
for the set of $g$-serviceable ends (the predicate name with one fewer
argument denotes the set).
\end{definition}
\begin{definition}[Served end]\label{def:served}
End $E$ is \emph{served by good $g$ at $(a,t)$} iff some unit of
$g$ has been allotted to $E$ at $(a,t)$:
\[
\Served(a,t,g,E) :\Leftrightarrow
\exists x\;\bigl(\UnitOf(x,g) \land \Allot(a,t,x,E)\bigr).
\]
Write $\Served(a,t,g) := \{E\in\mathsf{Ends} \mid \Served(a,t,g,E)\}$
for the set of ends $g$ serves.
\end{definition}
$\Serviceable$ ties feasibility to the menu of \emph{available}
uses; ``served'' ties it to the actor's \emph{allocation
schedule}, distinct from current performance.  At any single $t$
the allocation may serve many ends; the action $\Acts(a, \cdot, t)$
realizes one of these allotments while the rest remain pending
in the schedule.

\begin{definition}[Reduced served set]\label{def:reduced_served_set}
For a fixed unit $y$, the \emph{reduced served set} (the ``without
$y$'' served set) is the predicate
\[
\Served^{\setminus y}(a,t,g,F) :\Leftrightarrow
\exists x\;\bigl(\UnitOf(x,g) \land x \neq y \land
                  \Allot(a,t,x,F)\bigr),
\]
with set form $\Served^{\setminus y}(a,t,g) := \{F\in\mathsf{Ends} \mid
\Served^{\setminus y}(a,t,g,F)\}$.
This is the served set obtained by deleting $y$'s allotment
from the schedule; the superscript ``$\setminus y$'' is
set-difference notation, read ``with $y$ removed.''  In
Mises's ``$n$ vs.\ $n-1$ units''
language~\citep[Ch.~VII, \S1]{Mises1949}, $\Served$ is the
$n$-unit served set and $\Served^{\setminus y}$ is the
$n-1$-unit served set.
\end{definition}

\begin{definition}[Marginal end and marginal unit]\label{def:marginal_end}
Fix $a,t,g$.  The \emph{marginal end} is the $\pref{a}{t}$-least
served end,
\[
E \;:=\; \min\nolimits_{\pref{a}{t}} \Served(a,t,g),
\]
defined whenever $\Served(a,t,g) \neq \emptyset$.  Each unit being
a technological unit, the marginal end is served by a single
unit --- the unique $y$ with $\Allot(a,t,y,E)$ --- the
\emph{marginal unit}: Rothbard's ``one unit \ldots\ he must
consider giving up'' \citep[ch.~1]{Rothbard2009}, the one whose
withdrawal forfeits $E$.  After it is removed, the \emph{reduced
marginal end} is the $\pref{a}{t}$-least end still served,
\[
E^* \;:=\; \min\nolimits_{\pref{a}{t}} \Served^{\setminus y}(a,t,g).
\]
Each exists and is unique whenever its served set is non-empty
(Lemma~\ref{lem:served_choice_relevant}).\footnote{Finiteness is
essential: a non-empty \emph{totally ordered} set need not have a
least element --- e.g.\ the open interval $(0,1)$ --- whereas a
non-empty \emph{finite} subset of a total order always does.}
\end{definition}

\begin{definition}[Marginal utility of a good]\label{def:value}
Fix $a,t,g$ with $\Served(a,t,g) \neq \emptyset$.  The \emph{marginal
utility of good $g$ to $a$ at $t$} is the end served by its marginal
unit --- the marginal end (Definition~\ref{def:marginal_end}):
\[
\MU(a,t,g) \;:=\; \min\nolimits_{\pref{a}{t}} \Served(a,t,g),
\]
the lowest-ranked end the current stock serves --- the bottom rung
of $g$'s ladder of served ends.  Goods are compared by their
marginal utilities on the one scale $\pref{a}{t}$: $g$ has higher
marginal utility than $g'$ iff $\MU(a,t,g) \pref{a}{t} \MU(a,t,g')$.
It is undefined when $g$ serves no end --- a free good, carrying no
marginal utility.\footnote{This is Mises's own term: the unit of a
homogeneous supply that a man would forgo first is the \emph{marginal}
employment, ``and the utility derived from it marginal utility''
\citep[Ch.~VII, \S1]{Mises1949}.  In strict Mengerian usage a good's
\emph{value} is this same marginal utility; we keep ``value'' for the
looser pre-marginalist sense (the \emph{importance} of the kind, Smith's
\emph{value in use}; Remark~\ref{rem:value_paradox}) and reserve
$\MU(a,t,g)$ for the formal object, so the two never share a word.}
\end{definition}

\begin{remark}[The value paradox]\label{rem:value_paradox}
Pinning value to \emph{units} this way dissolves the classical value
paradox, which runs together two senses of ``value'' the apparatus
keeps apart.  One is the \emph{importance of the kind} --- Smith's
\emph{value in use} --- and it shows in how high a good's
\emph{serviceable} ends reach: \good{bread} can serve the most urgent
of wants, staying alive, where \good{platinum} serves only ornamental
ones, so \good{bread} reaches far higher in $\Serviceable(a,t,g)$.  The
other is what a single \emph{unit} is worth, the good's \emph{marginal
utility} --- its marginal served end, the bottom of its ladder,
$\MU(a,t,g) = \min\nolimits_{\pref{a}{t}}\Served(a,t,g)$
(Definition~\ref{def:value}).  These are two different ends --- the top
of $\Serviceable(a,t,g)$ and the bottom of $\Served(a,t,g)$ --- and the
paradox is only the surprise that they come apart.

Theorem~\ref{thm:DMU} is what drives them apart: each further unit
pushes the marginal end lower, so a good abundant relative to its uses
has a low marginal utility \emph{from abundance, not from unimportance
of the kind}.  \good{Bread}, abundant, has its margin driven down to a
trivial want, below the ends to which scarce \good{platinum} is
confined; so the kind that reaches the highest want has the
\emph{least} marginal utility --- \good{bread} the more important kind,
an ounce of \good{platinum} the dearer unit.  Only this marginal end is
formalized: by type discipline $\pref{a}{t}$ ranks \emph{ends} and
$\Allot$ assigns \emph{units}, a good enters only as a classification
of units (via $\UnitOf$), and no cardinal total of the wants a kind
could meet is available, so the importance of \good{bread}-as-a-kind
names no object in $\Lprx$ --- only $\MU(a,t,g)$ does.  Cross-good
comparison is well-formed exactly where Rothbard places it --- between
the goods' marginal utilities on the single scale of
\axref{O2}--\axref{O4}: ``whether, given the present available
stock\dots, a `loaf of bread' is more or less valuable\dots\ than `an
ounce of platinum'\,''~\citep[ch.~1, pp.~21--22]{Rothbard2009}.
\end{remark}

\subsubsection*{Additional axioms (the hidden premises)}

To derive diminishing marginal utility, we require five axioms
beyond the praxeological core.  These are \emph{not} part of
$\Tprx$; they encode the implicit assumptions in the standard
argument.  They speak of the actor's allocation \emph{schedule}
at $(a,t)$ --- the set of pairs $\{(x,E):\Allot(a,t,x,E)\}$, the
graph of $\Allot$ restricted to $(a,t)$ --- whose entries are its
\emph{cells} (single allotments); the schedule need not be total.

\begin{enumerate}[label=(MU\arabic*), leftmargin=3em, start=0]
\item\label{ax:MU0} \textbf{Functionality of the allocation schedule.}
For each $(a,t)$, the actor allots each unit $x$ of any good to
at most one end:
\[
\forall a\,t\,x\,E\,F\;\bigl(
  \Allot(a,t,x,E) \land \Allot(a,t,x,F) \,\Rightarrow\, E = F\bigr).
\]

\emph{Each unit is allotted to at most one end: the schedule is a
partial function from units to ends, though several units may
share one end (several cells with the same end).}\footnote{That is
\axref{MU0} on its own --- single-valued but not injective.  The
technological unit adds the injectivity (each end served by a
single unit), so at fixed $(a,t)$ the schedule is a bijection from
the allotted units onto $\Served(a,t,g)$, with the marginal unit
(Definition~\ref{def:marginal_end}) the unique unit allotted to
the marginal end; by \axref{MU3}, $\Served(a,t,g)$ is a proper
subset of the serviceable ends.}

\item\label{ax:MU1} \textbf{Feasibility of the allocation schedule.}
A unit is allotted only to an end its good can serve:
\[
\forall a\,t\,x\,g\,E\;\bigl(
  \UnitOf(x,g) \land \Allot(a,t,x,E) \,\Rightarrow\,
  \Serviceable(a,t,g,E)\bigr).
\]

\emph{Only feasible cells are populated: a unit of good $g$ is
allotted to end $E$ only when $g$ is serviceable for $E$
(Definition~\ref{def:serviceable}).  This is the bridge tying the
allocation schedule to the availability structure --- $\Avail$,
$\EndOf$, $\Use$ --- so that a cell is not a free-floating pairing
but one underwritten by an actually serviceable configuration.}

\item\label{ax:MU2} \textbf{Homogeneity of units (menu-level).}\footnote{There
is a longstanding controversy in Austrian circles regarding the
status of \emph{homogeneity} as we use it here.  Austrians
maintain that indifference is incompatible with human action ---
action absolutely requires preference: one option must be
preferred and the others rejected.  Yet the school does not
reject supply-and-demand curves; \citet{Mises1949} does not
employ them, but Rothbard and many others do.
\citet{Nozick1977} argues that Austrians cannot have it both
ways: the construction of a supply curve treats each part of the
curve as homogeneous, and homogeneity at the unit level seems
indistinguishable from indifference between units.
\citet{Block1980} offers a rejoinder: \emph{before} a choice is
made, one may be indifferent between, say, two physically
identical cans of Coca Cola; once one of them is selected the
chooser is no longer indifferent --- he chose \emph{that} one,
did he not?  The controversy has produced an extensive secondary
literature
\citep{Barnett2003,Block1999,Caplan1999,Hoppe2005,Machaj2007,ONeill2010,Wysocki2016}
and we do not adjudicate it here.  For our purposes \axref{MU2} is a
\emph{menu-level} statement: it says that if some unit of $g$
can serve end $E$ via an available action, then so can any
other unit of $g$; it does \emph{not} say the actor is
indifferent between two specific units once one has been
selected for use.  Still, \axref{MU2} is not a mere bookkeeping
convention: it takes a substantive position on which
unit-aggregates count as the same good, and the indifference
controversy just cited is precisely a dispute over when that
position is tenable.
}
\[
\begin{split}
\forall a\,t\,g\,E\,x\,y\;\Bigl(
&\UnitOf(x,g) \land \UnitOf(y,g) \land
\exists\alpha\bigl(\Avail(a,\alpha,t) \land
\EndOf(\alpha,E) \land \Use(\alpha,x)\bigr) \\
&\Rightarrow \exists\beta\bigl(\Avail(a,\beta,t)
\land \EndOf(\beta,E) \land \Use(\beta,y)\bigr)
\Bigr).
\end{split}
\]

\emph{If some unit of $g$ can be used (via some available
action) to serve $E$ at $(a,t)$, then any other unit of $g$
can also serve $E$.  Units of a homogeneous good are
interchangeable in the choice menu.}\footnote{Rothbard makes
the parallel point verbally for the cross-good setting: ``in
the act of choosing between giving up or adding units of either
$X$ or $Y$, the actor must have, in effect, placed both goods
on a single, unitary value scale\ldots\ the very fact of action
in choosing between more than one good implies that the units
of these goods must have been ranked for comparison on one
value scale of the actor''
\citep[ch.~1, pp.~31--32]{Rothbard2009}.  In our setup, this is
exactly the content of \axref{O2}--\axref{O4}: the strict total
order on $\EndAt(a,t)$ ranges over ends regardless of which
good serves them, so cross-good comparison is automatic.
\axref{MU2} governs the within-good homogeneity; the cross-good
unified scale is supplied by the Layer~2 preference order.}

\item\label{ax:MU3} \textbf{Scarcity (cardinal-free).}
Some good is genuinely scarce: at some $(a,t)$ a good has a
\emph{serviceable} end its stock leaves unserved:
\[
\exists a\,t\,g\,F\;\bigl(\Serviceable(a,t,g,F) \;\land\;
\neg\Served(a,t,g,F)\bigr).
\]

\emph{The (MU)-level counterpart of the base scarcity axiom
\axref{S1}: some good cannot serve every end within its
reach.}\footnote{Quantifying over the ends the good \emph{can}
serve --- not over every wanted end --- is what makes this
genuine scarcity, not merely a good's inability to meet some end.
And since only \emph{some} good need be scarce, others may not be:
a good that already serves every end it could (like \good{fish} in the
Crusoe box below) is simply not scarce --- Menger's line between
economic and non-economic goods.}

\item\label{ax:MU4} \textbf{Preference-respecting allocation (top-segment).}
If $E$ is served by $g$ at $(a,t)$, then every strictly
higher-ranked choice-relevant end \emph{that $g$ can serve} is
also served by $g$:
\[
\begin{split}
\Served(a,t,g,E) \;\Rightarrow\;
\forall F\;\Bigl(&\EndAt(a,t,F) \land
\Serviceable(a,t,g,F) \land
F \pref{a}{t} E \\
&\Rightarrow\; \Served(a,t,g,F)\Bigr).
\end{split}
\]
\emph{No unit is allocated to a lower-ranked end while a
higher-ranked end that the good could serve remains unserved,
given homogeneity and feasibility of reallocation within the
choice menu.}\footnote{The qualifier ``that the good could
serve'' asks only about ends the good can actually serve.  \good{Water}
put to washing is not faulted for leaving hunger unmet --- \good{water}
cannot feed anyone --- even though hunger outranks washing.  Drop
the qualifier and the axiom would wrongly rule out any schedule
in which two goods take turns down one value scale, such as the
two-good Crusoe box below.}

\smallskip\noindent 
Equivalently, in set terms the served set is a \emph{top-segment} of
the $g$-serviceable ends --- closed upward within them:
\[
\Served(a,t,g) \;=\; {\uparrow}\Served(a,t,g) \,\cap\, \Serviceable(a,t,g).
\]
(Here ${\uparrow}S := \{F\in\mathsf{Ends} \mid \exists E\in S,\ F\pref{a}{t}E \lor F=E\}$
is the top-segment generated by $S$.\footnote{In words,
${\uparrow}S$ is $S$ together with every end ranked at least as
high as some member of $S$; since $\pref{a}{t}$ is a total order
this is every end at or above the lowest-ranked member of $S$.
For instance, with $\mathsf{Drink} \succ \mathsf{Eat} \succ
\mathsf{Cook} \succ \mathsf{Store} \succ \mathsf{Wash} \succ
\mathsf{Garden}$, one has ${\uparrow}\{\mathsf{Cook}\} =
\{\mathsf{Drink}, \mathsf{Eat}, \mathsf{Cook}\}$ and
${\uparrow}\{\mathsf{Cook}, \mathsf{Wash}\} = \{\mathsf{Drink},
\mathsf{Eat}, \mathsf{Cook}, \mathsf{Store}, \mathsf{Wash}\}$ ---
everything down to $\mathsf{Wash}$, with only $\mathsf{Garden}$
excluded.})
\end{enumerate}

\begin{crusoe}[a two-good allotment schedule]
Suppose at some $(a,t)$ Crusoe holds three units
$w_1, w_2, w_3$ of the good $\good{water}$ and two units
$f_1, f_2$ of the good $\good{fish}$, with six
choice-relevant ends ranked
\[
\mathsf{Drink} \;\succ\; \mathsf{Eat} \;\succ\; \mathsf{Cook}
\;\succ\; \mathsf{Store} \;\succ\; \mathsf{Wash} \;\succ\;
\mathsf{Garden}.
\]
\good{Water} can serve $\mathsf{Drink}$, $\mathsf{Cook}$,
$\mathsf{Wash}$, $\mathsf{Garden}$; \good{fish} can serve
$\mathsf{Eat}$, $\mathsf{Store}$.  His allotment schedule:

\medskip
\begin{center}
\renewcommand{\arraystretch}{1.2}
\begin{tabular}{clll}
\toprule
\textbf{Rank} & \textbf{End} & \textbf{Serviceable by} &
\textbf{Cells (unit $\to$ end)} \\
\midrule
1 & $\mathsf{Drink}$  & \good{water} & $w_1 \to \mathsf{Drink}$ \\
2 & $\mathsf{Eat}$    & \good{fish}  & $f_1 \to \mathsf{Eat}$ \\
3 & $\mathsf{Cook}$   & \good{water} & $w_2 \to \mathsf{Cook}$ \\
4 & $\mathsf{Store}$  & \good{fish}  & $f_2 \to \mathsf{Store}$ \\
5 & $\mathsf{Wash}$   & \good{water} & $w_3 \to \mathsf{Wash}$ \\
6 & $\mathsf{Garden}$ & \good{water} & --- (unserved) \\
\bottomrule
\end{tabular}
\end{center}
\medskip

The \emph{schedule} is the set of five cells in the last
column.  Each unit occupies exactly one cell \axref{MU0} and
every cell pairs a unit with an end its good can serve
\axref{MU1}.  Because each unit is a technological unit, one
unit serves one want by construction: every served end takes
exactly one unit, every marginal cell is uniquely served, and the
hypothesis of Theorem~\ref{thm:DMU} holds automatically.  A want
needing a larger quantity --- Mises's $30$-log cabin --- is just a
larger technological unit, not a multi-unit cell.
\axref{MU3} holds: $\mathsf{Garden}$ is \good{water}-serviceable but
unserved, so \good{water} is genuinely scarce --- one scarce good is all the
existential \axref{MU3} asks for.  \good{Fish}, by contrast, serves both its
serviceable ends ($\mathsf{Eat}$, $\mathsf{Store}$) and is here
non-scarce, which the axiom permits.
\axref{MU4} holds good by good: \good{water}'s served set
$\{\mathsf{Drink}, \mathsf{Cook}, \mathsf{Wash}\}$ is a
top-segment of the \good{water}-serviceable ends
$\{\mathsf{Drink}, \mathsf{Cook}, \mathsf{Wash},
\mathsf{Garden}\}$, and \good{fish}'s served set
$\{\mathsf{Eat}, \mathsf{Store}\}$ exhausts the
\good{fish}-serviceable ends.  The goods \emph{alternate} down the
single value scale --- \good{water}, \good{fish}, \good{water}, \good{fish}, \good{water} ---
which is exactly what the $\Serviceable$ qualifier in
\axref{MU4} admits: \good{water} serving $\mathsf{Cook}$ is not
faulted for $\mathsf{Eat} \succ \mathsf{Cook}$ going unserved
by \good{water}, since \good{water} cannot serve $\mathsf{Eat}$.  Each good
has its own marginal end on the one unified scale
($\mathsf{Wash}$ for \good{water}, $\mathsf{Store}$ for \good{fish}) ---
the Mengerian picture of one value scale allocating several
distinct supplies at once.
\end{crusoe}

\begin{lemma}[Disjointness]\label{lem:disjointness}
If serviceability is \emph{disjoint} at $(a,t)$ ---
$\Serviceable(a,t,g)\cap\Serviceable(a,t,g')=\emptyset$ for $g\neq g'$,
as in the box above --- then the served sets are pairwise disjoint, and
each good's served set, marginal end, and \axref{MU4} top-segment depend
on $g$ alone.  Diminishing marginal utility then applies to each good
\emph{independently}.
\end{lemma}

\begin{proof}
$\Served(a,t,g)\subseteq\Serviceable(a,t,g)$ by \axref{MU1}, so
disjoint serviceable sets force disjoint served sets.  Each good's
$\Served(a,t,g)$, its $\pref{a}{t}$-minimum, and the \axref{MU4}
identity $\Served(a,t,g)={\uparrow}\Served(a,t,g)\cap\Serviceable(a,t,g)$
mention $g$ and no other good; hence the goods neither share ends nor
constrain one another.
\end{proof}

\noindent\emph{Beyond disjointness: imputation.}
Disjointness is what makes the per-good apparatus self-contained.  When
it fails --- one end $F$ serviceable by two goods,
$F\in\Serviceable(a,t,g)\cap\Serviceable(a,t,g')$ --- the goods become
\emph{substitutes} for $F$, and the per-good picture breaks: serving
$F$ with $g'$ alone leaves $g$'s served set with a gap at $F$, so
\axref{MU4} for $g$ fails even though the allocation is rational.  The
deeper issue is valuation: a good has no value in itself, only what is
\emph{imputed} from the ends it serves.  Under disjointness this is
immediate --- a good's marginal utility is its own marginal end.  Under
substitution it is not: losing a unit no longer forfeits the shared end
(the other good serves it), so the unit's worth is the loss suffered
after re-allocation --- Wieser's \emph{loss principle}.  Pinning down
what each unit is really worth once goods overlap like this is the full
\emph{imputation} problem; the case treated here, where no two goods
serve the same end, is only its simplest.  The general case is the
imputation theorem of future research (``individual imputation requires
non-specificity''), in the tradition of Menger, Wieser, and
B\"ohm-Bawerk \citep{Menger1871,Wieser1889,BohmBawerk1889}.

\smallskip\noindent\emph{The value paradox, already in the schedule.}
No third good is needed --- read \good{fish} as the ``\good{platinum}'' of
Remark~\ref{rem:value_paradox}.  \good{Water} serves the vital
$\mathsf{Drink}$, yet, being the more abundant good, its three units
reach down to $\mathsf{Wash}$; \good{fish}'s two leave its marginal end at
$\mathsf{Store} \succ \mathsf{Wash}$.  A unit of \good{fish} therefore
outvalues a unit of \good{water} --- though \good{water} serves the more vital end ---
just as an ounce of \good{platinum} outvalues a loaf of \good{bread}.  The
dissolution is the same: each good's marginal unit is valued by the
lowest end its \emph{own} stock reaches, so \good{water}'s marginal unit ranks
low \emph{from abundance, not unimportance of the kind} --- though
\good{water} reaches higher than \good{fish} in $\Serviceable$, its
abundance carries its margin below \good{fish}'s.

A first consequence of the (MU)-axioms makes the marginal-end
apparatus well-posed.

\begin{lemma}[Served ends are choice-relevant]
\label{lem:served_choice_relevant}
For all $a,t,g,E$, 
\[
\Served(a,t,g,E) \Rightarrow \EndAt(a,t,E),
\]
and likewise for the reduced served set (since
$\Served^{\setminus y}(a,t,g,E) \Rightarrow \Served(a,t,g,E)$).
Hence the (reduced) served set $\Served(a,t,g)$ (resp.\
$\Served^{\setminus y}(a,t,g)$) is a subset of the end menu
$\EndAt(a,t)$, 
\[
  \Served(a,t,g) \subseteq \EndAt(a,t), \qquad (\text{resp. } \Served^{\setminus y}(a,t,g) \subseteq \EndAt(a,t)),
\]
and any nonempty served set has a \emph{unique}
$\pref{a}{t}$-minimal element (i.e. the marginal end).
\end{lemma}

\begin{proof}
If $\Served(a,t,g,E)$ then some $x$ satisfies
$\UnitOf(x,g) \land \Allot(a,t,x,E)$
(Definition~\ref{def:served}).  By \axref{MU1} feasibility,
$\Serviceable(a,t,g,E)$ holds. Definition~\ref{def:serviceable} yields an
available action $\alpha$ with $\EndOf(\alpha,E)$, i.e.\
$\EndAt(a,t,E)$.  The reduced case follows by dropping the
$x \neq y$ restriction.  The inclusion in $\EndAt(a,t)$ is then
immediate.  Comparability \axref{O2}, transitivity \axref{O3},
and asymmetry \axref{O4} make $\pref{a}{t}$ a strict total order;
finiteness of $\EndAt(a,t)$ (Remark~\ref{rem:finiteness}) makes
the served set a nonempty finite subset of it.  A nonempty finite
subset of a strict total order has a unique minimal element ---
the order axioms supply comparability and uniqueness, finiteness
supplies existence.
\end{proof}

Before stating the theorem, recall Mises's definition of the
marginal employment of a homogeneous supply:
\begin{quote}
``We call that employment of a unit of a homogeneous supply
which a man makes if his supply is $n$ units, but would not
make if, other things being equal, his supply were only
$n - 1$ units, the least urgent employment or the marginal
employment, and the utility derived from it marginal utility''
\citep[Ch.~VII, \S1]{Mises1949}.\footnote{Mises's ``unit of a
homogeneous supply'' is the \emph{technological unit} introduced
earlier at the definition of a good --- the smallest quantity
that enters action \citep[ch.~1]{Rothbard2009}.  The grain is
essential: only at the technological unit does the step from $n$
to $n-1$ put exactly one employment at stake.  Below it the
comparison misfires --- removing a single grain of rice forgoes
no employment at all, while removing one log of a thirty-log
cabin forgoes the whole cabin (each log then ``worth'' the
entire cabin, the rising-marginal-utility illusion).  Both are
artefacts of counting beneath the unit that enters action.}
\end{quote}
The diminishing-marginal-utility theorem says this marginal
employment is for a less urgent want than any of those served
by the $n-1$ unit supply:
\begin{quote}
``If the supply available increases from $n - 1$ units to $n$
units, the increment can be employed only for the removal of
a want which is less urgent or less painful than the least
urgent or least painful among all those wants which could be
removed by means of the supply $n - 1$''
\citep[Ch.~VII, \S1]{Mises1949}.
\end{quote}
Rothbard restates the same point in the supply-loss
direction, in language closer to standard economic
exposition:
\begin{quote}
``The actor gives up the lowest-ranking want that the original
stock\dots was capable of satisfying.  This one unit that he
must consider giving up is called the marginal unit\dots\ the
least important end fulfilled by the stock is known as the
satisfaction provided by the marginal unit, or the utility of
the marginal unit''
\citep[ch.~1, pp.~26--27]{Rothbard2009}.
\end{quote}
The formal counterpart of ``$n$ units versus $n-1$ units'' is
the contrast between the actor's allotment schedule with the
new unit's cell and the same schedule without it.  In the
theorem below, $y$ is the additional ($n$-th) unit, $E$ its
allotment cell (the \emph{least urgent} or marginal employment
of the $n$-unit supply), and $E^*$ the least-preferred end
served by the $n-1$-unit supply --- not assumed but
\emph{derived} in the proof.  The hypotheses
(existence of the marginal unit $y$; non-emptiness of the
reduced supply) and their Misesian correspondences are stated
with the theorem.

\begin{maintheorem}[Diminishing marginal utility]\label{thm:DMU}
Fix $a,t,g$ and suppose the served set is non-empty, so the
\emph{marginal unit} $y$ exists
(Definition~\ref{def:marginal_end}).\footnote{The satiation principle
underlying the law of diminishing marginal utility is due to
Carl \citet[ch.~III]{Menger1871}, who introduced the
marginal-utility framework contemporaneously with Jevons and
Walras.  The Hilbert-style formulation makes Mises's verbal
hypotheses explicit: his ``least urgent employment'' is the
marginal end $E$ (the $\pref{a}{t}$-least served end); his
``employment which a man would not make if his supply were only
$n-1$ units'' is the marginal unit $y$, the single unit whose
withdrawal drops the supply to $n-1$; and his
implicit $n \geq 2$ is the non-emptiness of
$\Served^{\setminus y}$, under which the reduced marginal end
$E^*$ is not assumed but derived.}  Then
\[
\Served(a,t,g) = \Served^{\setminus y}(a,t,g) \cup \{E\};
\]
and if moreover $\Served^{\setminus y}(a,t,g) \neq \emptyset$, write
$\MU^{\setminus y}(a,t,g) := \min\nolimits_{\pref{a}{t}}
\Served^{\setminus y}(a,t,g) = E^*$ for the good's marginal utility one
unit short.  Then the marginal utility \emph{falls when the $n$-th unit
is added}: at $n$ units it is $E = \MU(a,t,g)$, at $n-1$ the strictly
more urgent $E^*$, so
\[
\MU^{\setminus y}(a,t,g) \;\pref{a}{t}\; \MU(a,t,g)
\qquad(\text{equivalently } E^* \pref{a}{t} E).
\]
\end{maintheorem}
This is the law in Mises's own form: each added unit is put to a want
less urgent than any the previous units served, so the marginal
utility of the good falls as its stock grows.
\begin{proof}[Proof of Theorem~\ref{thm:DMU}]
At the level of sets it is elementary.  By \axref{MU0}
functionality $y$ serves only $E$, and as the marginal unit $y$
is the only unit serving $E$; so removing $y$ deletes exactly
$E$ from the served set,
\[
\Served(a,t,g) = \Served^{\setminus y}(a,t,g) \cup \{E\},
\]
the first claim.  By Lemma~\ref{lem:served_choice_relevant} both
$\Served(a,t,g)$ and $\Served^{\setminus y}(a,t,g)$ are finite
subsets of $\EndAt(a,t)$, totally ordered by $\pref{a}{t}$.  So
when $\Served^{\setminus y}(a,t,g)$ is non-empty its minimum
$E^*$ exists; and since $\Served^{\setminus y}(a,t,g) =
\Served(a,t,g)\setminus\{E\}$ drops the $\pref{a}{t}$-least
element $E$ of $\Served(a,t,g)$, its minimum is strictly larger: $E^* \pref{a}{t} E$.
\end{proof}
Appendix~\ref{app:dmu_proof} spells out the predicate-level steps
the type-checker performs; \axref{MU4} enters only in the
corollary below, and Remark~\ref{rem:dmu_method} accounts for
each axiom's role.

\begin{corollary}[Structure preservation under marginal removal]
\label{cor:dmu_structure}
Under the hypotheses of Theorem~\ref{thm:DMU} and 
assuming \axref{MU4} holds\footnote{A premise Theorem~\ref{thm:DMU} itself does not assume.} 
for $\Served(a,t,g)$ --- the
``$y$-removed'' allotment $\Served^{\setminus y}(a,t,g)$ is again a top-segment of the
$g$-serviceable ends:
\[
\Served^{\setminus y}(a,t,g) \;=\; {\uparrow}\Served^{\setminus y}(a,t,g)
\,\cap\, \Serviceable(a,t,g).
\]
\end{corollary}

\begin{proof}[Proof of Corollary~\ref{cor:dmu_structure}]
Of the two inclusions in the stated identity, only $\supseteq$
has content.  The inclusion $\subseteq$ is automatic: a
reduced-served end ranks $\succeq$ itself --- so lies in
${\uparrow}\Served^{\setminus y}(a,t,g)$ --- and is serviceable by
\axref{MU1}.  For $\supseteq$ we must show the reduced served set
is \emph{upward-closed}: every $g$-serviceable end $G$ ranked
above some reduced-served $F'$ is itself reduced-served.

Fix such a $G \succ F'$.  Since $F'$ is served and $G$ outranks
it, \axref{MU4} --- the full served set is a top-segment ---
makes $G$ served as well, so some unit of $g$ is allotted to it.
That unit cannot be the marginal unit $y$: by \axref{MU0} a unit
serves only one end, and $y$ already serves the marginal end $E$,
so a $y$ that also served $G$ would force $G = E$ --- impossible,
since $G \succ F' \succ E$ (every reduced-served end outranks the
marginal $E$).  The unit serving $G$ is therefore one of the
others, so $G$ remains served once $y$ is removed:
$G \in \Served^{\setminus y}(a,t,g)$.

Appendix~\ref{app:dmu_proof} gives the Hilbert-style version with
step-by-step commentary.
\end{proof}

\begin{remark}[The law iterates: from one step to the whole schedule]
\label{rem:dmu_iterates}
Theorem~\ref{thm:DMU} compares a single pair of supplies, $n$ and
$n-1$; Corollary~\ref{cor:dmu_structure} is what lets the comparison
\emph{repeat}.  Because the reduced served set is again a top-segment
--- a rational allocation in its own right --- the theorem applies to
it afresh: remove its marginal unit, and the marginal end rises again.
As long as units remain this continues, and the successive marginal
ends form a strictly ascending chain
\[
E \;\prec\; E^* \;\prec\; E^{**} \;\prec\; \cdots
\]
running from the least urgent want the full stock serves up to the
single most urgent want one unit would serve.  Read in the other
direction --- the supply growing from one unit to $n$ --- it is a
strictly \emph{falling} schedule of marginal utility: the good's whole
Mengerian value scale, each added unit serving a less urgent want than
the one before.

This is the inductive content of Mises's ``$n$ versus $n-1$'' framing:
the single-step theorem is the induction step, and the structure
preservation of the corollary is the invariant that carries it along
the whole supply.  Theorem~\ref{thm:labor} runs exactly this iteration
on leisure-time.
\end{remark}

\begin{remark}[Methodological significance and Mises's ``already implied'' claim]
\label{rem:dmu_method}
The role each axiom plays in
Theorem~\ref{thm:DMU} and Corollary~\ref{cor:dmu_structure} is:

\begin{center}
\begin{tabular}{@{}p{2.4cm}p{4.5cm}p{6cm}@{}}
\toprule
\textbf{Axiom} & \textbf{Facet of rationality} & \textbf{Role} \\
\midrule
\axref{MU0}    & functionality (schedule is a partial function)      & \textbf{logical}: gives the deletion identity $\Served^{\setminus y} = \Served\setminus\{E\}$ (so the reduced marginal $E^*$ exists once the reduced set is non-empty, by finiteness, Lemma~\ref{lem:served_choice_relevant}) and, in the corollary, rules $y$ out as the unit serving a higher end \\
\axref{MU1}    & feasibility (only serviceable cells populated)      & \textbf{logical}: underwrites Lemma~\ref{lem:served_choice_relevant} --- every served end is choice-relevant \\
\axref{O2}--\axref{O4} & preference (strict total order on $\EndAt(a,t)$) & \textbf{logical}: makes ``least preferred'' meaningful for the marginal-end definitions \\
\axref{MU2}    & fungibility (units interchangeable)                 & \textbf{semantic}: makes ``another unit of $g$'' a meaningful operation; without it, units are distinguishable and the marginal-employment doctrine doesn't apply \\
\axref{MU3} & scarcity awareness                    & \textbf{non-triviality}: ensures some good is genuinely scarce (a serviceable end its stock leaves unserved), so the framework is not about a free-good world \\
\axref{MU4}    & coherence (top-segment served set)                  & \textbf{logical, in the corollary}: invoked to propagate the top-segment shape from the full served set to the reduced served set \\
\bottomrule
\end{tabular}
\end{center}

Taking the \emph{marginal} unit is not gratuitous: \axref{MU4}
alone permits ``wasteful'' Mengerian-looking allocations on which
a non-marginal unit violates the theorem's inequality.  An actor
with units $x_1, \ldots, x_5$ of \good{water} and ranked ends
$\mathsf{drink} \succ \mathsf{cook} \succ \mathsf{wash} \succ
\mathsf{garden} \succ \mathsf{fun}$ (together with further,
lower-ranked \good{water}-serviceable ends that remain unserved, so
\axref{MU3} holds) may allot $x_1 \to \mathsf{cook}$, $x_2 \to
\mathsf{wash}$, $x_3 \to \mathsf{garden}$, $x_4 \to \mathsf{fun}$,
$x_5 \to \mathsf{drink}$.  The served set is a top-segment, so all
(MU)-axioms hold.  But $x_5$, which uniquely serves
$\mathsf{drink}$, is \emph{not} the marginal unit --- \good{water}'s
marginal end is $\mathsf{fun}$, served by $x_4$.  Removing $x_5$
leaves reduced marginal $\mathsf{fun} \prec \mathsf{drink}$, so
the would-be conclusion $\mathsf{fun} \succ \mathsf{drink}$ is
false.  The theorem applies only at the margin: $y$ must serve
the \emph{least urgent} end.

Mises is emphatic that ``the law of marginal utility is already
implied in the category of action''
\citep[Ch.~VII, \S1]{Mises1949}, and the structures isolated above
might appear to contradict this.  The two views agree once one
notices where Mises packs the work.  ``Human action is
necessarily always rational; the term `rational action' is
therefore pleonastic'' \citep[Ch.~I, \S4]{Mises1949}; ``means are
necessarily always limited\dots\ if this were not the case, there
would not be any action with regard to them''
\citep[Ch.~IV, \S1]{Mises1949}; ``every action is always in
perfect agreement with the scale of values''
\citep[Ch.~IV, \S2]{Mises1949}.  The (MU)-axioms, together with
the Layer-2 preference order, are not hypotheses Mises adds to
acting --- they are constitutive of what acting \emph{is} for
him.  The Hilbert-style treatment is \emph{diagnostic} of Mises's
reasoning, not corrective: it isolates the structure he carries
entirely inside his concept of acting man.\footnote{The point
that introducing a preference scale silently imports completeness
and transitivity is also recorded by
\citet{Hudik2012}.  The present treatment isolates the
same kind of hidden premise at finer granularity: the strict
order in \axref{O2}--\axref{O4}, the homogeneity in \axref{MU2},
the top-segment shape in \axref{MU4}, and so on.}
\end{remark}

\begin{crusoe}[diminishing marginal utility, on the two-good schedule]
Return to the two-good allotment schedule introduced after the
(MU)-axioms: three units $w_1, w_2, w_3$ of $\good{water}$
and two units $f_1, f_2$ of $\good{fish}$, ends ranked
$\mathsf{Drink} \succ \mathsf{Eat} \succ \mathsf{Cook} \succ
\mathsf{Store} \succ \mathsf{Wash} \succ \mathsf{Garden}$.
Theorem~\ref{thm:DMU} applies once \emph{per good}, each time at
that good's own margin.

\medskip
For \good{water}, the marginal end is $\mathsf{Wash}$, the
lowest of the \good{water}-served $\{\mathsf{Drink}, \mathsf{Cook},
\mathsf{Wash}\}$, so the marginal unit is $y = w_3$, the one
unit serving it.  The reduced supply still serves
$\mathsf{Drink}$ and $\mathsf{Cook}$, hence is non-empty.  The theorem delivers the marginal of the reduced
\good{water}-served set $\{\mathsf{Drink}, \mathsf{Cook}\}$, namely
$E^* = \mathsf{Cook}$, and indeed
$\mathsf{Cook} \succ \mathsf{Wash}$: if Crusoe loses a bucket
of \good{water}, it is the washing that goes, not the drinking or the
cooking.

\medskip
For \good{fish}, the marginal unit is $y = f_2$, serving the
marginal end $\mathsf{Store}$, with $f_1$ serving
$\mathsf{Eat}$ (so the reduced supply is non-empty); the theorem
yields
$E^* = \mathsf{Eat} \succ \mathsf{Store}$: losing a \good{fish} costs
tomorrow's stored meal, not tonight's dinner.

\medskip
Each good's successive units serve less urgently desired ends
--- DMU in purely ordinal terms --- and the two per-good margins
interleave on the single value scale without the language ever
comparing ``\good{water} as a class'' against ``\good{fish} as a class''
(Remark~\ref{rem:value_paradox}).  Here each end is served by a
single technological unit (the schedule above), so every marginal
cell is uniquely served and the $n \to n-1$ transition strips
exactly one end.

\medskip
Corollary~\ref{cor:dmu_structure} adds that each reduced
allotment is still a top-segment of its good's serviceable
ends: \good{water}'s $\{\mathsf{Drink}, \mathsf{Cook}\}$ within
$\{\mathsf{Drink}, \mathsf{Cook}, \mathsf{Wash},
\mathsf{Garden}\}$, \good{fish}'s $\{\mathsf{Eat}\}$ within
$\{\mathsf{Eat}, \mathsf{Store}\}$.  \axref{MU2},
\axref{MU3}, and \axref{MU4} make the setup itself coherent
(see Remark~\ref{rem:dmu_method}): \axref{MU3} is shown by
\good{water}'s $\mathsf{Garden}$ (serviceable but unserved), and
the $\Serviceable$ qualifier of \axref{MU4} is doing visible
work --- $\mathsf{Eat}$ outranks $\mathsf{Cook}$, yet \good{water}'s
serving $\mathsf{Cook}$ is not faulted, since \good{water} cannot
serve $\mathsf{Eat}$.
\end{crusoe}

\begin{remark}[Two opposite descents]
\label{rem:valuation_vs_performance}
The (MU)-apparatus values a stock; it does not describe what the actor
\emph{does} next.  Reading ``drink, then cook, then wash'' as walking
\emph{down} a fixed scale --- the satisfied end dropping from the
\emph{top} --- is the \emph{action} descent, which want is met next:
the most urgent first \citep[Ch.~VII, \S1]{Mises1949}.  The
marginal-utility descent of Theorem~\ref{thm:DMU} runs the other way,
removing the \emph{bottom}: the marginal unit serves the least-urgent
end, so losing a unit shrinks the served set from below, its marginal
end $E$ becoming unserved (Definition~\ref{def:marginal_end},
Corollary~\ref{cor:dmu_structure}).  The two answer different
questions --- \emph{what is a unit worth} versus \emph{what is done
first} --- and should not be conflated; the (MU)-apparatus formalizes
only the former.  And because this is a momentary valuation, it does
not compound across time: wants recur, so the scale cycles rather than
contracting monotonically (Remark~\ref{rem:constancy}).
\end{remark}

\subsection{Disutility of labor and the marginal utility of leisure}
\label{subsec:derived_labor}

Mises insists that ``the expenditure of labor is deemed painful''
and that ``leisure is, other things being equal, preferred to
travail'' \citep[Ch.~VII, \S3, p.~131]{Mises1949}.  He
treats leisure as ``an economic good of the first order'' and
applies the law of marginal utility to conclude that ``the
disutility of labor felt by the worker increases in a greater
proportion than the amount of labor expended'' (p.~134).  We
formalize this within $\Lprx$ using one additional predicate and
one axiom, from which both the disutility of labor and its
\emph{increasing} margin follow.

Why a separate axiom rather than a base-theory result?  Mises is explicit that disutility of labor is not part
of the praxeological core: ``the disutility of labor is not of a
categorial and aprioristic character''
\citep[Ch.~II, \S10]{Mises1949}.  It is a contingent fact about
human beings as biological agents, not a structural feature of
purposeful action.  What is contingent, precisely, is that labor is
\emph{means-only} --- undergone for a distinct produced good, never
enjoyed for its own sake; a non-Misesian world allows an activity at
once productive and its own reward, where labor and leisure share an
end.  That is the content of axiom~\axref{LB1} below.  The
\emph{ranking} of labor's end above leisure is then no further
postulate but a consequence of demonstrated preference
(Proposition~\ref{prop:disutility}).  The predicate $\Labor$
and the axiom \axref{LB1} are introduced as a small enrichment on top
of $\Tprx$, not as elements of the base axiomatization: a model of
$\Tprx$ in which no action is labor, or in which labor is not
means-only, remains a perfectly admissible $\Tprx$-model; only
$\Tprx + (\text{LB1})$ distinguishes the Misesian sub-class.

\[
\Labor(\alpha) \quad
\text{(``action $\alpha$ is an expenditure of labor'')}.
\]

Labor is the employment of physiological functions and
manifestations of human life as a means.  It is distinguished
from immediately gratifying activity (leisure, play) by the
fact that the actor submits to it only for the sake of the
eventual gratification it produces.  The predicate
$\Labor(\alpha)$ marks this distinction within the
Actions sort.\footnote{The predicate $\Labor$ is scoped to
means--end actions producing distinct consumers' goods, not to
intrinsically rewarding activity.  \citet[ch.~1,
p.~46]{Rothbard2009} draws the line explicitly: ``Those
activities which are engaged in purely for their own sake are
not labor but are pure play, consumers' goods in themselves.''
A ``play'' action is not tagged $\Labor$; axiom~\axref{LB1}
attaches to labor actions only.}

The ends an immediately-gratifying (non-labor) action can secure
are the actor's \emph{leisure-ends}:
\[
\Leisure(a,t,E) \;:\Leftrightarrow\;
\exists\alpha\,\bigl(\Avail(a,\alpha,t) \wedge
\neg\Labor(\alpha) \wedge \EndOf(\alpha,E)\bigr),
\]
a derived classification of ends, built from the
$\Labor$ predicate --- no new primitive.  Let $E_\ell$ denote such a leisure-end
($\Leisure(a,t,E_\ell)$); we assume $\EndAt(a,t,E_\ell)$
whenever $a$ has the option of not expending labor at $t$ ---
i.e.\ whenever some $\beta$ with $\neg\Labor(\beta)$ is
available.

\medskip
\noindent\textbf{Axiom.}

\begin{enumerate}[label=(LB1), leftmargin=3em]
\item\label{ax:LB1} \textbf{Labor is means-only} \emph{(rationality of labor).}
A labor action and a non-labor action never share an end:
\[
\forall \alpha\,\beta\;\bigl(\Labor(\alpha) \land
\neg\Labor(\beta) \;\Rightarrow\; \EndOf(\alpha) \neq
\EndOf(\beta)\bigr).
\]
\end{enumerate}

\noindent The disutility of labor then follows at once, by demonstrated
preference --- it is not a further axiom.

\begin{proposition}[Disutility of labor]\label{prop:disutility}
If a labor action is performed at $(a,t)$ while some non-labor
alternative is available --- $\Acts(a,\alpha,t)$,
$\Labor(\alpha)$, and $\Avail(a,\beta,t)$,
$\neg\Labor(\beta)$ with $\EndOf(\beta,E_\ell)$ --- then the end
of $\alpha$ ranks above leisure: $\EndOf(\alpha) \pref{a}{t} E_\ell$.
\end{proposition}

\begin{proof}
The performed $\alpha$ and the available $\beta$ are distinct (by
$\Labor(\alpha) \land \neg\Labor(\beta)$) with distinct
ends (by \axref{LB1}), so Definition~\ref{def:revpref} records
$\EndOf(\alpha) \revpref{a}{t} E_\ell$, and \axref{O0} carries it into
the order.
\end{proof}

\noindent The worker submits to labor only for an end he ranks above
leisure: the disutility is the demonstrated-preference record applied to
the labor/leisure choice, carried into the order by \axref{O0}, not an
extra behavioral postulate.  All the enrichment genuinely adds is the
means-only structure \axref{LB1}; that leisure and labor compete on the
one scale of values $\pref{a}{t}$ then follows.

Treat leisure-\emph{time} as a homogeneous good $g_\ell$ --- its
units the actor's free hours, satisfying \axref{MU0}--\axref{MU4}
(Section~\ref{subsec:derived_dmu}) and allotted across the
leisure-ends $\Leisure(a,t,\cdot)$ --- and let a
\emph{unit of labor} be one such hour withdrawn, spent working
rather than enjoyed.\footnote{The units $u$ are the actor's free
hours, taken as homogeneous resource-units (\emph{Things}), so
$\UnitOf(u,g_\ell)$ is well-typed and the (MU)-apparatus applies
verbatim.  We deliberately do \emph{not} widen $\UnitOf$ to the
$\mathsf{Times}$ sort: that retyping would thread through $\Use$,
$\Allot$, and $\Serviceable$, with consequences across the
apparatus.  The (MU)-layer needs only \emph{some} homogeneous
divisible good --- leisure-time supplies one, Mises's ``economic
good of the first order.''}

\begin{definition}[Marginal disutility of labor]\label{def:mdl}
The \emph{marginal disutility of labor} at $(a,t)$ is just the marginal
utility of the leisure-time good $g_\ell$ (Definition~\ref{def:value}):
\[
\MDL(a,t) \;:=\; \MU(a,t,g_\ell)
\;=\; \min\nolimits_{\pref{a}{t}} \Served(a,t,g_\ell),
\]
the least urgent leisure-end the actor's free hours still serve --- the one
a further hour of labor would forfeit.
\end{definition}

\begin{maintheorem}[Increasing marginal disutility of labor]\label{thm:labor}
The marginal disutility of labor rises with each unit of labor: withdrawing
one leisure-hour moves $\MDL(a,t)$ strictly up the scale of values
$\pref{a}{t}$, so the leisure-end forfeited by each successive labor-hour is
strictly more urgent than the one before.\footnote{The roundabout-production
framework that underwrites the labor/capital trade-off is due to
\citet{BohmBawerk1889}; the present theorem extracts
the labor-side disutility gradient from a single layer of the
axiom system.}
\end{maintheorem}

\begin{proof}
The claim is Theorem~\ref{thm:DMU} applied to the leisure-time good
$g_\ell$, read with the stock falling; we spell out the iteration.
Take $g_\ell$ with $\UnitOf(u,g_\ell)$
for each unit $u$ of leisure available at $(a,t)$, allotted
across the actor's leisure-ends, giving the served set
$\Served(a,t,g_\ell)$.

By \axref{MU2} the leisure units are interchangeable, so
withdrawing one --- one unit of labor --- cancels the least urgent
cell of the schedule: the
marginal leisure-end
$E = \min\nolimits_{\pref{a}{t}}\Served(a,t,g_\ell)$, served by
the marginal unit $y$ alone, with the reduced set
$\Served^{\setminus y}(a,t,g_\ell)$ still non-empty.

\medskip

By Theorem~\ref{thm:DMU} applied to $g_\ell$, the reduced
marginal end
$E^* = \min\nolimits_{\pref{a}{t}}\Served^{\setminus y}(a,t,g_\ell)$
satisfies $E^* \pref{a}{t} E$.  A further unit of labor withdraws
another unit of leisure and --- again by \axref{MU2} --- gives up
the reduced stock's least-urgent served end, its marginal $E^*$,
strictly more urgent than the $E$ given up before.  The
opportunity cost of labor rises from each unit to the next.

\medskip

The argument iterates --- this is the general law of
Remark~\ref{rem:dmu_iterates}, here on leisure-time: by
Corollary~\ref{cor:dmu_structure} the
reduced set $\Served^{\setminus y}(a,t,g_\ell)$ is again a
top-segment of the $g_\ell$-serviceable ends --- a rational
allocation in its own right --- so Theorem~\ref{thm:DMU} applies
to it afresh, then to its own reduction, and so on (the
leisure-end menu stays finite throughout,
Remark~\ref{rem:finiteness}).  After $k$ withdrawals
the forfeited ends form a strictly ascending chain
$E \prec E^* \prec E^{**} \prec \cdots$ on the actor's scale.  This is what Mises means by
increasing marginal disutility of labor: ``the first unit of
leisure satisfies a desire more urgently felt than the
second one, the second one a more urgent desire than the
third one, and so on'' \citep[p.~134]{Mises1949}.
Reversing, the first unit of labor sacrifices the
\emph{least} urgent leisure-end; the second sacrifices a
\emph{more} urgent one; and so on.\footnote{Two scope notes.
First, the chain is comparative-static, exactly as in Mises's
``$n$ versus $n-1$ units'' framing: it compares leisure stocks
of successive sizes at one valuation moment.  Reading it as a
\emph{temporal} process --- labor extending over successive
moments $t < t' < \cdots$ --- additionally assumes that the
actor's leisure-scale is constant across those moments, an
assumption $\Tprx$ deliberately does not impose
(Remark~\ref{rem:constancy}).  Second, the $k$-step chain
formally requires the reduced served set of
Definition~\ref{def:reduced_served_set} generalized from one
withheld unit to a finite set of withheld units; the
generalization is mechanical, and only the consecutive-pair
instance is used above.}
\end{proof}

\begin{remark}[Scope of the labor theorem]\label{rem:labor_scope}
Theorem~\ref{thm:labor} inherits the hidden premises
\axref{MU0}--\axref{MU4} via its use of Theorem~\ref{thm:DMU} and
Corollary~\ref{cor:dmu_structure} --- \axref{MU4}, in particular,
enters only through the corollary's iteration step.  The
additional structure is minimal: one predicate
($\Labor$), the derived leisure-end classification
$\Leisure$ it induces, and the identification of
leisure-\emph{time} as the good $g_\ell$ serving those ends.  The result confirms Mises' methodological claim
that the disutility of labor ``does not need to be corrected
or complemented by an additional statement'' beyond the
ordinary law of marginal utility \citep[p.~137]{Mises1949}.
\end{remark}

\begin{crusoe}[disutility of labor: berry-picking]
Following \citet[ch.~1, pp.~47--49]{Rothbard2009},
suppose Crusoe finds himself on the island able to pick edible
berries by hand at a constant rate of 20 per hour; with 24
hours in the day his choice is how many to allocate to labor
(berry-picking, producing the end $\mathsf{Eat}$) versus
leisure (any non-labor action producing leisure-ends
$\mathsf{L}_1, \mathsf{L}_2, \ldots$).

\medskip
Treat leisure-hours as units of a homogeneous good $g_\ell$;
his most urgent leisure-end $\mathsf{L}_1$ is, say, rest after
exertion, $\mathsf{L}_2$ swimming, $\mathsf{L}_3$ exploring the
island, and so on.  By Theorem~\ref{thm:DMU} applied to
$g_\ell$, the first hour of leisure serves $\mathsf{L}_1$, the
second a less urgent leisure-end, and so on down his scale.

\medskip
Reversing the perspective gives Theorem~\ref{thm:labor}: the
first labor-hour gives up the \emph{least} urgent
leisure-end on his scale; the second labor-hour gives up a
\emph{more} urgent leisure-end; and so on.  Berries arrive at a
constant 20 per hour, so the labor's \emph{physical} marginal
product is constant; but the marginal \emph{utility} of those
berries falls as the stock grows (Theorem~\ref{thm:DMU}), while
the marginal disutility of labor rises (Theorem~\ref{thm:labor})
--- both blades of the scissors move.

\medskip
Crusoe stops working at the hour where the next leisure-hour
would serve a more urgent end than the next 20 berries can
serve --- i.e., when the marginal disutility of labor first
exceeds the marginal utility of its product.  In the language of
Definitions~\ref{def:value} and~\ref{def:mdl} this is one comparison
on the single scale: writing $g_b$ for the berries, another labor-hour
lowers $\MU(a,t,g_b)$, the marginal utility of the product
(Theorem~\ref{thm:DMU} on $g_b$), and raises $\MDL(a,t)$, the
marginal disutility of labor (Theorem~\ref{thm:labor}); Crusoe works while
$\MU(a,t,g_b) \pref{a}{t} \MDL(a,t)$ and stops at the labor-hour where $\MDL(a,t)$ first
outranks $\MU(a,t,g_b)$.\footnote{Two clocks, and $t$ wears both: at $t$ Crusoe \emph{values}
and \emph{acts}.  The equilibrium concerns the valuing --- the scale
$\pref{a}{t}$ and the allocation schedule $\Allot(a,t,\cdot)$, which
can hold the \emph{whole} $24$-hour plan because \axref{C1} caps
\emph{performance}, not the schedule's size.  Acting is the other
role: \axref{C1} executes one allotment at $t$, the rest pending for
later moments (by \axref{P5}, distinct moments).  So ``the $n$-th
labor-hour'' is a cell of the \emph{plan}; the rule ranks plans
($n$ versus $n\pm1$ labor-hours) at the one valuing $t$ ---
comparative-static.  Carrying a plan out is the temporal process
$t < t' < \cdots$, which would further need the scale held constant
(the proof's first scope note).}
Rothbard frames
the equilibrium condition directly: ``a man will stop work
when the marginal disutility of labor is greater than the
marginal utility of the increased goods provided by the
effort''~\citep[ch.~1, p.~46]{Rothbard2009}.  The allocation is
determined entirely by Crusoe's value scale; no numerical
leisure/labor exchange rate appears.
\end{crusoe}

\section{Future research}
\label{sec:outlook}

The base system $\Tprx$ is sufficient for everything proved in
this paper but it is not yet sufficient to express the full
conceptual apparatus of \emph{Human Action}.  Future research
extends $\Tprx$ along two complementary axes.

\medskip
\noindent\textbf{Extensions.}
$\Lprx$ can be enriched with five further layers: production
structure (a $\Result$ predicate, a $\Consumable$ predicate,
a higher-order goods hierarchy\footnote{Garrison's
capital-structure framework~\citep{Garrison2001} frames the same
hierarchy in terms of \emph{stages of production}, with earlier
stages further from final consumption.}), ownership (an
$\Owns$ predicate), exchange (an $\Exch$ predicate with
ownership-prerequisite, ownership-transfer, bilateral
distinctness, and symmetry axioms), monetary calculation
(a $\Money$ predicate), and a transition map $\tau$
sending each global state of the economy to its
$\tau$-successor.  The five layers are organized in two
independent branches that reconnect only through dynamics.
A second batch of Misesian theorems then becomes derivable:
mutual gain from voluntary exchange (purely ordinal,
\emph{ex ante}, no cardinal utility), the logical possibility
of profit and loss in an enriched dynamic system, the
subjectivity of production valuation
(via the higher-order-goods hierarchy), individual imputation
requires competing uses of factors, and a Coase-style
characterization of complete property rights as the absence of
an externality residual.  The base system $\Tprx$ of the
present paper is the substrate on which all of this can be
built.

\medskip
\noindent\textbf{Two impossibilities.}
With the language $\Lprx$ established (the present paper) and the
production enrichment available (future research), the policy-maker can
be introduced as a distinguished actor and the Economic
Calculation Problem \textsc{calc}\footnote{``\textsc{calc}'' is
short for Mises's \emph{economic calculation}: the formal
decision problem that asks, of an arbitrary praxeological economy
and a designated policy objective, whether some allocation in the
economy achieves the objective.} can be stated formally: does there
exist a feasible allocation achieving a designated target?
First, a faithfulness check:
Mises' classical pricing argument can be recovered as a theorem
of $\Tprx$.  Two structurally symmetric actors --- symmetric in
their availability, use, and \emph{preference} structure --- are
interchangeable by an automorphism of every $\Tprx$-model, so
no formula of $\Lprx$ expresses social optimality of an
allocation, and adding a price primitive to the language is
exactly what makes such a formula well-formed.  Second, a fact
about the expressive power of the model class: once the production enrichment is in place, $\Tprx$ contains a
praxeological economy $\Emw$ for every Turing machine $M$ and
input $w$, and consequently \textsc{calc} is undecidable as a
uniform decision problem over the model class.\footnote{Neither
\textsc{calc} undecidability nor any other claim of the present
work is intended to rule out the metaphysical possibility of an
omniscient planner.  An entity with full insight into every
human valuation and full predictive access to the future ---
e.g.\ a deity --- is by hypothesis not a praxeological
policy-maker subject to the kinds of structural and
informational limits that Mises and Hayek wrote about, and
$\Tprx$ does not attempt to model such an entity.  The
undecidability of \textsc{calc} is a statement about
algorithms running on actual physical devices and about
formulas expressible in the finite first-order language
$\Lprx$; both are bounded in ways that an omniscient observer
is not.  Mises and Hayek wrote about flesh-and-blood policy
agents, and the present work stays within the same scope.}
Neither observation is presented as superseding Mises' or
Hayek's classical arguments; both sit alongside them as
mathematical statements about the model class of the
axiomatization established here.

\medskip
\noindent
Future research building on the present paper will work
in the language $\Lprx$ introduced in
Section~\ref{sec:language}, against the model class
delineated by $\Tprx$, and using theorems of the kind proved
in Section~\ref{sec:derived_base}.  Readers may take the
formal apparatus of the present paper as background and
may, if so inclined, consult the \Lean{} companion of
Appendix~\ref{sec:lean} to verify that the apparatus is
consistent.

\section*{Acknowledgments}
\addcontentsline{toc}{section}{Acknowledgments}

The authors used large language model assistance during the
preparation of this paper.  Specifically: for bibliographic
search and identification of secondary sources; for textual
editing and revision passes; and for code
generation in the \Lean{} companion of Appendix~\ref{sec:lean}
and in utility scripts used during revision.  All
mathematical content --- the axiom system $\Tprx$, the
language $\Lprx$, every definition, theorem, and proof, and
the overall argument structure --- was developed and verified
by the authors.  All model-assisted output (text, code, and
bibliographic suggestions) was reviewed, edited, and verified
by the authors prior to inclusion.  The \Lean{} companion in
particular is machine-checked by the type-checker
independently of any language model.

\appendix

\section{Machine verification of the axiom system}
\label{sec:lean}

A common reaction to any new axiom system is to ask whether the
axioms are mutually consistent.  In informal Austrian economics this
question is rarely raised because the underlying claims are stated
verbally and the inferences from them feel compelling.  Once the
claims are stated as formal axioms, however, consistency becomes a
concrete and answerable question --- and not merely answerable in
principle.  It is answerable by a machine.

Accompanying this paper is a self-contained \Lean{} source file,
\texttt{Praxeology.lean}, which encodes the core praxeological
signature of $\Lprx$ as a \Lean{}
\emph{type class} --- a bundle of the sorts, relations, and
axioms packaged as one definition.\footnote{The file is distributed at the
dedicated public repository
\href{https://github.com/rafkom72/praxeology-lean}{\texttt{github.com/rafkom72/praxeology-lean}}.
No \Lean{} installation is required to inspect it: the file can be
pasted directly into the online \Lean{} playground at
\href{https://live.lean-lang.org}{\texttt{live.lean-lang.org}} and
will compile in roughly thirty seconds.  Local installation of
\Lean{} with the VS~Code extension provides an interactive proof
state for readers who wish to step through the deductions.}
The five sorts (Actor, Action, End, Thing, Time) become \Lean{}
types and all six primitive relations become predicates.  The
axioms are organized into two base-theory classes mirroring the
paper's layers.  The first (the action core) carries the
time-order axioms \axref{T1}--\axref{T4}, the incidence axioms
\axref{P1}--\axref{P5}, and the action choice axiom \axref{C1},
with the exact statements they have in this paper.  The second
adds the valuational primitive $\pref{a}{t}$ with the grounding
axiom \axref{O0} and the order axioms \axref{O1}--\axref{O4}, the
remaining time axioms \axref{T0}, \axref{T5}, \axref{T6},
free-good exclusion \axref{P6}, and the scarcity axiom \axref{S1}
--- the full base theory.  Anyone who
reads \Lean{} syntax can verify, by reading the class definitions
alone, that the encoded axioms are the same axioms the paper
analyzes.

The file then does two things.  First, it proves the paper's
foundational results by formal deduction inside \Lean{}: the
asymmetry of demonstrated preference in its
definitional-clause form (Theorem~\ref{thm:asymm}), the
existence of opportunity cost (Theorem~\ref{thm:opp_cost}), the
chosen-end maximality of Proposition~\ref{prop:chosen_max},
and --- in the (MU)-layer --- the diminishing-marginal-utility
theorem and its structure-preservation corollary in exactly the
forms of Section~\ref{subsec:derived_dmu}.  \Lean{}'s kernel
accepts these
proofs only if they are valid in its underlying type theory; in
practice, every inferential step is checked symbol by symbol.

Second --- and this is the more important guarantee for an axiom
system this novel --- the file constructs concrete
\emph{instances} (actual models) of these classes.  The first is a finite,
three-period structure --- the snapshot $\{0,1,2\}$ of the
$\mathbb{N}$-time Crusoe model of Section~\ref{subsec:crusoe},
taken on its own: the actor
(Crusoe), five actions
(\textsf{Forage}, \textsf{BuildNet}, \textsf{ShoreFish},
\textsf{BuildBoat}, \textsf{DeepSeaFish}), four ends
(\textsf{Subsist}, \textsf{Capital}, \textsf{ShoreCatch},
\textsf{DeepCatch}), six things
(\textsf{Wood}, \textsf{Net}, \textsf{Boat},
\textsf{Plant}, \textsf{Fish}, \textsf{Tuna}), the three time
points ($t_0 < t_1 < t_2$) of the snapshot, and the
corresponding predicates; \Lean{} mechanically verifies every axiom
of the action core on it by exhaustive case
analysis over the finitely many parameter combinations.

The full base theory needs more: axiom \axref{T6} gives every
moment a successor, so the full theory has \emph{no finite models
at all}.  The companion therefore constructs a second instance
with $\mathsf{Times} := \mathbb{N}$ --- the state diagram of
Section~\ref{subsec:crusoe} walked indefinitely, with a
non-greedy history (deep-sea fishing on even days, shore fishing
on odd days, so that \axref{P6} always has its unrealised-employment
witness) and preference orders that co-vary with the history (so
the grounding axiom \axref{O0} holds at every moment; the
preference reversal between odd and even days is the
constancy-free variation of Remark~\ref{rem:constancy},
machine-checked).  When the
type-checker accepts these two base-theory instance blocks, what
it has produced is a \emph{constructive consistency proof of the
full base theory $\Tprx$} --- \axref{T0}--\axref{T6},
\axref{P1}--\axref{P6}, \axref{C1}, \axref{O0}--\axref{O4},
\axref{S1}: a model exists,
and the proof that it is a model has been checked by the kernel.
What remains for the reader is only to confirm that the \Lean{} fields
transcribe the axioms of $\Lprx$ (Section~\ref{sec:language}).

A reader who suspects the axiom system is subtly inconsistent --- so
that a contradiction could be derived from it, after which any
purported theorem would be vacuously valid --- now has the
concrete artefact to inspect and the means to refute the
suspicion.  If the axioms were inconsistent, no model could exist;
\Lean{} would refuse to construct the Crusoe instances.  It does not
refuse.

The (MU)-enrichment is a further class on top of the base theory:
it reuses the base scale of values $\pref{a}{t}$, inherits the order
axioms \axref{O2}--\axref{O4}, and encodes \axref{MU0}--\axref{MU4}
(with \axref{MU4} in its relativized form).  A third instance ---
the two-good water/fish schedule of
Section~\ref{subsec:derived_dmu}, over $\mathsf{Times} := \mathbb{N}$
so that it satisfies \axref{T6} --- is a model of this combined
class, so a \emph{single} structure witnesses the base theory and
the (MU)-enrichment jointly, with the diminishing-marginal-utility
theorem and its structure-preservation corollary proved in the forms
of Section~\ref{subsec:derived_dmu}.  The companion does \emph{not}
formalize the further enrichment layers --- ownership, exchange,
money, the transition map --- nor the policy-maker construction,
the Mises Inexpressibility theorem, or the undecidability results;
those are directions for future work.

The economic claims of this paper are not made to depend on \Lean{}.
A reader who never opens the file loses nothing, since every
relevant proof is also given in standard mathematical prose.
The companion serves a different purpose: it changes the cost of
checking the framework from \emph{trust the authors} to
\emph{run the type-checker}.  In a debate where the validity of an
axiomatic argument has historically been adjudicated by the
authority of the argument's proponent, that shift in cost is the
methodological contribution we want to highlight.  We hope it is
useful both for skeptical readers --- who can satisfy themselves
that the axioms are consistent without taking our word for it ---
and for sympathetic readers building on the framework, who now have
a foundation that has been checked symbol by symbol.

A small \texttt{mypy} companion ships alongside the \Lean{} file and
illustrates the underlying logical principle for readers more
familiar with statically typed programming languages: under the
Curry--Howard correspondence, type errors and proof errors are the
same kind of object.  A correct proof type-checks, while a
would-be proof of a false claim does not --- the same mechanism by
which \Lean{} validates the praxeological theorems above, in a
vocabulary some readers may know better.

\section{Proofs of Theorem \ref{thm:DMU} and Corollary \ref{cor:dmu_structure}}
\label{app:dmu_proof}

This appendix supplies the full proofs of Theorem~\ref{thm:DMU}
and Corollary~\ref{cor:dmu_structure} stated in
Section~\ref{subsec:derived_dmu}.  Both theorems are
machine-verified in \Lean{} (Appendix~\ref{sec:lean}, as \texttt{DMU}
and \texttt{DMU\_structure}); the proofs given here are the
human-readable counterparts of those checked deductions --- the same
step-by-step reasoning, on the same hypotheses and axioms, that the
type-checker certifies --- set out for the convenience of a reader who
would rather follow the argument in prose than in \Lean{}.
Both follow the Hilbert-style
format used throughout the paper: formal steps interspersed with
verbal commentary in italics.  The reduced-served predicate
$\Served^{\setminus y}$ is given by
Definition~\ref{def:reduced_served_set}.

\begin{proof}[Proof of Theorem~\ref{thm:DMU} (Diminishing marginal utility)]
Fix $a, t, g, y, E$ as in the theorem statement.  Its two
hypotheses unfold into three working clauses.  The first,
``the marginal unit $y$ exists'' (Definition~\ref{def:marginal_end}),
is the conjunction of (i) uniqueness of $y$ on $E$ --- to no unit
other than $y$ is $E$ allotted --- and (ii) marginality of $E$ in
the full served set.  The second, ``$\Served^{\setminus y} \neq
\emptyset$'', is (iii) non-emptiness of the reduced supply (some end
other than $E$ is served by $g$).  We argue from (i)--(iii).

\medskip

\noindent\textbf{Step 1 (The reduced served set sits above $E$).}
By marginality (ii), every $F$ with $\Served(a,t,g,F)$ and
$F \neq E$ satisfies $F \pref{a}{t} E$.  By uniqueness (i), any
$F$ with $\Served^{\setminus y}(a,t,g,F)$ also has
$\Served(a,t,g,F)$ (just drop the $x \neq y$ restriction); and
$E$ itself fails $\Served^{\setminus y}$ (no unit other than
$y$ serves $E$).  So every $F$ with
$\Served^{\setminus y}(a,t,g,F)$ satisfies $F \neq E$, and
hence by marginality $F \pref{a}{t} E$.

This is the central observation: removing $y$'s cell from the
schedule strips the bottom of the served set (since $E$ was
the marginal end and uniquely served by $y$), leaving exactly
the upper part.  Every end remaining in the reduced served set
is above $E$ by the marginality of $E$.

\medskip

\noindent\textbf{Step 2 (The marginal $E^*$ of the reduced
served set exists).}
By non-emptiness (iii), pick $F_0 \neq E$ with
$\Served(a,t,g,F_0)$, witnessed by some unit $x_0$ with
$\UnitOf(x_0,g) \land \Allot(a,t,x_0,F_0)$.  If $x_0 = y$,
then $\Allot(a,t,y,F_0)$ together with the theorem premise
$\Allot(a,t,y,E)$ would force $F_0 = E$ by \axref{MU0}
functionality --- contradiction.  So $x_0 \neq y$, hence
$\Served^{\setminus y}(a,t,g,F_0)$: the reduced served set is
non-empty.  By Lemma~\ref{lem:served_choice_relevant} it is a
subset of $\EndAt(a,t)$, which is finite
(Remark~\ref{rem:finiteness}) and strictly totally ordered by
$\pref{a}{t}$; a non-empty finite subset of a strict total
order has a unique minimal element, which is the reduced
marginal end $E^*$ of Definition~\ref{def:marginal_end}.

The existence of the previously-marginal end is not an extra
assumption: it is manufactured from non-emptiness by the
functionality of the schedule.  This is where \axref{MU0}
enters the theorem's deduction directly.

\medskip

\noindent\textbf{Step 3 (DMU conclusion).}
$E^*$ lies in the reduced served set, so by Step~1,
$E^* \pref{a}{t} E$.

Step~2 supplies the marginal $E^*$; Step~1 places it strictly
above $E$.  This is the classical DMU statement.  The result is
purely ordinal: no cardinal utility function is invoked.
\end{proof}

\begin{proof}[Proof of Corollary~\ref{cor:dmu_structure} (Structure preservation)]
Take any $F' \in \Served^{\setminus y}$ and any choice-relevant
$G$ with $\Serviceable(a,t,g,G)$ and $G \pref{a}{t} F'$.  By
\axref{MU4} in the full served set,
$\Served(a,t,g,G)$ holds, so some unit $z$ of $g$ is allotted
to $G$.  If $z = y$, then $\Allot(a,t,y,G)$ combined with
$\Allot(a,t,y,E)$ forces $G = E$ by \axref{MU0} functionality;
but by Step~1 of the theorem proof, $F' \pref{a}{t} E$, so
$G \pref{a}{t} F' \pref{a}{t} E$, contradicting $G = E$ under
\axref{O3} and \axref{O4}.  Hence $z \neq y$, and so $G \in
\Served^{\setminus y}$.  This establishes \axref{MU4} for the
reduced allotment: the reduced served set is also a top-segment
of the $g$-serviceable choice-relevant ends under $\pref{a}{t}$.

Removing the bottom of a top-segment yields a smaller top-segment.
\axref{MU4} in the full served set guarantees the top-segment shape;
\axref{MU0} functionality applied to $y$ ensures that $y$'s
removal only strips $E$ (not any higher end); marginality of $E$
ensures the strip happens at the bottom rather than the middle.
This is the formal counterpart of Mises's implicit claim that
the $n$-to-$n-1$ transition is structurally clean: no
Swiss-cheese appears in the reduced view.
\end{proof}

\section{A first-order logic primer}\label{app:fol}

This appendix collects the logical vocabulary the paper uses and
explains how the apparatus relates to ordinary set theory.  None
of it is needed to follow the economic argument; it is a reference
for readers who want the terms spelled out, and it makes precise
two places where the logic does real work --- the consistency
claim of Appendix~\ref{sec:lean} and the finiteness convention of
Remark~\ref{rem:finiteness}.

\subsection*{The framework in brief}

A \emph{many-sorted first-order language} (or \emph{signature})
fixes a finite list of \emph{sorts} and a list of \emph{primitive
relation symbols}, each with a fixed number and kind of argument.
For $\Lprx$ the sorts are the five of Section~\ref{sec:language}
($\mathsf{Actors}$, $\mathsf{Actions}$, $\mathsf{Ends}$,
$\mathsf{Things}$, $\mathsf{Times}$) and the primitive relations
are the six $\Acts$, $\Avail$, $\EndOf$, $\Use$, $\succ$, $<$.
The language has no function symbols, so its \emph{terms} are just
variables.  \emph{Atomic formulas} are a primitive relation or an
equality applied to terms (e.g.\ $\Acts(a,\alpha,t)$ or $x = y$);
\emph{formulas} are built from these with the connectives
$\land, \lor, \lnot, \rightarrow$ and the quantifiers
$\forall, \exists$ ranging over a sort; a \emph{sentence} is a
formula with no free variables.  A \emph{theory} is a set of
sentences, its \emph{axioms}; $\Tprx$ is one such theory.

On the semantic side --- the structures, and what it is for a
formula to hold in one --- the apparatus is, at bottom, elementary
set theory, and translates back to it with nothing added and
nothing lost.  A \emph{sort} is
a set; a \emph{primitive relation} is a subset of a Cartesian
product of those sets (for instance $\Acts$ is a set of triples
$(a,\alpha,t) \in \mathsf{Actors} \times \mathsf{Actions} \times
\mathsf{Times}$); a \emph{structure} (or \emph{interpretation}) is
a choice of a nonempty set for each sort together with such a
subset for each primitive relation; an \emph{axiom} is a
constraint the sets and subsets must meet; and
``$S \models \varphi$'' (read ``$S$ \emph{satisfies} $\varphi$'')
means $\varphi$ is true of the sets-and-relations $S$, under
Tarski's recursive definition of truth.  A structure is a
\emph{model} of a theory $T$ when it satisfies every axiom of $T$.
The framework introduces no objects beyond these sets and subsets;
what it adds is a disciplined language for writing the constraints
down and a precise notion of what follows from them.

This dictionary extends to the logical operations.  A formula
$\varphi(x_1,\dots,x_n)$ denotes its \emph{extension}, the subset
$\{\bar a : \varphi(\bar a)\}$ of the relevant product of domains,
and on extensions the connectives are the Boolean set operations:
conjunction is intersection, disjunction is union, negation is
complement, $\top$ the whole product and $\bot$ the empty set.  The
quantifiers are the one step beyond, and they are not Boolean.
Picture a relation $R \subseteq X \times Y$ as a region in the
plane.  Then $\exists y\,R(x,y)$ --- ``\emph{some} $y$ works for
this $x$'' --- is the \emph{shadow} the region casts on the
$X$-axis: its
\emph{projection},\footnote{The algebraic study of this view ---
Boolean algebra enriched with projection (the quantifier operation)
and diagonal (equality) operators --- is \emph{cylindric algebra}
(Tarski) and \emph{polyadic algebra} (Halmos).} keeping each $x$
for which at least one point $(x,y)$ lies in $R$ and forgetting
which $y$.  Dually, $\forall y\,R(x,y)$ --- ``\emph{every} $y$
works'' --- keeps an $x$ only when the whole vertical line above it
lies in $R$.  Projection cannot be built from $\cap, \cup$, and
complement; that is precisely what lifts first-order logic above
Boolean set algebra.  A structure is
thus a family of subsets of the products of its domains, closed
under intersection, union, complement, and projection --- together
with the diagonals $\{(a,a)\}$ that interpret equality --- and the
axioms of $\Tprx$ are constraints on that family.  Projection is
also where the metatheory re-enters: $\exists$ over an infinite
domain need not attain a witness, which is exactly the point behind
the finiteness convention of Remark~\ref{rem:finiteness} (a least
served end exists only because the menu is finite).

That last notion takes two forms, which first-order logic proves
equivalent.  \emph{Logical consequence}, $T \models \varphi$,
means every model of $T$ satisfies $\varphi$; \emph{provability},
$T \vdash \varphi$, means $\varphi$ has a finite formal proof from
$T$ in a standard deductive calculus.  Soundness
($\vdash \Rightarrow \models$) and G\"odel's completeness theorem
($\models \Rightarrow \vdash$) together give
$T \models \varphi \iff T \vdash \varphi$: ``true in every model''
and ``has a finite, mechanically checkable proof'' coincide.  This
equivalence is what the \Lean{} companion of Appendix~\ref{sec:lean}
trades on.  Two further facts are used below: a theory is
\emph{consistent} when it has a model (equivalently, proves no
contradiction), and \emph{compactness} says that if every finite
subset of a theory has a model then the whole theory does.

\subsection*{First-order logic and set theory}

The semantics just sketched is set-theoretic through and through.
It is tempting to conclude that the apparatus as a whole ``is just
set theory'' --- but that over-reaches the semantic picture, and
the relationship is more layered.  The layers matter at the two
points flagged above; three are worth separating.

\begin{enumerate}[label=(\arabic*),leftmargin=2.2em,itemsep=2pt]
\item \emph{Syntax and proof.}  The formulas of $\Lprx$ are finite
strings, and a proof is a finite tree of inference steps ---
\emph{finitary} objects, codeable as natural numbers with the
syntactic operations (``is a formula,'' ``is a proof,''
substitution) primitive recursive.  Reasoning about them needs only
a weak, finitary metatheory, not the infinitary set theory the
semantics of~(3) rests on.\footnote{If one takes set theory as a
universal background, finite strings are of course coded as sets as
well (a string is a function from a finite ordinal into the
alphabet); the point is that the syntax does not \emph{require} that
infinitary apparatus --- Hilbert's finitist standpoint --- which is
what makes proof-checking a decidable, terminating task, whereas
truth in all models is not even semidecidable in general.}  Checking
a proof is a finite, mechanical task; this is exactly what makes the
machine verification of Appendix~\ref{sec:lean} possible.
\item \emph{The object theory.}  $\Tprx$ is a particular set of
$\Lprx$-sentences --- the constraints this paper writes down and
reasons within.
\item \emph{The metatheory.}  To say what a \emph{model} is, and to
prove results \emph{about} $\Tprx$ (that it is consistent, that no
theory captures finiteness, and so on), one reasons in a background
theory.  The standard choice is set theory: models are
set-theoretic structures and ``$S \models \varphi$'' is defined
set-theoretically.  This is the sense in which set theory sits
``underneath'' the enterprise.
\end{enumerate}

Two clarifications keep the layers straight.  First, \emph{set
theory is itself a first-order theory}: \textsc{zfc}\footnote{Zermelo-Fraenkel set theory with the axiom of choice} has one
signature (the membership relation $\in$) and a list of first-order
axioms.  First-order logic is thus the general framework and set
theory one theory expressed in it --- not the reverse; ``first-order
logic $=$ set theory'' inverts the containment.  Second, \emph{the
metatheory is a choice, not a fixture}: Appendix~\ref{sec:lean}
proves $\Tprx$ consistent by building a model as an object of
\Lean{}'s dependent type theory, not as a set.  There the sorts
become types, and ``is a model'' is settled by the type-checker.  It is
the same axiom system either way; only the background framework
used to check it has changed, and \Lean{}'s framework involves no sets
at all.  The economic conclusions of this paper come out the same
whichever framework a reader adopts.

The finiteness convention of Remark~\ref{rem:finiteness} sits in
the third layer.  We assume each actor's choice menu is finite ---
a condition on the \emph{models}, not an axiom of $\Tprx$.  It
cannot be an axiom: by compactness no first-order theory has
exactly the finite-menu structures as its models (adjoin constant
symbols forcing an arbitrarily large menu; every finite subset of
the enlarged theory is satisfiable, so by compactness the whole has
a model, one with an infinite menu).  This is the ordinary
situation, not a defect peculiar to praxeology: ``let $G$ be a
\emph{finite} group'' likewise restricts the models from outside
and is no axiom of group theory, which has infinite models too.
Finiteness is accordingly carried as an explicit hypothesis in the
results that use it (Definition~\ref{def:marginal_end},
Lemma~\ref{lem:served_choice_relevant}, Theorem~\ref{thm:DMU}),
exactly as one writes ``let $G$ be finite.''

\subsection*{Glossary}

\begin{description}[leftmargin=1.4em, style=nextline,
  font=\normalfont\bfseries, itemsep=2pt]
\item[Signature (language)] the list of sorts and primitive
  relation symbols with their arities; here $\Lprx$.
\item[Sort] a basic kind of object; in a structure, a set.  $\Lprx$
  has five: $\mathsf{Actors}$, $\mathsf{Actions}$,
  $\mathsf{Ends}$, $\mathsf{Things}$, $\mathsf{Times}$.
\item[Primitive relation] an undefined relation symbol; in a
  structure, a subset of a product of domains.
\item[Term] an expression denoting an object; in $\Lprx$, having no
  function symbols, simply a variable.
\item[Atomic formula] a primitive relation or an equality applied
  to terms.
\item[Formula] an expression built from atomic formulas with
  $\land, \lor, \lnot, \rightarrow, \forall, \exists$.
\item[Sentence] a formula with no free variables --- one with a
  definite truth value in each structure.
\item[Free / bound variable] a variable lying outside / inside the
  scope of a quantifier that binds it.
\item[Theory] a set of sentences, taken as axioms; e.g.\ $\Tprx$.
\item[Axiom] a sentence posited as given.
\item[Structure (interpretation)] a nonempty set (domain) for each
  sort together with a relation for each primitive symbol.
\item[Model (of $T$)] a structure satisfying every axiom of $T$,
  written $S \models T$.
\item[Satisfaction ($S \models \varphi$)] $\varphi$ is true in the
  structure $S$.
\item[Logical consequence ($T \models \varphi$)] every model of $T$
  satisfies $\varphi$.
\item[Provability ($T \vdash \varphi$)] $\varphi$ has a finite
  formal proof from $T$.
\item[Soundness] everything provable is true in every model
  ($\vdash \Rightarrow \models$).
\item[Completeness (G\"odel)] everything true in every model is
  provable ($\models \Rightarrow \vdash$); for first-order logic the
  two notions coincide.
\item[Consistency] $T$ proves no contradiction --- equivalently,
  $T$ has a model (by the completeness theorem).
\item[Compactness] if every finite subset of $T$ has a model, then
  $T$ has a model.
\item[Categoricity] all models of $T$ are isomorphic; a first-order
  theory with an infinite model is never categorical in this absolute
  sense --- by the L\"owenheim--Skolem theorems it has models in
  different infinite cardinalities, which cannot be isomorphic ---
  though it may still be \emph{$\kappa$-categorical} (all its models
  of some fixed infinite cardinality $\kappa$ isomorphic).  Hence
  $\Tprx$ is not categorical (see the footnote to \axref{T6}).
\item[First- vs.\ second-order] quantifiers ranging over elements
  of the domains (first-order) versus over relations or subsets of
  them (second-order); finiteness is second-order, not first-order.
\item[Metatheory] the background theory in which structures, truth,
  and results \emph{about} a theory are framed --- set theory by
  default, \Lean{}'s dependent type theory in Appendix~\ref{sec:lean}.
\end{description}


\end{document}